%% file: det.tex
\documentclass[11pt]{article}

\def\version{1}%long

\usepackage[T1]{fontenc}
\usepackage{fullpage}
\usepackage{amsthm}

\usepackage{amssymb}
\usepackage{amsmath}
\usepackage{graphicx,color}
\usepackage{amsmath}

\usepackage{enumitem}
\usepackage{subfigure}
\usepackage{epsfig}
\usepackage{minibox}

\usepackage{algorithm}
\usepackage{algpseudocode}

\newtheorem{theorem}{Theorem}
\newtheorem{lemma}[theorem]{Lemma}
\newtheorem{corollary}[theorem]{Corollary}
\newtheorem{definition}[theorem]{Definition}
\newtheorem{claim}[theorem]{Claim}

\DeclareMathOperator{\sgn}{sgn}

\algnewcommand\comment[1]{\hfill$\triangleright$ #1}

\newenvironment{itemize2}
{\begin{itemize}\setlength{\itemsep}{-0.05cm}\vspace{-0.1cm}}
{\vspace{-0.05cm}\end{itemize}}

\def\f#1{\lfloor #1 \rfloor}
\def\case #1{{\bf Case}\ {\it #1:}}

\def\twofigs #1{\hbox to \textwidth{#1}}

\def\yclaim #1 {\medskip\noindent{\bf Claim #1: }}
\long\def\xclaim #1 {\medskip\noindent{\bf Claim } {\it #1}\medskip}
\def\claim #1 #2 {\medskip\noindent{\bf Claim #1.} {\em #2}\medskip}
\newcommand{\ecproof}{\hfill $\diamondsuit$\\}
\newcommand{\bproof}{\noindent{\bf Proof: }}
\newcommand{\xecproof}{\ $\diamondsuit$}

\def\remark #1{\noindent{\bf Remark:} #1\\}

\input{macros}
\input{pmacros}

\input{prelude}
\begin{document}

\title{Algebraic Algorithms for {$b$}-Matching,
Shortest
Undirected Paths,
and {$f$}-Factors\thanks{Research was supported by  the ERC StG project
PAAl no. 259515.}}
\author{%
 Harold N.~Gabow
 \thanks{Department of Computer Science, University of Colorado at Boulder,
 Boulder, Colorado 80309-0430, USA. e-mail: {\tt hal@cs.colorado.edu}
 }
\and
Piotr Sankowski
\thanks{Institute of Informatics, University of Warsaw, Banacha 2,
02-097, Warsaw, Poland, and Department of Computer and
System Science, Sapienza University of Rome. email: {\tt sank@mimuw.edu.pl}
 }
}

\maketitle
\begin{abstract}
Let $G=(V,E)$ be a graph with $f:V\to \mathbb Z_+$ a function assigning
degree bounds to vertices. We present the first efficient algebraic
algorithm
to find an $f$-factor. The time is
$O(f(V)^{\omega})$.
More generally for graphs with integral edge weights of maximum absolute
value $W$ we find
a maximum weight $f$-factor in
time $\tilde{O}(Wf(V)^{\omega})$.
(The algorithms are randomized,
correct
with high probability and Las Vegas;
the  time bound is worst-case.)
We also present three specializations of
these algorithms:
For
maximum weight perfect $f$-matching the algorithm is considerably simpler
(and almost identical to its special case of ordinary weighted matching).
For the single-source shortest-path problem in undirected graphs with conservative edge weights,
we present a generalization of the shortest-path tree, and we compute it in
$\tilde{O}(Wn^{\omega})$ time. For bipartite graphs, we
improve the known complexity bounds
for vertex capacitated max-flow and min-cost max-flow on
a subclass of graphs.
\end{abstract}
\input{intro}

\input{definitions}

\input{preliminaries}

\input{bipartite-determinant}

\input{bipartite-algorithms}
\input{bipartite-weights}

\input{bipartite-max-flow}

\input bduality
\input shtstpth

\input background
\input fred

\input{spalg}
\input shtstpcomb

\input{determinant}
\input{algorithms}
\input{weights}

\input{conclusions}

\bibliographystyle{abbrv}
\bibliography{weighted}

\appendix
\input{appendix-allowed}

\end{document}

%% file: macros.tex
\def\o{\overline} 
\def\hi{\advance \parindent by 20pt}
\def\H#1{\widehat{#1}}

\def\mathy #1{{\ifmmode {#1} \else{$ #1 $}\fi}}
\def\b #1.{\mathy {\overline {#1}}}
\def\sb #1.{\mathy {\bar {#1}}}
\def\set{}
\def\I.{\mathy {\cal I}}

\newcommand{\ai}[1]{\mbox{\it (i)}}
\newcommand{\ii}[1]{\mbox{\it (ii)}}
\newcommand{\iii}[1]{\mbox{\it (iii)}}
\newcommand{\iv}[1]{\mbox{\it (iv)}}
\newcommand{\av}[1]{\mbox{\it (v)}}
\newcommand{\vi}[1]{\mbox{\it (vi)}}
\newcommand{\vii}[1]{\mbox{\it (vii)}}
\newcommand{\viii}[1]{\mbox{\it (viii)}}
\newcommand{\ix}[1]{\mbox{\it (ix)}}
\newcommand{\x}[1]{\mbox{\it (x)}}
\newcommand{\axi}[1]{\mbox{\it (xi)}}
\newcommand{\xii}[1]{\mbox{\it (xii)}}
\newcommand{\xiii}[1]{\mbox{\it (xiii)}}

\def\T.{\mathy {\cal T}}
\def\B.{\mathy {\cal B}}
\def\C.{\mathy {\cal C}}
\def\F.{\mathy {\cal F}}
\def\S.{\mathy {\cal S}}
\def\M.{\mathy {\cal M}}
\def\W.{\mathy{\cal W}}
\def\Z.{\mathy {\cal Z}}
\def\fs #1.{\mathy {f\hskip-4pt\updownarrow\hskip-4pt {\scriptstyle{#1}}}}
\def\Fs #1.{\mathy {F\hskip-4pt\updownarrow\hskip-4pt {\scriptstyle{#1}}}}
\def\bfs #1.{\mathy {\bar f\hskip-4pt\updownarrow\hskip-4pt {\scriptstyle{#1}}}}
\def\bFs #1.{\mathy {\bar F\hskip-4pt\updownarrow\hskip-4pt {\scriptstyle{#1}}}}

\def\adj{\operatorname{adj}}

%% file: pmacros.tex
\def\today{\ifcase\month\or
January\or February\or March\or April\or May\or June\or
July\or August\or September\or October\or November\or December\fi
\ \number\day, \number\year}
\def\date#1.#2.{\ifcase#1\or
January\or February\or March\or April\or May\or June\or
July\or August\or September\or October\or November\or December\fi
\ #2, \number\year}
\def\ydate#1.#2.#3.{\ifcase#1\or
January\or February\or March\or April\or May\or June\or
July\or August\or September\or October\or November\or December\fi
\ #2, 199#3}
\def\nydate#1.#2.{\ifcase#1\or
January\or February\or March\or April\or May\or June\or
July\or August\or September\or October\or November\or December\fi
\ #2}
\def\doublespace{\multiply\baselineskip by3\divide\baselineskip by2%
                 \def\doublespace{}}%don't doublespace again!
\def\bigdoublespace{\multiply\baselineskip by2%
                 \def\bigdoublespace{}}%don't doublespace again!
\def\imp{\ifmmode {\ \Longrightarrow \ }\else{$\ \Longrightarrow \ $}\fi}
\def\rimp{\ifmmode {\ \Longleftarrow \ }\else{$\ \Longleftarrow \ $}\fi}
\def\ximp{\ifmmode {\Longrightarrow\ }\else{$\Longrightarrow\ $}\fi}
\def\xrimp{\ifmmode {\Longleftarrow\ }\else{$\Longleftarrow\ $}\fi}
\def\iff{\ifmmode {\ \Longleftrightarrow \ }\else{$\ \Longleftrightarrow \ $}\fi}
\def\xiff{\ifmmode {\Longleftrightarrow\ }\else{$\Longleftrightarrow\ $}\fi}

\def\endskip{\medskip}%\smallskip
\def\rqed{\hfill\hbox to 24 pt{\vrule width4pt depth-1pt
height7pt\hfil}\bigskip}
\def\rqedn{\hfill\hbox to 24 pt{\vrule width4pt depth-1pt height7pt\hfil}}
\def\log{\ifmmode \,{ \rm log}\,\else{\it log }\fi}
\def\con {\subseteq}
\def\pcon{\subset}
\def\firstnumstp#1 {\bigskip \noindent{\it Step} #1.\newquad}
\def\numstp#1 {\endskip\noindent{\it Step} #1.\newquad}
\def\newquad{\hskip1ex}
\def\stp#1.{\endskip
\noindent{\it #1 Step.}\newquad}
\def\firststp#1.{\bigskip
\penalty-1000
\noindent{\it #1 Step.}\newquad}
\def\cas#1 {\smallskip\noindent{\bf Case} #1.\ } %\quad}
\def\nsec#1{\penalty-2000%
\noindent{\bf #1\hfill\break}%\hskip\parindent
\hbox to \parindent{\hfill}\ignorespaces}
\long\def\res #1. #2{\bigskip
\penalty-1000
\noindent {\bf #1.}\newquad%
#2 \bigskip}
\long\def\nres #1. #2{\bigskip
\noindent {\bf #1.}\newquad%
#2}
\def\pf{\noindent {\bf Proof.}\newquad}
\def\cont{\ifmmode\star\else$\star$\fi}
\def\+{\tabalign} %ordinarily \+ is \outer

\def\i{($i$) } \def\xi{($i$)}
\def\ii{($ii$) } \def\xii{($ii$)}
\def\iii{($iii$) } \def\xiii{($iii$)}
\def\iv{($iv$) }

\def\hi{\hskip20pt\i} \def\hii{\hskip20pt\ii}

\def\({{\rm(}} \def\){{\rm)}}
\def\c#1{\lceil {#1} \rceil}
\def\f#1{\lfloor {#1} \rfloor}
\def\x{\iffalse}
\def\b{\bigskip}
\def\set #1#2{\{ #1:#2 \}}

\def\hi{\advance\parindent by 20pt}

%% file: prelude.tex
\input pmacros
\def\o{\overline} 
\def\hi{\advance \parindent by 20pt}
\def\H#1{\widehat{#1}}

\def\mathy #1{{\ifmmode {#1} \else{$ #1 $}\fi}}
\def\b #1.{\mathy {\overline {#1}}}
\def\sb #1.{\mathy {\bar {#1}}}
\def\I.{\mathy {\cal I}}
\def\B.{\mathy {\cal B}}
\def\C.{\mathy {\cal C}}
\def\F.{\mathy {\cal F}}
\def\M.{\mathy {\cal M}}
\def\N.{\mathy {\cal N}}
\def\P.{\mathy {\cal P}}
\def\R.{\mathy {\cal R}}
\def\S.{\mathy {\cal S}}
\def\T.{\mathy {\cal T}}
\def\V.{\mathy {\cal V}}
\def\W.{\mathy{\cal W}}
\def\Z.{\mathy {\cal Z}}
\def\fs #1.{\mathy {f\hskip-4pt\updownarrow\hskip-4pt {\scriptstyle{#1}}}}
\def\Fs #1.{\mathy {F\hskip-4pt\updownarrow\hskip-4pt {\scriptstyle{#1}}}}
\def\bfs #1.{\mathy {\bar f\hskip-4pt\updownarrow\hskip-4pt {\scriptstyle{#1}}}}
\def\bFs #1.{\mathy {\bar F\hskip-4pt\updownarrow\hskip-4pt {\scriptstyle{#1}}}}

%% file: intro.tex
\section{Introduction}

$b$-matching and $f$-factors are basic combinatorial notions that generalize non-bipartite matching,
min-cost network flow, and others.
This paper presents the first efficient algebraic algorithms for both weighted and unweighted $b$-matchings
and $f$-factors. Our algorithms for this broad class of problems are the most efficient algorithms known for
a subclass of instances (graphs of high density, low degree constraints
and low edge weights). We also discuss single-source all-sinks shortest paths in conservative undirected graphs.
(There is no known reduction to directed graphs.) We prove the existence of a simple  shortest-path
"tree" for this setting. We also  give efficient algorithms
-- combinatoric for sparse graphs and algebraic for dense -- to construct it.

We must first   define  $b$-matching and $f$-factors. The literature is
inconsistent but in essence we follow the classification of Schrijver \cite{Schrijver}.
For an undirected multigraph $G=(V,E)$
with a function
$f:V\to \mathbb Z_+$, an
{\em $f$-factor} is a subset of edges
wherein each vertex $v\in V$ has degree exactly $f(v)$.
For an undirected graph
$G=(V,E)$
with a function
$b:V\to \mathbb Z_+$,
a (perfect) {\em $b$-matching}
is a function $x:E \to \mathbb{Z}_+$ such that each $v\in V$ has $\sum_{w: vw\in E} x(vw) = b(v)$.
The fact that $b$-matchings have an unlimited number of copies of each edge
makes them decidedly simpler. For instance
$b$-matchings have essentially the same blossom structure
(and linear programming dual variables)
as
ordinary matching \cite[Ch.31]{Schrijver}. Similarly our algorithm
for weighted $b$-matching is almost identical to its specialization to ordinary matching ($b\equiv 1$).
In contrast the blossoms and dual variables for weighted $f$-factors are more involved \cite[Ch.32]{Schrijver}
 and our algorithm
is more intricate.
Thus our terminology reflects the difference in complexity of the two notions.%
\footnote{Another version of $b$-matching considers $b(v)$ as an upper bound on the desired degree of $v$.
This easily reduces to weighted perfect $b$-matching by taking 2 copies of $G$ joined by
zero-weight edges.
On the other hand a {\em capacitated $b$-matching}
is defined by giving an upper bound $u(e)$ to each value $x(e)$.
The simplicity of
the uncapacitated case is lost, and we are back to $f$-factors.}

The paper begins with unweighted $f$-factors, i.e., we wish to find an $f$-factor
or show none exists. Let $\phi=f(V)$ (or $b(V)$).
We extend the Tutte matrix from matching to $f$-factors,
i.e., we present a $\phi \times \phi $ matrix that is
symbolically nonsingular iff the graph has an $f$-factor.
Such a matrix can be derived by applying the Tutte matrix to an enlarged version of the given graph,
or by specializing Lov\'asz's matrix
for matroid parity \cite{l}. But neither approach is compact enough to achieve our time bounds.%
\footnote{The Tutte matrix becomes too large, $m\times m$. Lov\'asz's matrix
is $\phi \times \phi$ but can involve
integers that are too large,  size $n^n$ or more. Our matrix only involves integers $\pm 1$.}
Then we reuse the elimination framework for maximum cardinality matching, due to
Mucha and Sankowski~\cite{ms04} and Harvey~\cite{Harvey06}. This allows us to find an $f$-factor in $O(\phi^{\omega})$ randomized time.%
\footnote{$O(n^{\omega})$ is the time needed for a straight-line program to multiply two $n\times n$ matrices; the best-known bound on $\omega$ is $<2.3727$~\cite{Williams}.} For dense graphs
and small degree-constraints this improves the best-known time bound of $O(\sqrt{\phi} m)$ \cite{hal83},
although the latter is deterministic.

\iffalse
The development of our algorithms would be impossible
without the recent advance in development of the combinatoric interpretation of dual variables for perfect matchings~\cite{CGSa,G}
as well as for $b$-matchings and $f$-factors~\cite{G}. This interpretation was the most crucial part that led to the first
$\tilde{O}(Wn^{\omega})$ time algorithm for finding minimum weight perfect matchings in graphs~\cite{CGSa}. It allows to relate
the dual problem to the perturbed primal problem and opens the way for the efficient computation of the dual variables
(for definition of perturbed problems we refer to Section~\ref{BackgroundSec}).
Here, we show how to turn this interpretation into efficient solutions for $f$-matching and $f$-factor problems
obtaining $O(W\phi^{\omega})$ randomized time algorithms. This result
improves upon previously known algorithms for dense graphs -- please see Table~1 and Table~2.
Our algorithms can be easily turned into Las Vegas algorithms, since we can check the optimality by using
dual variables.
%Our weighted algorithms imply $O(\phi^{\omega})$ time algorithm for computing maximum size $b$-matching as well.
\fi

\begin{center}
    {
    \noindent
      \begin{tabular}{|p{5cm}|p{4.0cm}|} \hline
        Time & Author \\ \hline
        $O(n^2 B)$  & Pulleyblank (1973) \cite{Pulleyblank}\\ \hline
        $O(n^2 m \log B)$ & Marsh (1979) \cite{Marsh}\\ \hline
        $O(m^2 \log n \log B)$ & Gabow (1983) \cite{hal83}\\ \hline
        $O(n^2 m + n \log B(m+n\log n))$ & Anstee (1987) \cite{anstee}\\ \hline
        $O(n^2 \log n(m+n\log n))$ & Anstee (1987) \cite{anstee}\\ \hline
        $\tilde{O}(W\phi^{\omega})$ & this paper\\ \hline
      \end{tabular}

      \vspace{0.2cm}
      {\bf Table 1:} Time bounds for maximum $b$-matching. $B$ denotes $\max_v b(v)$.
    }
\end{center}

\begin{center}
    {
    \noindent
      \begin{tabular}{|p{3.1cm}|p{4.8cm}|} \hline
        Complexity & Author \\ \hline
        $O(\phi n^3)$ & Urquhart (1965) \cite{urquhart} \\ \hline
        %$O(\phi m \log n)$ & Gabow (1983) \cite{hal83}\\ \hline
        %$O(\phi n^2)$ & Gabow (1983) \cite{hal83}\\ \hline
%$O(\phi(m+n \log n))$ & Gabow (1990) \cite{gabow-90}  \\ \hline
$O(\phi(m+n \log n))$ &
Gabow (1983$+$1990) \cite{hal83, gabow-90} \\ \hline
        $\tilde{O}(W\phi^{\omega})$ & this paper\\ \hline
      \end{tabular}

      \vspace{0.2cm}
      {{\bf Table 2:} Time bounds for maximum weight $f$-factors on simple graphs.}
}
\end{center}

We turn to the more difficult weighted version of the problem. Here every edge has a numeric
weight; for complexity results we assume  weights are  integers of magnitude
$\le W$.
We seek a {\em maximum
$f$-factor}, i.e., an $f$-factor with the greatest possible total weight.
Efficient algebraic algorithms
have been given for maximum matching ($f\equiv 1$) in time $\tilde{O}(W n^\omega)$,
first for bipartite graphs \cite{Sankowski09} and recently for general graphs \cite{CGSa}.\footnote{The $\tilde{O}$ notation ignores factors of
$\log (n\phi W)$.}% and $\log W$.}

The usual approach to generalized matching problems is by problem reduction.
For instance
in \cite{Schrijver}, Ch.31
proves the properties of the $b$-matching linear program and polytope
by reducing to ordinary matching via vertex splitting;
then Ch.32 reduces $f$-factors (capacitated $b$-matching) to $b$-matching.
Efficient algorithms also use vertex splitting \cite{hal83}
or
reduction
to  the bipartite case (plus by further processing) \cite{anstee}.
But reductions may obscure some structure.
To avoid this our algorithms use a direct approach, and we get the following rewards.
For $b$-matching, as mentioned, the similarity of blossoms to
ordinary matching blossoms leads to an algorithm that is no more involved than
ordinary matching.
For undirected shortest paths we get a simple definition of a generalized shortest-path tree.
(Again such a definition may have been overlooked due to reliance on reductions, see below.)
For $f$-factors we get a detailed understanding of the more complicated versions of
the structures that first emerge in $b$-matchings
(2-edge connected components giving the cyclic part of blossoms -- see Section \ref{VBSec})
and in shortest paths (bridges giving the incident edges of blossoms --
these correspond to
the (ungeneralized) shortest-path tree -- see Section \ref{IBSec}).

All three of our non-bipartite algorithms are implementations of
the "shrinking procedure" given in \cite{G} (a variant is the basis of the weighted matching
algorithm of \cite{CGSa}).
This procedure gives a direct way to find the optimum blossoms for a weighted $f$-factor --
simply put,
each blossom is (a subgraph of) a maximum weight "$2f$-unifactor" (a type of $2f$-factor) in
the graph with (the cyclic part of) all heavier blossoms contracted
(see \cite{G} or Section \ref{BackgroundSec}). Note that the classic weighted matching algorithm of Edmonds
\cite{e} finds the optimum blossoms, but only after  forming and discarding
various other blossoms. So this approach does not provide a direct definition of the optimum blossoms.

The first step of our algorithms use our generalized Tutte matrix to find
the optimum dual variables of the vertices.  Then we execute  the shrinking procedure to
get the blossoms, their duals, and a "weighted blossom tree" that gives the structure of the optimum $f$-factor.
(This step is combinatoric. It is based on the detailed structure of
$2f$-unifactors that we derive.)
The last step finds the desired $f$-factor  using
a top-down traversal of
the weighted blossom tree: At each node
we find an $f$-factor of a corresponding graph, using our algorithm for unweighted $f$-factors.
In summary our algorithms (like \cite{CGSa} for ordinary matching) can be viewed as a (combinatoric) reduction of the weighted
$f$-factor problem into two subproblems: finding the optimum dual variables of the vertices,
and finding an unweighted $f$-factor.

To facilitate understanding of the general $f$-factor algorithm we begin by presenting its specialization to
two subcases. First $b$-matching.
The blossoms, and hence the $2b$-unifactors, differ little from ordinary graph
matching. As a result our development for weighted $b$-matching  is essentially  identical
the special case of ordinary matching, in terms of both the underlying combinatorics and the algorithmic details.
When specialized to ordinary matching our algorithm provides a simple alternative to \cite{CGSa}. In fact an advantage
is that our algorithm is Las Vegas -- the dual variables allow us to check if the $b$-matching is truly optimum.
(Our approach to weighted matching/$b$-matching differs from \cite{CGSa} -- at the highest level,
we work with critical graphs while \cite{CGSa} works with perfect graphs.)

Next we discuss shortest paths in undirected graphs with a conservative weight function --
negative edges are allowed but not negative cycles. The obvious reduction to a directed graph
(replace undirected edge $uv$ by directed edges $uv,vu$) introduces negative cycles, and it is
unclear how to handle this problem by the usual shortest-path techniques.

We consider the
single-source all-sinks version of the problem.
Again, this problem is often solved by reduction, first to the single-source single-sink
version and then to perfect matching, using either T-joins
\cite[pp.485--486]{Schrijver}
or vertex-splitting
\cite[p.487]{Schrijver}.
A path can be viewed as a type of 2-factor.
(For instance an $ab$-path is an $f$-factor
if we enlarge $G$ with a loop at every vertex $v\in V$
and set $f(v)=2$ for $v\in V-\{a,b\}$, $f(a)=f(b)=1$.)
This enables us to
solve the all-sinks version directly. Examining the
blossom structure
enables us to
define a generalized shortest-path tree that, like the standard shortest-path tree for directed graphs,  specifies a shortest path to every vertex from a chosen source.
It is a combination of the standard shortest-path tree and
the blossom tree. We give a complete derivation of the existence of this shortest-path structure, as well
as an algebraic algorithm to construct it in time $\tilde{O}(Wn^\omega)$. We also construct the structure
with combinatoric algorithms, in time
$O(n(m+n\log n))$ or $O(\sqrt{n \alpha(m,n)\log n}\ m \log (nW))$. These bounds are all
within logarithmic factors of the best-known bounds for constructing the directed shortest-path tree ~\cite[Ch.~8]{Schrijver}, \cite{YusterZ05,Sankowski05}.

\iffalse
Additionally, we give an improved algorithms that constructs the shortest path structure. We essentially argue that instead of calling O(n) times perfect matching subroutine, a single call to it is enough. This allows us to give  $\tilde{O}(Wn^{\omega})$ time algebraic algorithm. Similarly, we can use algorithm of Gabow~\cite{gabow-90} to obtain an $O(n(m+n\log n))$ time algorithm for computing shortest path structure, and an $O(\sqrt{n \alpha(m,n)\log n}\ m \log (nW))$ time algorithm by using results of Gabow and Tarjan~\cite{gt91}. Hence, our time bounds up to logarithmic factors match the fastest algorithms known for directed case~\cite[Ch.~8]{Schrijver} and~\cite{YusterZ05,Sankowski05}).
\fi

Although the
shortest-path problem is classic, our definition of this structure appears to be new.
Most notably, Seb\"o has characterized the structure of single-source shortest paths in undirected graphs,
first for
graphs with $\pm 1$ edge weights \cite{Sebo} and then extending  to general weights by reduction \cite{sebo2}.  Equation (4.2) of \cite{Sebo} (for $\pm 1$-weights, plus its version achieved by reduction
 for arbitrary weights)
 characterizes the shortest paths from a fixed
source
in terms of  how they enter and leave "level sets"
determined by the distance function. \cite{Sebo} also shows that the distances from the source can be computed using $O(n)$ perfect matching computations. Our structure differs from
 \cite{Sebo,sebo2}: it does not give a necessary and sufficient condition to be a shortest path, but it
gives an exact specification of a specific set
of shortest paths that are simply related to one another (as in the standard shortest-path tree). Note that one
can give an alternative proof of the existence of our structure  by starting from the results of \cite{Sebo,sebo2}.

The general algorithm for maximum $f$-factors is the most difficult part of the paper.
It involves a detailed study of the properties of blossoms. A simple example of how
these blossoms differ from ordinary matching is that the hallmark of Edmonds' blossom
algorithm -- "blossoms shrink" -- is not quite true. In other words for ordinary matching
a blossom can be contracted and it becomes just an ordinary vertex. For $f$-factors we can contract
the "cycle" part of the blossom, but its incident edges remain in the graph and must be treated differently from
ordinary edges (see Section \ref{IBSec}).
Our discussion of shortest paths introduces this difficulty in
the simplest case -- here a blossom has exactly 1 incident edge (as opposed to an arbitrary number).
Even ignoring this issue, another difficulty is that there are three types of edges that behave differently
(see \cite{G}, or Lemmas \ref{UnifactorLemma} and \ref{NUnifactorCor}) and the type of an edge is
unknown to the algorithm! Again the three types are seen to arise naturally in shortest paths.
Our contribution is to develop the combinatoric properties of these edges and blossoms
so the shrinking procedure can be executed efficiently, given only the information provided by the
Tutte matrix in our algebraic algorithm.

While non-bipartite graphs present the greatest technical challenge,
we also achieve some best-known time bounds for
two bipartite problems, maximum network flow
and min-cost network flow.
Bipartite $f$-factors generalize network flow: max-flow
(min-cost max-flow)
is a special case of
unweighted (weighted) bipartite $f$-factors, respectively~e.g.~\cite{gabow-tarjan-89}.
The question of
an efficient algebraic max-flow algorithm has confronted the community for
some time. The only advance is the algorithm of Cheung et.~al.~\cite{cheung}, which checks
whether $d$ units of flow can be sent across a unit-capacity network in $O(d^{\omega-1}m)$ time. We
consider
 networks with integral vertex and edge
capacities bounded by $D$.
We find a max-flow in time
$\tilde{O}((Dn)^{\omega})$ time
and a
min-cost max-flow
in $\tilde{O}(W(Dn)^{\omega})$ time.
%(by~\cite{gabow-tarjan-89} $D$ translates to $f$-values).
The latter algorithm handles convex edge cost functions
(with integral break-points)
in the same
time bound.
The max-flow problem has a rich history (see e.g.~\cite[Chs.~10, 12]{Schrijver})
and our time bounds are the best-known for
dense graphs with moderately high vertex capacities.
Specifically, previous algorithms for vertex-capacitated max-flow in dense networks (i.e., $m=\Theta(n^2)$)
use $O(n^3/\log n)$ time~\cite{Cheriyan90} or $O(n^{8/3}\log D)$ time~\cite{Goldberg98}.
Previous algorithms for dense graph
min-cost max-flow use $O(n^3 \log D )$ time~\cite{edmonds-karp-72} or $O(n^3\log n)$ time~\cite{Orlin88}.
Previous algorithms for
minimum convex-cost max-flow
use $O(D n^3\log D)$ time
(by simple reduction to min-cost max-flow)
or in $O(n^3\log D\log (nW))$ time~\cite{gabow-tarjan-89}.

In summary the novel aspects of this paper are:

$\bullet$
new time bounds for the fundamental problems of $b$-matching, undirected single-source shortest paths, and $f$-factors;

$\bullet$
extension of the Tutte matrix for matching to $f$-factors;

$\bullet$
definition of the shortest-path structure for undirected graphs, plus algebraic and combinatorial
algorithms to construct it;

$\bullet$
an algebraic algorithm for $b$-matching that is no more involved than ordinary matching;

$\bullet$
an algebraic algorithm for $f$-factors based on new combinatorial properties of blossoms,
which are qualitatively different from matching blossoms;

$\bullet$
new time bounds for vertex-capacitated max-flow, min-cost max-flow and minimum convex-cost max-flow on dense graphs.

\begin{table}[t]
{\def\x{{\sf X}}
\def\a{{\sf A}}
\def\c{{\sf C}}
\begin{center}
    {
    \noindent
      \begin{tabular}{|l*{6}{|c}|}%\hline%p{3.1cm}|p{4.8cm}|} \hline
        %Complexity & Author \\ \hline
\multicolumn{7}{r}{\hbox to 1in{\bf Section}\hfill}\\
\cline{2-7}
  \multicolumn{1}{c|}{\quad}            &5&8&9&10&12&13\\ \hline
%{\bf Shortest paths}&&&&&&\\ \hline
\multicolumn{1}{|l|}{\bf Shortest paths}&\multicolumn{6}{l|}{\quad}\\ \hline
\quad find distances& & & &\c && \x  \\ \hline
\quad find forest& &\c & & & &  \\ \hline
\quad extract sp-tree& & &\c & & & \\ \hline
%{\bf \boldmath $b$-matching}\\ \hline
\multicolumn{1}{|l|}{\bf \boldmath $b$-matching}&\multicolumn{6}{l|}{\quad}\\ \hline
\quad find duals& & & & & &\x  \\ \hline
\quad find forest \& extract&\c & & & & &  \\ \hline
\quad unweighted $b$-matching & & & & & \x&  \\ \hline
%{\bf \boldmath $f$-factors}\\ \hline
\multicolumn{1}{|l|}{\bf \boldmath $f$-factors}&\multicolumn{6}{l|}{\quad}\\ \hline
\quad find duals& & & & & &\x  \\ \hline
\quad find forest \& extract& &\c & & & &  \\ \hline
%\quad extract & & &\x & & &  \\ \hline
\quad unweighted $f$-factor& & & & &\x &  \\ \hline
\end{tabular}

      \vspace{0.2cm}
      {{\bf Table 3:}
Sections for each step of the algorithm.
X is an algebraic algorithm, C is combinatoric.}

 {\hskip -1.77in Shortest-path distances can be found algebraically  or
combinatorially.}

}
\end{center}
}
\end{table}

\paragraph*{Organization of the paper}
The next two sections define our terminology and review the algebraic tools that we use.

Section~\ref{section:bipartite} gives an overview of the entire paper, by discussing
the special case of bipartite $f$-factors.
(The reader should bear in mind that the non-bipartite case
is the highlight of this paper. It is technically much more demanding.)
In detail, Section~\ref{section:bipartite}  starts with an $O(\phi^{\omega})$ time algorithm
for unweighted bipartite $f$-factors. This requires our generalized Tutte matrix
(Section~\ref{section:biparite:formulations};
Section~\ref{sec:determinant}
extends this to non-bipartite graphs)
plus the
algorithmic details (Section~\ref{section:biparite:algorithms};
Section~\ref{sec:algorithms} generalizes these details to non-bipartite graphs).
Section~\ref{section:weighted-bipartite}
gives the algorithm for weighted bipartite $f$-factors.
Flows are discussed in Section~\ref{FlowSec}.

As mentioned above the algorithms for general graphs
have three steps:

\bigskip

\i find the weights of factors of perturbed graphs;

\ii use the weights to construct a blossom forest;

\iii use an unweighted algorithm to extract an  optimal solution from the blossom forest.

\bigskip

\noindent
These steps are implemented in different sections, depending on the problem of interest,
as indicated in Table 3. For instance the complete $b$-matching algorithm starts with
the algorithm of Section~\ref{sec:weights} to get the dual variables; then it constructs
a weighted blossom forest and traverses the forest to extract the optimal solution, using
algorithms in
Section~\ref{GeneralbMatchingSec}; the traversal uses the algorithm of Section~\ref{sec:algorithms}
to find various unweighted $b$-matchings.

\iffalse
Specifically for
{$b$}-matching,  \i is in Section~\ref{sec:weights}, \ii and \iii are in Section~\ref{GeneralbMatchingSec},
and
the unweighted algorithm of \iii is in Section~\ref{sec:algorithms}.
For shortest paths,
the algebraic algorithm for \i is in Section~\ref{sec:weights} and
the combinatoric algorithm is in Section~\ref{SPcombinatoric}.
\ii is in Section~\ref{fFactorSec} and
\iii is in Section~\ref{SPAlgSec}.
For {$f$}-factors, \i is in Section~\ref{sec:weights}, \ii and \iii are in Section~\ref{fFactorSec},
with the unweighted algorithm of \iii in Section~\ref{sec:algorithms}.
\fi

The most involved part of the paper is the construction of blossom forest for $f$-factors in Section~\ref{fFactorSec}.
This section is preceded by two sections that introduce our combinatoric ideas in simpler settings:
Section \ref{GeneralbMatchingSec} gives the combinatoric portion of our weighted $b$-matching algorithm. Section \ref{SPSec} proves the existence of our shortest-path structure.
Section~\ref{BackgroundSec} reviews fundamental background material on $f$-factors
(illustrating it by  shortest-paths).
After the combinatoric algorithm for weighted $f$-factors in Section~\ref{fFactorSec},
Section \ref{SPAlgSec} gives the   algorithm to construct the shortest-path structure,
and Section \ref{SPcombinatoric} gives the combinatorial  algorithms for shortest-path weights.
Sections~\ref{sec:determinant}--\ref{sec:algorithms}
show how to find unweighted $f$-factors in general graphs.
Section~\ref{sec:weights} shows how to compute the weights of perturbations of $f$-factors.
Finally Section \ref{ConclusionSec} concludes by posing several new open problems.

Given the length of this submission, we remark that various portions
can be read independently.
Section~\ref{section:bipartite} gives the whole development in the bipartite case.
Section~\ref{GeneralbMatchingSec} gives the  combinatoric part of the
$b$-matching algorithm starting from first principles.
Section~\ref{SPSec} derives the shortest-path structure from first principles, entirely in the language
of shortest paths rather than matching. The related material in Section~\ref{SPAlgSec} is itself
independent, given the definition of the shortest-path structure (Section~\ref{SPDefnSec}).
Alternatively Sections~\ref{GeneralbMatchingSec}--\ref{SPSec} can be skipped to go directly
to the combinatoric part of the $f$-factor algorithm (Sections~\ref{BackgroundSec}--\ref{fFactorSec}).

%% file: definitions.tex
\section{Problem definitions}
\label{sec:definitions}

The symmetric difference  of sets is denoted by $\oplus$, i.e.,
$A\oplus B = (A-B)\cup (B-A)$.
We use a common convention to sum function values:
If $f$ is a  real-valued function on elements and $S$ is a set
of such elements, $f(S)$ denotes $\sum \set {f(v)} {v\in S}$.
Similarly if $z$ is a function on sets of elements then
$z\set {S} {S\in \S.}$ denotes $\sum \set {z(S)} {S\in \S.}$.

Let $G=(V,E)$ be an undirected graph, with vertex set $V = \{1, \ldots, n\}$.
We sometimes write $V(G)$ or $E(G)$ to denote vertices or edges of graph $G$. 
A {\em walk} is a sequence $A=v_0,e_1,v_1,\ldots,e_k,v_k$
for
vertices $v_i$ and edges $e_i=v_{i-1}v_i$.
The notation {\em $v_0v_k$-walk} provides the two endpoints.
The {\em length} of $A$ is $k$, and
the parity of $k$ makes
$A$ {\em even} or {\em odd}.
A {\em trail} is an edge-simple walk.
A {\em circuit} is a trail that
starts and ends at the same vertex.
The vertex-simple analogs are {\em path} and {\em cycle}.

In an undirected multigraph  $G=(V,E)$ each edge $e\in E$ has a positive
multiplicity $\mu(e)$. Each copy of a fixed $e\in E$ is a distinct edge, e.g., a trail
may up to $\mu(e)$ distinct copies of $e$.
For a multiset $S$, $2S$ denotes $S$ with every multiplicity doubled.
Similarly $2G$ denotes the multigraph $(V,2E)$.
If every multiplicity of $S$ is even then $S/2$ denotes $S$ with every
multiplicity halved.

\iffalse
For multigraphs a trail can contain
parallel edges since they are distinct.

In a multigraph a loop is considered to be a simple cycle.
If a path $P$
contains vertices $i,j$ then $P(i,j)$ denotes the subpath from $i$ to $j$.
\fi

For a set of vertices $S\subseteq V$ and a subgraph $H$ of $G$,
$\delta(S,H)$ ($\gamma(S,H)$) denotes the set of edges with exactly
one (respectively two) endpoints in $S$ (loops are in $\gamma$ but not
$\delta$). $d(v,H)$ denotes the degree of vertex $v$ in $H$.  When
referring to the given graph $G$ we often omit the last argument and write, e.g.,
$\delta(S)$.

 An edge weight function $w$ assigns a numeric weight to each edge. For complexity bounds
we assume the range of $w$ is $[-W..W]$, i.e., the set of integers of magnitude $\le W$.
 The {\it weight of  edge set} $F\subseteq E$ is $w(F) = \sum_{e\in F}w(e)$.
$w$ is {\em conservative} if there are no negative weight cycles. %~\cite{LP}.
 For multigraphs we sometimes write
$w(e,k)$ to denote
the weight of the $k$th copy of edge $e$.

%$\mathbb Z^+$ is the set of positive integers.

 Let $G=(V,E)$
 be a multigraph.
 For a function $f:V\to \mathbb Z_+$, an {\em $f$-factor} is a subset of edges $F\subseteq E$
such that $d(v,F)=f(v)$ for every $v\in V$.
Let $G=(V,E)$
be a graph,
where $E$ may contain loops $vv$ but no parallel edges.
For a function $b:V\to \mathbb Z^+$, a (perfect) {\em $b$-matching}
is a function $x:E \to \mathbb{Z}_+$ such that $\sum_{w: vw\in E} x(vw) = b(v)$.
%(For technical reasons these definitions differ from Section 1 by requiring $f$ and $b$
%o be positive. This is wlog since a vertex with $f(v)=0$ can be deleted from $G$. )
A {\em maximum
$f$-factor} is an $f$-factor $F$ with maximum weight $w(F)$. Similarly
a {\em maximum
$b$-matching} is a perfect $b$-matching of maximum weight.

%+++++
%We are mostly working with critical graphs that are defined in the following way. For each vertex $v\in V$ define $f_v$, the {\em lower perturbation} of $f$ at $v$, by decreasing $f(v)$ by one.
%Similarly define $f^v$, the {\em upper perturbation} of $f$ at $v$, by increasing $f(v)$ by one.
%%=f- \delta_v, \ f^v= f+\delta _v$.
%Each $f_v$, $f^v$, $v\in V$
%is a {\it perturbation} of $f$.
%The notation \fs v. stands for a fixed  perturbation
%that is either
%$f_v$ or $f^v$.
%%$f^-$ stands for a fixed perturbation $f^-_v$ and similarly for  $f^+$.
%\iffalse
%Also for $\sigma \in \{-,+\}$,
%let $\fsi.= f \ \sigma\ \delta _i$, so that
%\fsi. is \fmi. if $\sigma=-$ else \fpi..
%\fi
%We say that graph is {\it $f$-critical} or simply {\em critical} if it has an $f'$-factor for every
%perturbation $f'$ of $f$.
%+++++

%In the reduction we simply set
%$\mu(vu) = \min(f(v),f(u))$. Similar reduction holds for capacitated $f$-matchings.

To simplify the time bounds we assume  matrix multiplication time is $\Omega(n^2\log n)$.
This allows us to include terms like $O(m)$ and $O(m\log n)$ within our overall bound $O(\phi^\omega)$.
Observe that $m=O(n^2)=O(\phi^2)$ if there are no
parallel edges, or if all copies of an edge have the same weight.
In the most general case -- arbitrary parallel edges --
we can assume $m=O(n\phi)=O(\phi^2)$ after linear-time preprocessing.
In proof, for each edge $uv$ the preprocessing discards all but
the $f(u)$ largest copies. This leaves $\le \sum \set {d(u)f(u)} {u\in V}
\le \sum \set {nf(u)} {u\in V}
\le n\phi$
edges in $G$.

%% file: preliminaries.tex
\section{Algebraic preliminaries}
One of the fundamental ideas of our algorithms is to encode graph problems in matrices, in such
a way that determinant of a matrix is (symbolically) non-zero if and only if
the problem has a solution.  The Schwartz-Zippel Lemma~\cite{z79,s80}
provides us with an efficient non-zero test for such symbolic determinants. For our purposes
the following simplified version of it suffices.

\begin{corollary}[Schwartz-Zippel]
  \label{corollary-zippel-schwartz}
For any prime $p$,  if a (non-zero)
  multivariate polynomial of degree $d$ over $\mathcal{Z}_p$ is evaluated at a random point,
  the probability of false zero is
$\le d/p$.
%  most $\Theta\left(\frac{d}{p}\right)$.
% = \Theta(\frac{1}{n|X|})$.
\end{corollary}
In order to use this lemma, we will choose primes $p$ of size $\Theta(n^c)$ for some constant $c$.
We note that in a RAM machine with word size $\Theta(\log n)$, arithmetic modulo $p$ can be realized in constant time.

%We note that multivariate determinants
%have been used in previous work such as
%Kirchhoff's Matrix-Tree Theorem,
%or results on Tutte and Edmonds on perfect matchings.

However, finding the right encoding is just the first and the easiest step, whereas the more complicated
part is to extract the actual solution from this encoding. In order to
do it we use the following algebraic tools. The first one allows us to update the inverse of the matrix 
after we have changed the matrix itself.
\begin{lemma}\hspace{-0.1cm}\textsc{(Sherman-Morrison-Woodbury Formula).}
Let $A$ be $n\times n$ non-singular matrix, and $U,V$ be $n \times k$ matrices, then
\begin{itemize2}
\item $A+UV^T$ is non-singular if and only if the $k\times k$ matrix $I_k+V^TA^{-1}U$ is non-singular and
$\det(A+UV^T) = \det(A)\det(I_k+V^TA^{-1}U)$,
\item if $A+UV^T$ is non-singular then
$
(A+UV^T)^{-1} = A^{-1} - A^{-1}U (I_k+V^TA^{-1}U)^{-1}V^TA^{-1}.
$
\end{itemize2}
\end{lemma}
When $k=1$ the matrices $U$ and $V$ become length $n$ vectors. Such an update is called {\em rank-one
update}. In this special case the above lemma is called Sherman-Morrison formula.
Observe that for $k=1$ we can compute $(A+UV^T)^{-1}$ from $A^{-1}$  in $O(n^2)$ arithmetic operations.

%\begin{lemma}[Sherman-Morrison formula]
%\label{lemma:sherman}
%Let $A$ be a $n \times n$ non-singular matrix and let $u,v$ be vectors of length $n$, then
%\begin{itemize}
%\item $A+uv^T$ is non-singular if and only if $1+v^TA^{-1}u \neq 0$,
%\item if $A+uv^T$ is non-singular then $(A+uv^T)^{-1} = A^{-1}-\frac{A^{-1}uv^TA^{-1}}{1+v^TA^{-1}u}$.
%\end{itemize}
%\end{lemma}
%Observe that having $A^{-1}$ we can compute $(A+uv^T)^{-1}$ in $O(n^2)$ arithmetic operations.

In our algorithms we use the above formula but restricted to submatrices. Let $R$ ($C$) denote set of
rows (respectively columns) of matrix $A$. We denote the submatrix of $A$ restricted to rows in $R$ and
columns in $C$ by $A[R,C]$. Harvey~\cite[Corollary 2.1]{Harvey06} has observed that
$A^{-1}[S,S]$ can be computed in $O(|S|^{\omega})$ time after updates have been made to submatrix $A[S,S]$.

%\begin{corollary}[Corollary 2.1 from~\cite{Harvey06}]
%\label{corollary:Harvey}
%Let $A$ be a $n\times n$ non-singular matrix, and let $\tilde{A}$ be a matrix which
%is identical to $A$ except for $\tilde{A}_{S,S}\neq A_{S,S}$, then
%\begin{itemize}
%\item $\tilde{A}$ is non-singular if and only if the $|S|\times |S|$ matrix $I+ (\tilde{A}_{S,S}- A_{S,S})\cdot (A^{-1})_{S,S}$ is non-singular,
%\item if $\tilde{M}$ is non-singular then
%\[
%\tilde{A}^{-1} = A^{-1} - (A^{-1})_{*,S}\left[I_{|S|} + (\tilde{A}_{S,S}- A_{S,S})\cdot (A^{-1})_{S,S}\right]^{-1} (\tilde{A}_{S,S}- A_{S,S})\cdot (A^{-1})_{S,*}.
%\]
%\end{itemize}
%Note that by the first point testing whether $\tilde{A}$ is non-singular can be realized in $O(n^{\omega})$ time.
%\end{corollary}

The final tool is rather recent result in symbolic computation by
Storjohann~\cite{storjohann03}. He has shown how to
compute a determinant of a polynomial matrix, as well as, how to
solve a rational system for polynomial matrix.

\begin{theorem}[Storjohann '03]
  \label{theorem-storjohann}
Let $K$ be an arbitrary field,
  $A \in K[y]^{n \times n}$ a polynomial matrix of degree $W$,
  and $b \in K[y]^{n \times 1}$ a polynomial vector of the same degree. Then
  \begin{itemize2}
  \item rational system solution $A^{-1}b$ (Algorithm~5 \cite{storjohann03}),
  \item determinant $\det(A)$ (Algorithm~12 \cite{storjohann03}),
  \end{itemize2}
  can be computed in $\tilde{O}(Wn^{\omega})$ operations in $K$, w.h.p.
\end{theorem}

%% file: bipartite-determinant.tex
\section{Outline of the paper, and the bipartite case}
\label{section:bipartite}
This section has two purposes. First it presents our algorithms for
bipartite graphs, a simplification of the general approach.
(But even in the bipartite case our techniques were not previously known.)
Second, it gives a guide to the entire paper: Each time we introduce a
construction we discuss how it can be extended to the general case.
The section ends by giving our algorithms for vertex-capacitated flow.

\subsection{Determinant formulations}
\label{section:biparite:formulations}
Consider a simple bipartite graph $G$, with both vertex sets $V_0, V_1$ numbered from 1 to $n$.
Let $\phi=\sum_i f(i)/2$. Define a $\phi \times \phi$ matrix $B(G)$,
the {\em symbolic adjacency matrix of $G,f$}, as follows. 
% ????????????????????????is this really an adjacency matrix 
A
vertex $i\in V_0$ is associated with $f(i)$ rows, which are indexed by a pair $i, r$, for
$0\le r <f(i)$. Similarly $j\in V_1$ is associated with $f(j)$ columns indexed by
$j,c$, for $0 \le c <f(j)$. $B(G)$ uses indeterminates
$x^{ij}_r$ and $y^{ij}_c$ and is defined by
\begin{equation}
\label{SimpleBipartiteMatrixEqn}
B(G)_{i,r,j,c} =
\begin{cases}
x^{ij}_{r} y^{ij}_{c}& ij\in E,\cr
0&\text{otherwise.}
\end{cases}
\end{equation}
Observe that each edge in the graph is represented by a rank-one submatrix given by the product
of two vectors $x^{ij} (y^{ij})^T$. Before we prove the main theorems we make the following 
observation that will show each edge can be used only once.

\begin{lemma}
\label{lemma-submatrix}
Let $A$ be a symbolic $n \times n$ matrix, let $R$ be the set of $m$ rows, and $C$ be
the set of $m$ columns of $A$. If $A[R,C]$ has rank bounded by $r$ then each term in the
expansion of $\det(A)$ contains at most $r$ elements from $A[R,C]$.
\end{lemma}
\begin{proof}
Using Laplace expansion we expand $\det(A)$ into $m \times m$ minors that contain all $m$ rows of $R$, i.e.,
 \[
\det(A)= \sum_{M\subseteq V_1,|M|=m} \sgn(M) \det(A[R,M])\det(A[V_0-R,V_1-M]),
 \]
where $\sgn(M) = \prod_{c\in M}(-1)^c$. Consider now each element of the above sum separately.
If $M$ contains $>r$ columns of $C$ then $\det(A[R,M])=0$, so the elements contributing to $\det(A)$ have $\le$ columns of $C$.
Moreover, $A[V_0-R,V_1-M]$ has no rows of $A[R,C]$, so $\det(A[R,M])\det(A[V_0-R,V_1-M])$ has $\le r$ entries from $A[R,C]$.
%A rank $r$ matrix of size $n\times n$ can be written as a product two $n\times r$ matrices, i.e., we have
%\[
%A_{R,C} = U_{R,*} V_{C,*}^T,
%\]
%where $U$ is $n\times r$ matrix with zeros in rows not in $R$ and $V$ is $n\times r$ matrix with zeros in rows not in $C$. From Sherman-Morrison-Woodbory formula we have
%\[
%\det(A) = \det(A-VU^T+VU^T) =  \det(A-VU^T)\det(I_r+V^TA^{-1}U).
%\]
%Observe that $(A-VU^T)_{R,C}=0$ and $\det(I_r+V^TA^{-1}U)$ is a polynomial of degree $r$, so elements of $V$ and $U$ appear in
%$\det(A)$ with degree at most $r$. Each element of $A_{R,C}$ is a degree $1$ polynomial of elements from $U$ and $V$ so it cannot
%appear with degree more then $r$ in $\det(A)$.
\end{proof}

The determinant of $B(G)$ is the sum of many different terms each containing exactly $\phi$ occurrences of variable
pairs $x^{ij}_{r}y^{ij}_{c}$. Each pair $x^{ij}_{r}y^{ij}_{c}$ corresponds to an edge $ij \in E$.
For a term $\sigma$ let $F_{\sigma}$ denote the multiset of edges that correspond to the variable pairs in $\sigma$. 
Define $\F.$ to be the function that maps each term $\sigma$ to $F_{\sigma}$.

\begin{theorem}
\label{BipSimpleDetThm}
Let $G$ be a simple bipartite graph. The function $\F.$ from terms in $\det(B(G))$ is
a surjection onto the $f$-factors of $G$. Consequently, $G$ has an $f$-factor if and only if $\det(B(G))\ne  0$.%
\footnote{The second part of the theorem suffices for undirected $f$-factors. But we will need the stronger claim of surjectivity for the weighted case.}
\end{theorem}

\begin{proof}
First we show that the image of $\F.$ contains all $f$-factors of $G$.
Suppose $F$ is an $f$-factor in $G$. Order the edges of $F$
that are incident to each vertex arbitrarily.
If $ij \in F$ is the $r+1$st edge at $i$ and the $c+1$st edge at $j$
then it corresponds to entry $B(G)_{i,r,j,c}$. Thus $F$
corresponds to a nonzero term $\sigma$ in the expansion of  $\det(B(G))$.
Observe that entries $B(G)_{i,r,j,c}$ define $\sigma$ in
a unique way, so no other term has exactly the same
indeterminates. Of course there can be many terms representing $F$.

Can 
$\det(B(G))$ contain
terms that do not correspond to $f$-factors? We show the answer is no.
Suppose 
$\det(B(G))\ne 0$ and take any term $\sigma$ in the expansion of $\det(B(G))$.
$\sigma$ corresponds to an $f$-factor unless
more than one entry corresponds to the same edge
of $G$. This is impossible because edges are represented by rank-one
submatrices and Lemma~\ref{lemma-submatrix} shows elements of such submatrices
appear at most once.
\end{proof}

Now let $G$ be a bipartite multigraph. Let $\mu(e)$ denote the
multiplicity of any edge $e\in E$. %, and index the copies of $e$ as $e_k$, %%\ldots, \mu(e)$.
%Assume that $\mu(e)=0$ if $e\notin E$.
Redefine the corresponding entry in $B(G)$ by
\begin{equation}
\label{MultiBipartiteMatrixEqn}
B(G)_{i,r,j,c} = \sum _{k=1}^{\mu(ij)} x^{ij,k}_{r} y^{ij,k}_{c}.
\end{equation}
In other words, now  edge $ij$ of multiplicity $\mu(ij)$ is represented by a submatrix
of rank $\mu(ij)$. Hence Lemma~\ref{lemma-submatrix} shows edge $ij$ can appear in a term of $\det(B(G))$
at most $\mu(ij)$ times. This leads to the following generalization of Theorem~\ref{BipSimpleDetThm}.

\begin{corollary}
\label{BipartiteMultiDetCor}
Let $G$ be a bipartite multigraph. The function $\F.$ from terms in $\det(B(G))$ is
a surjection onto the $f$-factors in $G$. Consequently, $G$ has an $f$-factor iff $\det(B(G))\ne  0$.
\end{corollary}

%\begin{proof}
%The argument is almost identical but we give it in detail anyway.
%
%First, we show that the image of $\F.$ contains all $f$-factors of $G$.
%Suppose $G$ has an $f$-factor $F$. Order the edges of $F$ that are
%incident to each vertex arbitrarily.  If the $k$th copy of $ij \in
%F$ is the $r+1$st edge at $i$ and the $c+1$st edge at $j$ then it
%corresponds to the term $x^{ij,k}_{r} y^{ij,k}_{c}$ in the entry
%$b_{i,r,j,c}$ of $B$. Thus $F$ corresponds to a nonzero term
%$\sigma$ in the expansion of $\det(B)$. Observe that entries $
%b_{i,r,j,c}$ define $\sigma$ in an unique way, so no other term has exactly the same
%indeterminates.
%
%%Any term in the expansion has
%%exactly one variable $x^\cdot_{ir}$ with given values of $i$ and $r$,
%%and similarly for $y^\cdot_{jc}$. The terms in the summation of
%%(\ref{MultiBipartiteMatrixEqn}) involve distinct variables.  So no
%%other term of $\det(B)$ has exactly the same indeterminates as
%%$\sigma$, i.e., $\det(B)$ has a nonzero term.
%
%Suppose that
%$\det(B)\ne 0$, and take any term $\sigma$ in the expansion of $\det(B)$.
%The term $\sigma$ corresponds to an $f$-factor unless
%an edge $ij$ is used more then $\mu(ij)$ times. However, this is not possible because edges
%are represented by  submatrices of rank $\mu(ij)$ and by Lemma~\ref{lemma-submatrix} elements of such submatrix
%can appear at most $\mu(ij)$ times.
%\end{proof}

%\bigskip

The final point of this section
concerns the complexity of using the $B(G)$ matrix.
As in most algebraic algorithms we evaluate $B(G)$ 
using a random value for each indeterminate to get a matrix $B$. 
If $G$ is a simple graph
this is easily done in time $O(\phi^2)$. But the situation is less clear for multigraphs. The most extreme case is exemplified by $b$-matching. 
Consider
an arbitrary edge $ij$  and let $B[I,J]$ denote the $f(i)\times f(j)$ submatrix of $B$ that represents it.
$ij$  has multiplicity $\mu(ij)\le \min\{f(i),f(j)\}$.
When $f(i)$ and $f(j)$ are $\Theta(\phi)$ the time to compute
the expression of \eqref{MultiBipartiteMatrixEqn} is $\Theta(\phi)$.
So computing the $\Theta(\phi^2)$ entries 
of $B[I,J]$ 
uses time $\Theta(\phi^3)$. 
But this can be avoided and we can construct $B$ in time
$O(\phi^\omega)$, as follows. 

As observed above $B[I,J]$ 
is the product of an $f(i) \times \mu(ij)$ matrix
$X$ of indeterminates $x^{ij,k}_{r}$ and an 
$\mu(ij)\times f(j)$ 
matrix $Y$
of indeterminates $y^{ij,k}_{c}$. Wlog assume $f(i)\le f(j)$. We have the
trivial bound $\mu(ij)\le f(i)$. 
Break up the product $XY$ into products of $\mu(ij)\times \mu(ij)$ matrices,
i.e., break $X$ into $f(i)/\mu(ij)$ matrices of size $\mu(ij)\times \mu(ij)$,
and similarly $Y$, to get $XY$ as a sum of $\frac{f(i)}{\mu(ij)}
\frac{f(j)}{\mu(ij)}$ products of $\mu(ij)\times \mu(ij)$ matrices. Using fast matix multiplication on these products, the total time to compute $B$ is bounded by a constant times
 \[\sum_{f(i)\le f(j)} f(i)f(j)\mu(ij)^{\omega-2}   \le
\sum_{f(i)\le f(j)} f(i)^{\omega-1} f(j)\le
\sum_i f(i)^{\omega-1}(\sum_{f(i)\le f(j)} f(j))
\le \sum_i \phi^{\omega-2}f(i)\phi 
%\le   \phi^{\omega-1} \sum_i  f(i)
\le \phi^{\omega}.\]  

\paragraph{Generalizations}
The generalization of these ideas to non-bipartite graphs is given in Section~\ref{sec:determinant}.
We combine the above idea, submatrices with bounded rank,  with the idea of
Tutte to construct a skew-symmetric matrix. Additionally we need to take care of self-loops in multigraphs.
In the Tutte matrix self-loops do not need appear since they
cannot be used in a $1$-factor.

%% file: bipartite-algorithms.tex
\subsection{Finding $f$-factors}
\label{section:biparite:algorithms}

This section gives our algorithm to find an $f$-factor of
a  bipartite multigraph. We follow the
development from~\cite{ms04}:
We start with  an $O(\phi^3)$-time algorithm. Then we show
it can be implemented in $O(\phi^{\omega})$ time using the 
Gaussian elimination algorithm of Bunch and Hopcroft~\cite{hb}.

An {\em allowed edge} is an edge belonging to some $f$-factor.
For perfect matchings
the notion of allowed edge is easily expressible using the
inverse of $B(G)$:
$ij$ is allowed if and only if $B(G)^{-1}_{i,j}$ is non-zero~\cite{rv89}.
We will prove a similar
statement for bipartite $f$-factors.
(But such a statement fails for
non-bipartite graphs -- see Appendix~\ref{appendix-allowed-edges}.)
For a given $f$ define $f_{i,j}$ to be
\[
f_{i,j}(v) = \begin{cases}
f(v)-1 & \textrm{if } v = i \textrm{ or } v=j,\\
f(v) & \textrm{otherwise.}
\end{cases}
\]
\begin{lemma}
\label{BipartiteAllowedLemma}
Let $G$ be a bipartite multigraph having an $f$-factor. Edge $ij\in E$ is
allowed if and only if $G$ has an $f_{i,j}$-factor. 
%Additionally, if edge $ik\in E$, belongs to some $f_{i,j}$-factor in $G$ then it is allowed
%in $G$.
\end{lemma}
\begin{proof}
%TODO rysunek
The ``only if'' direction is clear:
If $F$ is an $f$-factor containing
$ij$, then $F-ij$ is an $f_{i,j}$-factor.

Conversely, suppose $F$ does not contain
the chosen edge $ij$. Take an $f_{i,j}$-factor $F'$ that maximizes $|F'\cap F|$.
$F'\oplus F$ contains an
alternating $ij$-trail $P$
that starts and ends with edges of $F$.
In fact $P$ is a path.
(Any cycle $C$ in $P$
has even length and so is alternating.
This makes $F'\oplus C$ an $f_{ij}$-factor containing
more edges of $F$ than $F'$, impossible.)

$ij$ is not the first edge of $P$ ($ij\notin F$).
So $ij\notin P$, since $P$ is vertex simple.
Thus $(F\oplus P) +ij$ is an $f$-factor containing $ij$.
\end{proof}

\iffalse
 and ends with edges of $F$.
So $ij$ is not the first edge of $P$.
In fact $ij\notin P$, since $P$ is vertex simple.
(In proof,
any cycle in $P$
has even length, so it can be discarded
to get a shorter alternating path.)
Thus $(F\oplus P) +ij$ is an $f$-factor containing $ij$.
\end{proof}

Assume now that $F'$ uses edge $ik$. Observe that path $P$
cannot contain edge $ik$ because it is vertex simple and
starts with edges of $F$. Hence, $(F\oplus P)+ij$ is
an $f$-factor containing both $ij$ and $ik$.
\end{proof}
\fi

\begin{lemma}
\label{lemma:f_i_j-factor}
Let $G$ be a bipartite multigraph having an $f$-factor, let $B(G)$ be its symbolic adjacency matrix, and let $i\in V_0$ and $j \in V_1$.
Then $(B(G)^{-1})_{j,0,i,0}\neq 0$ if and only if $G$ has $f_{i,j}$-factor.
\end{lemma}
\begin{proof}
Observe that
\[
(B(G)^{-1})_{j,0,i,0} = \frac{(-1)^{n(j,0)+n(i,0)}\det(B(G)^{i,0,j,0})}{\det(B(G))},
\]
where $B(G)^{i,0,j,0}$ is the matrix $B(G)$ with $i,0$'th row and $j,0$'th column removed, and $n(i,k)$ is the
actual index of the row or column given by the pair $i,k$. We have $\det(B(G))\neq 0$ since $G$ has an $f$-factor. Hence $(B(G)^{-1})_{j,0,i,0}\neq 0$
if and only if $\det(B(G)^{i,0,j,0}) \neq 0$. Furthermore $B(G)^{i,0,j,0}$ is the symbolic adjacency matrix
obtained from $G$ for $f_{i,j}$-factors.
\end{proof}
Observe that by the symmetry of the matrix $B(G)$, when $(B(G)^{-1})_{j,0,i,0}\neq 0$ then as
well $(B(G)^{-1})_{j,\kappa,i,\iota}\neq 0$ for all $0\le \iota <f(i)$ and $0\le \kappa <f(j)$.
Combining the above two lemmas with this observation we obtain the following.

\begin{corollary}
\label{corollary:allowed-edge}
Let $G$ be a bipartite multigraph having an $f$-factor, and let $B(G)$ be its symbolic adjacency matrix. The edge $ij\in E$ is
allowed if and only if $(B(G)^{-1})_{j,\kappa,i,\iota}\neq 0$ for any $0\le \iota <f(i)$ and $0\le \kappa <f(j)$.
\end{corollary}

%Let $e=ij \in E(G)$ then $B(G)$ and $B(G-e)$ differ from one another by a rank-one update, i.e.,
%\[
%B(G-e) = B(G) - x_{i}^e (y_j^e)^T
%\]
%where $x_{i}^e$ and $y_j^e$ are length $\phi$ vectors defines as
%\begin{eqnarray*}
%(x_{i}^e)_k = \begin{cases}
%x_{i,r-n(i)}^e & \textrm{if } n(i) \le k <n(i)+f(i),\\
%0 & \textrm{otherwise,}
%\end{cases} & \textrm{ and } &
%(y_{j}^e)_k = \begin{cases}
%y_{j,r-n(j)}^e & \textrm{if } n(j) \le k <n(j)+f(j),\\
%0 & \textrm{otherwise.}
%\end{cases}
%\end{eqnarray*}

Being equipped with a tool for finding allowed edges we can now use the Gaussian elimination framework from~\cite{ms04}.
The following observation is useful.
\begin{lemma}[\cite{ms04}]
\label{lemma-elimination}
Let $A$ be a non-singular $\phi \times \phi$ matrix and let $1\le i,j \le n$ be such that $(A^{-1})_{j,i} \neq 0$.
Let $A'$ be the matrix obtained from $A^{-1}$ by eliminating row $j$ and column $i$ using Gaussian elimination. Then
$A'= (A^{i,j})^{-1}$ (i.e., $A'$ is the Schur complement of $(A^{-1})^{j,i}$).
\end{lemma}

The above lemma can be used to obtain the following algorithm that finds an $f$-factor.

\begin{algorithm}[H]
\caption{An $O(\phi^3)$ time algorithm to find an $f$-factor in a bipartite multigraph $G$.}
\begin{algorithmic}[1]
\State{Let $B(G)$ be the $\phi \times \phi$ matrix representing $G,f$}
\State{Replace the variables in $B(G)$ by random values from $\mathcal{Z}_p$ for prime $p=\Theta(\phi^2)$ to obtain $B$}
\State{If $B$ is singular return "no $f$-factor"}
\State{(with probability $\ge 1-\frac{1}{\phi}$ matrix $B$ is non-singular when $B(G)$ is non-singular)}\comment{by
Cor.~\ref{corollary-zippel-schwartz}}
\State{$F:=\emptyset$}
\State{Compute $N:=B^{-1}$}
%\State{($\det(B)\neq 0$ so each row and column of $B^{-1}$ has a nonzero element, i.e., an allowed edge)}
\State{(each column $i,\iota$ of $B{^-1}$ has an allowed edge, since
$B B^{-1}=I$ gives $j,\kappa$ with
$B_{i,\iota,j,\kappa}B^{-1}_{j,\kappa,i,\iota}\neq 0$)} 
%\comment{by Lemma~\ref{corollary-zippel-schwartz}}
\For{$i=[1..n]$}
\For{$\iota=[0..f(i)-1]$}
\State{Find $j,\kappa$ such that $ij \in E - F$  and $N_{j,\kappa,i,\iota}\neq 0$}\comment{by Corollary~\ref{corollary:allowed-edge} edge $ij$ is allowed}
\State{Eliminate the $j,\kappa$'th row and the $i,\iota$'th column from $N$}\comment{using Gaussian elimination}
\State{(Lemma~\ref{lemma-elimination} shows $N=(B^{i,\iota,j,\kappa})^{-1}$, i.e., $N$ encodes $f_{i,j}$-factors, but see below)}
\label{algorithm:f-factor-bipartite-simple:eliminate}
%\State \comment{Equivalently we can zero out all element in $j,\iota$th row and the $i,\iota$th column, besides $N_{j,\iota,i,\iota}$}
\State{Set $F:=F+ ij$}
\EndFor
\EndFor
\State{Return $F$}
%\EndProcedure
\end{algorithmic}
\label{algorithm:f-factor-bipartite-multi}
\end{algorithm}

The comment of line 12 is adequate for $\iota=0$.
However $\iota>0$ requires an additional observation.
To see this first recall the logic of each iteration: Let $f'$ be the residual degree requirement function,
i.e.,
the current $F$, enlarged with an $f'$-factor
of the current graph, gives  an $f$-factor of $G$.
In line 10, the $f'$-factor $F'$ that contains $ij$
is a subgraph of the graph corresponding to (the current) $N$ and its corresponding matrix $B$. 
Now suppose the iteration for $\iota=0$ adds edge $ip$ to $F$.
When the row and column for $ip$ are deleted from $B$,
the remaining rows for vertex $i$ still contain entries corresponding to edge $ip$
(recall the definition of $B(G)$). So when the iteration for $\iota=1$ chooses its edge $ij$,
the corresponding $f'$-factor $F'$ may contain edge $ip$.  But $F$ cannot be enlarged
with $F'$, since that introduces two copies of $ip$.
The same restriction applies to iterations for $\iota>1$, but now it concerns all previously chosen
edges $ip$.

Actually there is no problem because we can guarantee an $f'$-factor avoiding all the previous
$ip$'s exists. The guarantee is given by the following corollary to Lemma \ref{BipartiteAllowedLemma}.
(Note when  $r=0$ the corollary is a special case of the lemma. Also, the converse of the corollary
holds trivially.)

\begin{corollary}
Consider a set of edges $P=\{i p_1, \,\ldots, i p_r\}$,   %\, j q_1,\,\ldots,\, j q_s\}$, $r,s\ge 0$.
$r\ge 0$. Suppose $G-P$ has an $f$-factor. If $G$ has an $f_{ij}$-factor for some edge $ij
\notin P$ then $G-P$ has an $f$-factor containing $ij$.
\end{corollary}

\begin{proof}
The proof of Lemma \ref{BipartiteAllowedLemma} applies, assuming we start by taking $F$
to be the assumed $f$-factor.
\end{proof}

Finally observe that Algorithm 1 is implementing Gaussian elimination on $B^{-1}$, the only difference being that pivot elements are chosen to
correspond to edges of the graph. If there exists an $f$-factor, there is an allowed edge incident to each vertex.
Hence, even with this additional requirement Gaussian elimination is able to find a non-zero element in each row of $B^{-1}$.

Bunch and Hopcroft~\cite{hb} show how to speed up the running time of Gaussian elimination
from $O(\phi^3)$ to $O(\phi^{\omega})$, by using lazy updates to the matrix. Let us divide the columns of
the matrix into two almost equal parts. Let $L$ denote the first $\lceil \phi/2 \rceil$
columns that are to be eliminated, whereas let $R$ denote the remaining $\lfloor \phi/2 \rfloor$ columns.
Bunch and Hopcroft observed that columns in $R$ are not used until we eliminate all columns from $L$.\footnote{In their paper
the elimination proceeds row by row, whereas it is nowadays more usual to present Gaussian elimination on columns.}
Hence all  updates to columns in $R$ resulting from elimination of columns in $L$ can be done once
using fast matrix multiplication in $O(|R|^{\omega})$ time. By applying this scheme
recursively one obtains an $O(\phi^{\omega})$ time algorithm.

\paragraph{Generalizations}
Section~\ref{sec:algorithms} gives $O(\phi^{\omega})$ time algorithms for finding $f$-factors in non-bipartite multigraphs. There are
several things that need to be done differently. As discussed above the criteria for finding allowed edges -- Corollary~\ref{corollary:allowed-edge} -- does not work any more.
We need to work with the weaker notion of {\em removable edges}, i.e.,
the edges that can be removed from the graph so that it still contains an $f$-factor. This forces us to use a different approach,
based on Harvey~\cite{Harvey06}, which works with removable edges. This poses a new challenge, to handle multiple copies
of edges in multi-graphs, as removing such edges one by one could require $O(\phi^3)$ time. To overcome this problem we use binary search with the
Sherman-Morrison-Woodbury formula to remove multiple copies of an edge in one shot.

%% file: bipartite-weights.tex
\subsection{Weighted \boldmath {$f$}-factors}
\label{section:weighted-bipartite}
In this section we discuss how to find a maximum $f$-factor in a weighted bipartite graph.
For the sake of simplicity we assume in this section that the weight function is non-negative ,i.e.,
$w:E\to [0..W]$. If this is not the case we can redefine $w(ij):=w(ij)+W$, what changes
the weight of each $f$-factor by exactly $W f(V)/2 = W\phi$.
Let us start by recalling the dual problem for maximum $f$-factors. In this
problem each vertex $v$ is assigned a real-valued weight $y(v)$. We say that the dual $y$
{\em dominates} the edge $uv\in E$ when $y(u)+y(v) \ge w(e)$, or it {\em underrates}
the edge $uv\in E$ when $y(u)+y(v) \le w(uv)$.  The objective that we need to minimize in
dual problem is
\[
y(V,E) = \sum_{v\in V}f(v)y(v) + \sum_{uv\in E\textrm{ is underrated}} w(uv)-y(u)-y(v).
\]
The dual $y$ minimizes $y(V,E)$, when there exists an $f$-factor $F$ such that $F$
contains only underrated edges, whereas its
complement contains only dominated edges. Observe that when we are given
the minimum dual $y$, then the above $f$-factor $F$ is a maximum weight $f$-factor.
On the other hand, in order to construct such maximum $f$-factor we need to take into it every strictly underrated edge
and arbitrary {\em tight} edges, i.e., edges $uv\in E$ for which $y(u)+y(v) = w(uw)$. Hence,
we can observe the following.
\begin{lemma}
\label{lemma-finding-weighted}
Given an optimal dual function $y$, a maximum $f$-factor of a bipartite multigraph can be constructed in $O(\phi^{\omega})$ time.
\end{lemma}
\begin{proof}
Let $U$ be equal to the set of underrated edges with respect to $y$. Set $f'(v) = f(v)-d(v,U)$.
Using Algorithm~\ref{algorithm:f-factor-bipartite-multi} find an $f'$-factor $T$ over the
set of tight edges with respect to $y$. The maximum $f$-factor is equal to  multiset sum of $U$ and $T$, i.e., to $U\uplus T$.
\end{proof}

This lemma shows that given an algorithm for finding unweighted $f$-factors all we need to know
is an optimal dual. Such an optimal dual can be obtained from the combinatorial interpretation
as given in~\cite{G}. Let us define $G^+$ to be $G$ with additional vertex $s\in V_1$ and
new 0 weight edges $su$, for all $u\in V_0$. In $G^+$ we set $f(s)=1$. Let $f_v$ be the
degree constraint function defined to be identical to $f$ except for $f_v(v)=f(v)+(-1)^i$, where
$v\in V_i$. Let $F_v$ be a maximum $f_v$-factor in $G^+$. To show that $F_v$ always
exists take $F$ to be any $f$-factor in $G$. Now, when $v\in V_0$ then $F+sv$ is an $f_v$-factor,
whereas when $v\in V_1-s$ then for any $uv\in F$, $F-uv+su$ is an $f_v$-factor

\begin{theorem}[\cite{G}]
\label{theorem-combinat-bip}
For a bipartite multigraph with an $f$-factor, optimal duals are given by
$y(v)=(-1)^iw(F_v)$ for $v\in V_i$.
\end{theorem}

Hence, in order to construct optimal dual we need to know weights $w(F_v)$, for all $v \in V_0\cup V_1$.
At first sight it might seem that we did not gain anything, as instead of finding one $F$ factor
now we need to find all factors $F_v$. However, we do not need to find these factors. We only
need to know their weights, which is much easier. And the following lemma shows that we just need
to know $w(F_v)$ for one side of the bipartite graph.

\begin{lemma}
\label{lemma-secondside-bip}
For a bipartite multigraph with an $f$-factor, let $y(v)$ be an optimal dual for each $v\in V_1$.
An optimal dual $y(u)$ for $u \in V_0$ is equal to the largest value $y_u$ that makes
at least $f(u)$ edges incident to $u$ underrated, i.e.,
$|\{uv^{\in E} : y_u\le w(uv)-y(v)\}|\ge f(u)$.
\end{lemma}
\begin{proof}
Observe that there are at least $f(u)$ underrated edges incident to $u$ with respect to optimal dual $y$, as each
maximum $f$-factor needs to contain only underrated edges. On the other hand, the complement 
 contains at least $d(u)-f(u)$
dominated edges. This fixes the largest possible value for $y(u)$ 
as the value given in the lemma.
\end{proof}

Now consider a simple bipartite graph $G$ and similarly to~\eqref{SimpleBipartiteMatrixEqn}
define $B(G)$ as
\begin{equation*}
\label{SimpleBipartiteMatrixEqnW}
B(G)_{i,r,j,c} =
\begin{cases}
z^{w(ij)} x^{ij}_{r} y^{ij}_{c}& ij\in E,\cr
0&\text{otherwise,}
\end{cases}
\end{equation*}
where $z$ is a new indeterminate.
Theorem~\ref{BipSimpleDetThm} shows that there is a mapping $\F.$ from
terms of $\det(B(G))$ onto $f$-factors in $G$. Consider a term $\sigma$ in $\det(B(G))$. Observe
that its degree in $z$ is equal to the weight of $\F.(\sigma)$ because the powers
of $z$ get added in the multiplication. For a polynomial $p$, denote the
degree of $p$ in $z$ by $\deg_z(p)$. We obtain the following
observation.
\begin{corollary}
\label{corollary-BipSimpleWeights}
For a simple bipartite graph $G$, $\deg_z(\det(B(G)))$ equals the weight of a maximum $f$-factor in $G$.
\end{corollary}

To compute $f_v$-factors in $G^+$ we use the following auxiliary graph. Let $G_*$ be $G^+$ with an additional vertex $t\in V_0$ that is joined to every vertex of $V_1-s$ by an edge of weight zero. 
%As previously, we s ???????????
Set $f(t)=1$. 
%For this graph the following holds.
\begin{lemma}
\label{lemma:LowerPerturbationBipartite}
Every $v\in V_1$ satisfies $\deg_z(\adj(B(G_*))_{v,0,t,0}) = w(F_v).$
\end{lemma}
\begin{proof}
Observe that
\[
\adj(B(G_*))_{v,0,t,0} = (-1)^{n(t,0)+n(v,0)} \det(B(G_*)^{t,0,v,0}),
\]
where $B(G_*)^{t,0,v,0}$ is the matrix $B(G_*)$ with row $t,0$ and column $v,0$ removed and $n(i,r)$ gives
the order of rows and columns indexed by pairs $i,r$. By Theorem~\ref{BipSimpleDetThm} we know that $\det(B(G_*))$ consists of terms corresponding to $f$-factors in $G_*$. Hence the above equality
shows  terms of $\adj(B(G_*))_{v,0,t,0}$ correspond to $f$-factors that use edge $tv$, but with this edge removed. These are exactly the $f_v$-factors in $G^+$, because forcing the $f$-factor to use edge $tv$  effectively
decreases $f(v)$ by $1$. As we observed in Corollary~\ref{corollary-BipSimpleWeights} the degree of $z$ equals the total weight of corresponding $f$-factor. The lemma follows.
\end{proof}

Recall that the adjoint of a nonsingular matrix $A$ is $\det(A) A^{-1}$.
The lemma shows we are interested in column $t,0$ of the adjoint.
So let
$e_{t,0}$
be a unit vector,
with 1 in row $t,0$ and   zeroes elsewhere. 
Then the desired weights are found in the vector
$\adj(B(G_*))  e_{t,0}=\det(B(G_*))         B(G_*)   ^{-1} e_{t,0}$.
This leads to the following algorithm to find optimal duals for weighted bipartite $f$-factors.

\begin{algorithm}[H]
\caption{An $\tilde{O}(W \phi^{\omega})$ time algorithm to find optimal duals $y$ in a bipartite graph $G,f$.}
\begin{algorithmic}[1]
\State{Let $B(G_*)$ be $\phi \times \phi$ matrix representing $G_*$}
\State{Replace $x$ and $y$ variables in $B(G_*)$ by random values from $\mathcal{Z}_p$ for prime $p =\Theta(\phi^3)$ to obtain $B$}
\State{Compute vector $a := \adj(B) e_{t,0} = \det(B) B^{-1} e_{t,0}$.} \comment{requires $\tilde{O}(W\phi^{\omega})$ time using Theorem~\ref{theorem-storjohann}}
\For{$v \in V_1$}
\State{($\deg_z(a_v) = \deg_z(\adj(B(G_*))_{v,0,t,0})$ with probability $\ge 1-\frac{1}{\phi^2}$)}\comment{by Cor.~\ref{corollary-zippel-schwartz}}
\State{Set $w(F_v) := \deg_z(a_v)$}\comment{equality holds by Lemma~\ref{lemma:LowerPerturbationBipartite}}
\State{Set $y(v) := w(F_v)$} \comment{$y(v)$ is optimal by Theorem~\ref{theorem-combinat-bip}}
\EndFor\comment{by union bound all $y(v)$ are correct with probability $\ge 1-\frac{1}{\phi}$}
\For{$u \in V_0$}
\State{Set $y(u) := \max \{y_u: |\{uv^{\in E} : y_u\le w(uv)-y(v)\}|\ge f(u)\}$}
\State \comment{$y(u)$ is optimal by Lemma~\ref{lemma-secondside-bip}}
\EndFor
%\EndProcedure
\end{algorithmic}
\label{algorithm:duals-simple-bip}
\end{algorithm}

Combining the above algorithm with Lemma~\ref{lemma-finding-weighted} we obtain an $\tilde{O}(Wn^{\omega})$ time algorithm for 
maximum $f$-factors in weighted bipartite graphs. It can be observed that this development works for bipartite multi-graphs
as well, when one changes the definition of $B(G)$ to
\begin{equation*}
\label{MultiBipartiteMatrixEqnW}
B(G)_{i,r,j,c} = \sum _{k=1}^{\mu(ij)} z^{w(ij,k)} x^{ij,k}_{r} y^{ij,k}_{c}.
\end{equation*}
As in Section \ref{section:bipartite}, some problem instances 
require a more careful construction of $B$.
But the algorithm of Section \ref{section:bipartite} is easily extended. 
First note that
$b$-matching problems can be handled using the algorithm
of Section \ref{section:bipartite} unchanged, since every
copy of a fixed edge $ij$ has the same weight $w(ij,k)=w(ij)$. 
The most general case ($f$-factor problems with high multiplicities
and parallel edges of different weights) 
is easily handled as follows.
Decompose $G$ into multigraphs $G_w$
containing the edges of weight $w$, for $ w=0,\ldots, W$.
So $B(G)=\sum_{w=0}^W z^w B(G_w)$, and analogously, 
$B=\sum_{w=0}^W z^w B_w$ (where all $B_w$'s use the same random values for
the $x$ and $y$ variables). 
Compute each $B_w$ using the algorithm of Section \ref{section:bipartite}
and combine. 
Each $B_w$ is found in time $O(\phi^\omega)$ so the total time is $O(W\phi^\omega)$.

\paragraph{Generalizations}
The relation between the primal and dual problems in the bipartite case is considerably simpler than  the general case. The latter
not only contains dual variables for vertices, but also for subsets of vertices. Such subsets with non-zero 
dual value 
are called blossoms. We can prove these blossoms are nested and so form a blossom tree. Moreover, for each blossom we need
to find a set of spanned edges and incident edges
that are all underrated. This cannot be done in such a simple way as  Lemma~\ref{lemma-secondside-bip}.
But again knowing the weights of maximum $f_v$-factors turns out to be  enough. The procedure that deduces all this information is highly nontrivial and is described in Section~\ref{fFactorSec}. This section is preceded with two special, simpler cases. First  Section~\ref{GeneralbMatchingSec} considers maximum $b$-matchings, where there are no underrated incident edges. Second Section~\ref{SPSec} considers shortest paths in undirected graphs with negative weights, where each blossom has exactly one underrated incident edge. Finally Section~\ref{sec:weights} generalizes Algorithm~\ref{algorithm:duals-simple-bip} to compute
$w(F_v)$ for non-bipartite graphs. Alternatively for the special case of shortest paths
$w(F_v)$ can be computed combinatorially as 
in Section~\ref{SPcombinatoric}.

%This cannot be done in such a simple way as in Lemma

%% file: bipartite-max-flow.tex
\subsection{Min-cost max-flow}
\label{FlowSec}
This section presents network flow algorithms. We discuss maximum flows and minimum cost maximum flows,
both in vertex-capacitated networks. 

We are given a directed network $N=(V,E)$, with source $s$ and sink
$t$, $s,t\in V$. For convenience let $V^-$ denote the set of nonterminals, $V-\{s,t\}$.
The 
edges and nonterminal
vertices 
have integral capacities given by 
$c:V^-\cup E\to [1..D]$. Let $g:V\times V\to \mathbb Z$ be a 
flow function. 
Besides the standard {\it edge capacity} and {\it flow conservation} constraints we have {\it vertex capacity}
constraints, i.e., for each vertex $v\ne s,t$ we require
\[
\sum_{u\in V} g(u,v) \le c(v).% \ \ \ \forall v\in V\setminus \{s,t\}.
\]

We begin by constructing a bipartite graph $G_N$ whose
maximum $f$-factor 
has weight equal to 
the value of a maximum flow in $N$. 
Wlog assume  that no edge enters $s$ or leaves $t$. The construction
proceeds as follows:

$\bullet$
 for each $v\in V^-$ place vertices $v_{in}, \, v_{out}$ in $G_N$;

$\bullet$ also place vertices $s_{out}, \, t_{in}$ in $G_N$;

$\bullet$
for each $v\in V^-$ add $c(v)$ copies of edge $v_{in}v_{out}$ to $G_N$;

$\bullet$
for each $(u,v)\in E$ add  $c(u,v)$ copies of edge $u_{out}v_{in}$ to $G_N$;

$\bullet$
for $v\in V^-$, set $f(v_{in})=f(v_{out})=c(v)$;

$\bullet$
set $f(s_{out}) = f(t_{in}) =c(V^-)$;

$\bullet$
set $w(e)=1$ for each edge $e$ leaving $s_{out}$ and $w(e)=0$ for every other edge of $G_N$;

$\bullet$ add $c(V^-)$ copies of edge $s_{out} t_{in}$ to $G_N$, all of weight 0.

\noindent
Note that in addition to the weight 0 edges $s_{out} t_{in}$,
$G_N$ may contain edges $s_{out} t_{in}$ of weight 1 corresponding to
an edge $st\in E$.

\begin{corollary}[\cite{gabow-tarjan-89}]
\label{corollary-max-flow}
Let $N$ be the flow network.
The weight of the maximum $f$-factor in $G_N$ is equal to the maximum flow value in $N$.
\end{corollary}
\begin{proof}
\begin{figure}
\centering
\includegraphics[width=0.7\textwidth]{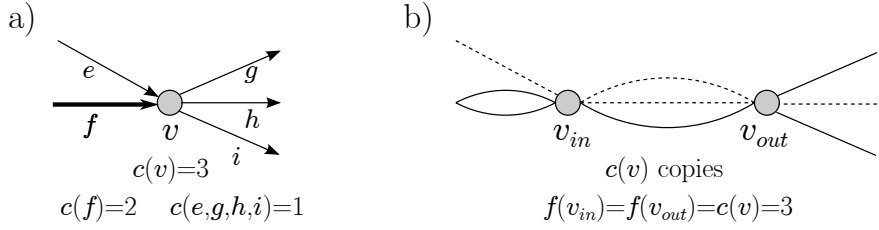}
\caption{A vertex $v$ of capacity $c(v)=3$ in $N$ is represented in $G_N$ by two vertices $v_{in}$ and $v_{out}$
connected by $3$ edges. The $f$-factor in (b) is marked with solid edges. Observe that the $f$-factor
must choose the same number of edges going in and out of $v$.}
\label{figure:maxflow}
\end{figure}
The main idea of the reduction is shown in Figure~\ref{figure:maxflow}. Observe that $f$-factors
in $G_N$ correspond to integral flows in $N$ that fulfill the flow conservation constraints. Moreover
the edge capacities are not exceeded since an edge cannot be used by an $f$-factor more times
than its capacity. Similarly a vertex cannot be used more times than its capacity. The only edges
with non-zero weights are edges incident to $s_{out}$, so the maximum $f$-factor maximizes the
amount of flow leaving $s$. Finally all this flow must wind up at $t$, since 
any $v\ne s,t$ has an equal number of $f$-edges incident to $v_{in}$ and $v_{out}$.
\end{proof}

Observe that  $G_N$ has $f(V)\le 4c(V^-)$. So the algorithm of Section~\ref{section:weighted-bipartite}
uses  $\tilde{O}(c(V^-)^{\omega})= \tilde{O}((Dn)^{\omega})$ time to find a maximum flow in a network.

Now assume that an edge $(u,v)$ of $N$ has a cost $a_{u,v} \in [-W..W]$, i.e., the cost of
sending $t$ units of 
flow on edge $(u,v)$ is linear and equals  $t a_{u,v}$. First find the maximum flow
value $f_{max}$ in $N$. Then modify the construction of $G_N$ to $G_{N,a}$ in the following way:

$\bullet$
add vertices $s_{in}$ and $t_{out}$;

$\bullet$
add $c(V^-)$ copies of edges $s_{in}s_{out}$ and $t_{in}t_{out}$;

$\bullet$
set $f(s_{in})= f(t_{out})=c(V^-)$ and
$f(s_{out})=f(t_{in})=c(V^-)+f_{max}$;

$\bullet$
for each copy of the edge $u_{out}v_{in}$ set $w(u_{out}v_{in})=-a_{u,w}$.

In Corollary~\ref{corollary-max-flow} we observed that $f$-factors in $G_N$ correspond to feasible flows in $N$. 
The new vertices  $s_{in}$ and $t_{out}$ allow 
an $f$-factor to model flow along cycles containing $s$ and $t$.
The new values of $f$ at the terminals make the corresponding flow from $s$ to $t$ have value $f_{max}$.
Thus $f$-factors in $G_{N,a}$ correspond
to maximum flows in $N$, and since costs of edges are negated between both networks, we have
the following.

\begin{corollary}[\cite{gabow-tarjan-89}]
\label{corollary-min-cost-max-flow}
Let $N$ be a flow network with linear edge costs.
A maximum $f$-factor in $G_{N,a}$ 
has
 weight 
equal to the minimum cost of a maximum flow in $N$.
\end{corollary}

Now the algorithm of Section~\ref{section:weighted-bipartite}
gives an $\tilde{O}(Wc(V^-)^{\omega})= \tilde{O}(W(Dn)^{\omega})$ time algorithm to find a min-cost max-flow.

Let us extend this reduction to convex cost functions. Assume 
the cost of sending $t$ units of flow on an edge $(u,v)$ is given by a convex function $a_{u,v}(t)$ such that marginal costs
satisfy $m_{u,v}(t) = a_{u,v}(t)-a_{u,v}(t-1)\in [-W..W]$. As usual
for this scenario~\cite{Ahuja93} we assume  $a_{u,v}$ is linear between successive
integers. This ensures that there exists an integral optimal solution.
We encode such cost functions in the graph by assigning different costs
to each copy of an edge, i.e., the $k$th copy of edge $u_{out}v_{in}$ has cost
$w(u_{out}v_{in},k)=-m_{u,v}(k)$ for $1\le k \le c(u,v)$. 

\begin{lemma}
\label{corollary-min-cost-max-flow-convex}
Let $N$ be a flow network with convex edge costs.
A maximum $f$-factor in $G_{N,a}$ 
has
 weight 
equal to the minimum cost of a maximum flow in $N$.
\end{lemma}
\begin{proof}
The maximum $f$-factor in $G_{N,a}$, when using $t$ copies of edges between $u_{out}$ and $v_{in}$, will
use the most expensive copies. The convexity of $a_{u,v}$ implies $m_{u,v}(t)$ is a non-decreasing
function. Hence we can assume that edges with costs $-m_{u,v}(1),\ldots,-m_{u,v}(t)$ are used.
These costs sum to $-a_{u,v}(t)$.
\end{proof}

Clearly this extension does not change the  running time of our algorithm. So
we find a min-cost max-flow for convex edge costs in time $\tilde{O}(W(Dn)^{\omega})$.

It is also easy to incorporate lower bounds on flow into the reduction.
Suppose that in addition to the upper bound function $c$ we have a lower bound function
$\ell:V^-\cup E\to [0..D]$; $\ell$ restricts the flow on edges and flow through 
nonterminal vertices in the obvious way.
To model $\ell$,
change the multplicity of $v_{in} v_{out}$ to $c(v)-\ell(v)$ and
the multiplicity of $u_{out} v_{in}$ to $c(u,v)-\ell(u,v)$;
in addition for $v\in V^-$, $f(v_{in})$ becomes $c(v)-\sum_u \ell(u,v)$ and
$f(v_{out})$ becomes $c(v)-\sum_u \ell(v,u)$.
It is easy to see that the correspondence between 
$f$-factors and feasible flows is maintained, as is the correspondence between weights and costs.
(In particular starting from an $f$-factor, adding 
$\ell(u,v)$ copies of each edge $u_{out} v_{in}$ gives a subgraph that obviously
corresponds to a flow satisfying all lower bounds on edges. It is easy to check this flow
also satisfies all lower bounds on vertices.)

The construction can also be generalized to bidirected flows. In directed
graphs an undirected edge can have two orientations, in-out and out-in, whereas a bidirected graph
allows four possible orientations in-in, out-in, in-out and out-out. (We also allow loops,
especially of type in-in and out-out.) A bidirected network $G$
gives a non-bipartite graph $G_N$. However as we show in the remainder of this
paper, the $f$-factor problem for non-bipartite graphs can be solved in the same time bounds
as the bipartite case. 

%% file: bduality.tex
\section{Weighted {\boldmath $b$}-matching}
\label{GeneralbMatchingSec}

This section
gives our algorithm to find a maximum $b$-matching.
$b:V\to \mathbb Z_+$ can be an arbitrary function. But we remark that
our algorithm is of interest
even for the case $b\equiv 1$ (ordinary matching): It achieves the same time bound
as \cite{CGSa},
and is arguably simpler in both derivation and algorithmic details.%
\footnote{The reader may enjoy working through the derivation for the matching case $b\equiv1$.
There will be no need to work with a multigraph. This endeavor will
 show how the two problems differ, especially in some details introduced by multigraphs.}

We start with two remarks that modify the definition of the problem. First it is most often convenient
to use the language of multigraphs. In this view we think of $G$ as a multigraph with
an unlimited number of copies of each edge.
Specifically each $e\in E$ has $1+\max_v b(v)$
copies, each with the same weight $w(e)$.
(Also recall from Section \ref{sec:definitions} that $G$ may have loops.)
A $b$-matching is a subgraph of this multigraph.

Second,
 our algorithm actually works on critical graphs, defined as follows.
Given a function $b:V\to \mathbb Z^+$,
for each  $v\in V$ define $b_v:V\to \mathbb Z_+$
by decreasing $b(v)$ by 1 (keep all other values  unchanged).
A graph is {\it $b$-critical} if it has a perfect  $b_v$-matching  for every
$v\in V$. Given a $b$-critical graph, our algorithm produces a
"blossom tree" from which, for any $v\in V$, a maximum $b_v$-matching can be easily extracted.

Such an algorithm
can find a maximum $b$-matching as follows.
Suppose we seek a maximum $b$-matching on $G$.
Assume every $b(v)$ is positive (discard 0-valued vertices).
Form $G'$ by adding
a vertex $s$ with $b(s)=1$, plus
edges $sv$, $v\in V(G)$ and loop $vv$, all of weight 0.
Any $v\in V+s$ has a $b_v$-matching.
(For $v=s$ take $F$,
an arbitrary $b$-matching on $G$.
For $v\ne s$,
for any edge $vw\in F$ take
$F-vw+ws$. This is a $b_v$-matching even if $vw$ is a loop.)
So $G'$ is $b$-critical and
a maximum $b_s$-matching on $G'$ is the desired $b$-matching.

%\ifcase\version%%%%%%%%%%%%%%%%%%%%%%%%%%%%%%%%%%%%%%%%%%
%Define $\zeta(uv) = w(F_u)+ w( F_v)+ w(uv)$
%(in contrast to Sec.~\ref{SPSec}).
%Use the shrinking procedure: Repeatedly contract an odd cycle of
%edges with maximum $\zeta$ value. Development given in the full version.
%\or

\subsection{The heaviest blossom}
This section defines the $\zeta$-value of an edge and shows
how it reveals the heaviest blossom (the set $B$ defined after Lemma
\ref{b2.2}).
Consider a weighted $b$-critical graph $G$.
Recall that $G$ may contain self-loops $vv$.
All of what follows is valid even when there are such loops.
Assume that each multiset of $\le \sum_v b(v)$ edges has a distinct weight,
if we do not distinguish between parallel copies of the same edge.
We can enforce this assumption by
taking a very small $\epsilon>0$ and
adding $\epsilon^i$ to the weight of every copy of the $i$th edge of $E$.
(Section \ref{EffBMAlgSec}
returns to the original unperturbed weights.)

Any vertex $v$ has a maximum $b_v$-matching $F_v$.
$F_v$ is unique up to parallel copies of the same edge, by the perturbed weight function.
Wlog
assume further that any two matchings $F_u$ and $F_v$ have as many common edges as possible (i.e., for any $xy\in E$
they use as many of the same copies of $xy$ as possible).

We start with a well-known principle.
Take two vertices $u,v$.
Call a trail
{\em alternating} (for $u$ and $v$) if as we traverse it
the  edges alternate between $F_u-F_v$ and $F_v-F_u$.

\begin{lemma}
\label{DSumLemma}
For any two vertices $u,v$, $F_u\oplus F_v$ is
an alternating $uv$-trail that starts with an edge of
$F_v-F_u$ and ends with an edge of $F_u-F_v$.
\end{lemma}

\begin{proof}
First we show that  $F_u\oplus F_v$  contains an alternating trail as described in the lemma.
(Then we show that trail constitutes all of $F_u\oplus F_v$.)
%Suppose not.
Let $T$ be
a maximal length alternating trail that starts at $u$ with an edge of
$F_v-F_u$. Such an edge exists since $u$ has greater degree in $F_v$ than $F_u$.
Let $T$ end at vertex $x$.

If $x\ne u,v$ then, since $d(x,F_u)=d(x,F_v)$,
we can extend $T$ with an unused alternating edge.
(This is true even if $T$ has a previous occurrence of $x$.)
%The assumption that $F_u$ and $F_v$ have as many common edges as possible
%is used in establishing that $T$ remains a trail.)
If $x=u$ then
we can extend $T$ with an unused alternating edge -- the argument is the same using the facts that
$d(u,F_v)=d(u,F_u)+1$ and the first edge of $T$ is in $F_v-F_u$.
Suppose $x=v$. If the last edge of $T$ is in $F_v-F_u$ we can extend $T$ with an unused edge of
$F_u-F_v$  since $d(v,F_u)=d(v,F_v)+1$.
The remaining possibility is that $T$ ends at $v$ with an edge of $F_u-F_v$, thus giving the desired trail.

Now we show there are no other edges.
$F_v$ is the disjoint union of its edges in $T$ and a multiset of edges $R$. ($R$ may contain copies
of edges in $T$.) We claim
any vertex $x$ satisfies
 \begin{eqnarray}
\label{REqna}
d(x,R)&=&d(x,F_v)-d(x,T\cap F_v)\\
&=&b(x)- \c{d(x,T)/2}.\label{REqnb}
\end{eqnarray}
Observe that this claim completes the proof of the lemma: Define the function $b'$ by setting
$b'(x)$ to the quantity of \eqref{REqnb}.
Thus $R$ is a $b'$-matching, in fact a
maximum $b'$-matching.
By symmetry
$F_u$ is the disjoint union its edges in  $T$ and a maximum
$b'$-matching. The perturbed weight function implies
that $b'$-matching is also $R$.
Thus $F_u$ and $F_v$ agree outside of $T$, so $F_u\oplus F_v=T$.

To prove the claim, \eqref{REqna} holds by definition.
For \eqref{REqnb} first observe that
$x\ne v$ implies
\[d(x,F_v)=b(x) \mbox{ and  } d(x,T\cap F_v)=\c{d(x,T)/2}.\]
In proof the first relation follows from definition of $F_v$.
For $x\ne u,v$  the second relation
holds because the edges through $x$
alternate.
For $x=u$ the second relation holds because
the first edge of $T$ is in $F_v$, and
all other pairs of edges through $u$ alternate.
Substituting the two relations into \eqref{REqna} gives \eqref{REqnb}.

If $x=v$ then
$d(v,F_v)=b(v)-1$ and  $d(v,T\cap F_v)=(d(v,T)-1)/2$. The latter holds because
the last edge of $T$ is not in $F_v$. Furthermore all other
pairs through $v$ alternate. Again substituting these two relations into \eqref{REqna} gives \eqref{REqnb}.
\end{proof}

In an arbitrary multigraph let $C$ be an odd circuit containing a
vertex $u$. Choose a traversal of $C$ that starts at $u$.
Define $C_u$ to consist of alternate edges in
this traversal,
omitting the first edge at $u$ as well as the
last.
When $C$ contains $> 2$ edges at $u$, $C_u$ will not be unique.
But suppose the traversal starts with edge $uv$, and we define $C_v$ by the same traversal
only starting at $v$. Then
 $C_u$, $C_v$ and $uv$ form a partition of $C$ (since $C_u$ starts by containing the edges
of $C-C_v$, and this pattern continues until the traversal reaches $v$).
Also given any choice of a $C_x$,
define $C^x$ to consist of alternate edges of $C$,
beginning and ending with the edge incident to $x$.  Clearly $C_x$ and
$C^x$ form a partition of the edges of $C$.

To define the central concept,
for any edge $uv$ let
\[\zeta(uv) = w(F_u)+ w( F_v)+ w(uv).\]
We shall see that $\zeta$ gives the values of the optimum blossom duals as well as the structure of the optimum blossoms.

\begin{lemma}
\label{b2.1}
Any edge $e$ of a  $b$-critical graph %$G$
belongs to an odd circuit  of edges with
$\zeta$-value  $\ge \zeta(e)$.
\end{lemma}

\begin{proof}
Let $e=uv$. Assume $u\ne v$ else the lemma is trivial.
Furthermore assume $e\notin F_u\cup F_v$.
This is the crucial assumption! It is justified since there are $b(v)+1$ copies of $e$; furthermore
proving the lemma for this copy of $uv$ proves it for every copy.

Let $T$ be the trail of Lemma \ref{DSumLemma}.
It clearly has even length. Extend it by
adding a copy of edge
$e$.
($T$ may already contain a different copy of $e$, one in $F_u\oplus F_v$.)
We get an odd circuit $C$.
Traverse $C$ by starting with the first edge of $T$ and ending with the edge $uv\in C-T$. We get
$C$ partitioned into $C_u$, $C_v$ and $uv$.
Furthermore
$C_u=C\cap F_u$, $C_v=C\cap F_v$, and a set of edges $R$ satisfies
$R=F_u-C=F_v-C$. Thus
\[\zeta(uv)=w(F_u) + w(F_v)+w(uv)=w(C)+2w(R).
\]

Now take any edge $rs\in C$. Traverse $C$ by starting with edge $rs$.
For $t \in \{r,s\} $ let $H_{t}$
be the multiset $R\cup C_t$.
It is easy to see $H_t$ is a $b_t$-matching
by comparing it with $F_v=R\cup C_v$.
Then
\[\zeta(rs) \ge
w(H_{r}) + w( H_{s}) +w(rs) = w(C)+2w(R).\]
The two displayed equations show
$C$ is the desired circuit. (In this argument some vertices $t$ may have $C_t$ multiply defined. That's OK.)
\end{proof}

Let $\zeta^*$ be the maximum value of $\zeta$.
Let $E^*$ be the set of edges of $\zeta$-value $\zeta^*$.
We shall see that $E^*$ is essentially the heaviest blossom
and $\zeta^*$ its dual value.

\begin{lemma}
\label{b2.2}
Any edge $e=vw\in E^*$
belongs to an odd circuit $C\con E^*$.
Furthermore \hbox{\rm (i)} $F_v-C=F_w-C$ and \hbox{\rm (ii)}
$F_v\cap (\gamma(v)\cap \delta(v))\con E^*$.
\end{lemma}

\begin{proof}
Lemma \ref{b2.1} shows the odd circuit $C$ exists.
To prove (i) recall from the proof that
\[w(C)+2w(R)=\zeta^*\]
and  $F_v= R\cup C_v$,
$F_w= R\cup C_w$.
%(Even though the initial definition of $C_v$ is
The last two equations imply (i).

For (ii)
take any edge $uv\in F_v$. Keep $C$ and $R$ as already defined for $vw$.
We can assume $uv\in R$, since otherwise $uv\in C_v\con C\con E^*$.
Let $H_u=R -uv+C^v$. It is easy to see $H_u$ is a $b_u$-matching
by comparing it with $F_v=R\cup C_v$ (note $u$ may or may not belong to $C$).
Then
\[\zeta(uv) \ge
w(H_{u}) + w( F_{v}) +w(uv) = w(C)+2w(R) =\zeta^*.\]
This implies equality holds
and proves (ii).
 \end{proof}

Call a connected component $B$ of $E^*$ {\em nontrivial} if it
spans at least one $\zeta^*$-edge.%
\footnote{For ordinary matching $B$ is a cycle that comprises all of $E^*$.}
($B$ may consist of a single vertex $v$ with one or more loops $vv$.)
The next lemma shows that $B$ behaves like a blossom, i.e.,
it can be shrunk to a single vertex. We begin the proof with two observations.
First, any two vertices  $v,w\in B$ have
\begin{equation}
\label{BEqn}
F_v-E^*(B)=F_w-E^*(B).
\end{equation}
In proof,
since $B$ is connected
we need only show (\ref{BEqn}) when $vw\in E^*(B)$.
That case follows by applying
Lemma \ref{b2.2}(i) with $C\con E^*(B)$.

Next observe for any $v\in B$,
\begin{equation}
\label{BFE*Eqn}
F_v\cap (\gamma(B)\cup \delta(B)) \con E^*(B).
\end{equation}
This follows since an edge $xy$ in the left set
but not in the right, with $x\in B$,
belongs to $F_v-E^*(B)=F_x-E^*(B)$
by
(\ref{BEqn}).
But Lemma \ref{b2.2}(ii) shows
$xy \in F_x\cap\delta(x) \con E^*$. This implies
$y\in B$ and
$xy\in E^*(B)$, contradiction.

\iffalse
For $C\con V$ let $b|_C$ denote the restriction of $b$ to $C$.
X is critical for $b|_X$. %-matching critical graph
\fi

\begin{lemma}
\label{bMatchingRespectLemma}
 For any vertex $v$,
$|F_v\cap \delta(B)|$ equals 0 for $v\in B$ and 1 for $v\notin B$.
%For any  $v\in V$, $F_v$ respects $V(B)$.
\end{lemma}

\begin{proof}
For $v\in B$ this follows from (\ref{BFE*Eqn}).
So suppose
$v \notin B$. Choose any vertex $t \in B$.
Lemma \ref{DSumLemma} shows
$F_v \oplus  F_t$ is
an alternating $vt$-trail.
Let $T$ be the subtrail from $v$ to the first vertex of $B$, say  vertex
 $x$, and let $e$ be the last edge of $T$.
$e\in F_v$ ($e\notin F_t$
by the lemma for $t$, i.e.,  $F_t\cap \delta(B)=\emptyset$).
Applying (\ref{BEqn})
to $t$ and $x$ shows
$F_v \oplus F_{ x}$ contains $T$.
So Lemma \ref{DSumLemma} applied to $v$ and $x$  shows
$F_v \oplus F_{ x}=T$. (More precisely, the first part of the proof of Lemma 1
could have chosen the current $T$ as its $vx$ trail. The second part of the proof shows
$F_v \oplus F_{ x}$ has
no
other edges.)
Thus $e$ is the unique edge of $F_v$ incident to $B$.
\end{proof}

\subsection{Iterating the construction}
The last lemma  generalizes \cite[Lemma 3.2]{G} for ordinary matching.
The rest of the derivation closely parallels ordinary matching, as described in
\cite{G} and \cite{CGSa}. In a nutshell,
the optimum blossoms and their duals are found
by running the shrinking procedure of \cite{G}.
The nesting of the contracted blossoms gives a "blossom tree" \B..
Any desired maximum $b_v$-matching ($v\in V(G)$) can be constructed by a top-down traversal
of \B. that finds a perfect $b'$-matching (for appropriate $b'$) at each node of \B..
This last step is the biggest difference from ordinary matching: In ordinary matching
the edges found at each node of \B. are simply alternate edges of a cycle.
For completeness we present all of these remaining details, of course modified
for $b$-matching.

\iffalse
The $z$ function is defined in the same way as matching
(see \cite{G}, or simply
modify \eqref{DualDefnEqn}
replacing $w(U(B))$ by the value of $\zeta^*$ when $B$ is created).
\fi

Consider a nontrivial connected component $B$ of $E^*$, as above.
Contract $B$ to a vertex \b B., with \b G. the resulting multigraph.
Assume  a contraction operation
can create parallel edges but not loops.
(Thus the graph changes even when $B$ is a single vertex $v$ with $E^*$ consisting of loops $vv$.)
Extend the degree-constaint function $b$ to \b G. by setting $b(\b B.)=1$.
Lemma
\ref{bMatchingRespectLemma}
shows \b G. is $b$-critical. From now on $V$ designates the vertex set of the original given graph.

Next
we  define a weight function on \b G..
% so the maximum $b_v$-matchings don't change.
For convenience we
designate edges of \b G. by their corresponding edge in $G$
(i.e., an edge of \b G. is written as $uv$ where $u,v\in V$ and possibly
one of them belongs to $B$).
For $v \in B$ let
\begin{equation*}
%\label{BxEqn}
B_v=F_v \cap \gamma(B).
\end{equation*}
For edge $uv$ in \b G. define a weight $\b w.(uv)$ by
\[
\begin{array}{lcll}
{\b w.}(uv) &= &w(uv)&u,v \notin B,\\
&= &w(uv) + w(B_v)  &v \in B.
\end{array}
\]

This definition preserves the
 structure of $G$ in the following sense.
For any vertex $v$ let \b v. be its image in \b G..
Let $F_{\b v.} $ denote the (unique) maximum  $b_{\b v.}$-matching
in \b G..
For a fixed vertex $v$, let vertex $x\in B$ be  $v$ if $v\in B$, else
the end of an edge of $F_v$.
$x$ is uniquely defined by Lemma \ref{bMatchingRespectLemma}.

\begin{lemma}
\label{wBarLemma}
 \i $F_v = F_{\b v.} \cup B_x$.

\ii Any edge of \b G. has  the same
 $\zeta$-value
in $G$ and \b G..
\end{lemma}

\begin{proof}
\i
Let $H_{\b v.}$ be the image of $F_v$ in \b G..
Lemma \ref{bMatchingRespectLemma} shows
$H_{\b v.}$  is a  $b_{\b v.}$-matching.
 Furthermore $F_v = H_{\b v.} \cup B_x$ (for $v\notin B$ this follows from
the optimality of $F_v$). So we must show
$H_{\b v.}= F_{\b v.}$. Clearly it suffices to show
\begin{equation}
\label{wBarEqn}
 \b w.( F_{\b v.})\le  \b w.(H_{\b v.}).
\end{equation}

Suppose $v\in B$.  $F_{\b v.}\cup B_v$ is a $b_v$-matching. Since it
weighs no more than $F_v$, $ w(F_{\b v.})\le w(H_{\b v.})$.  This is
equivalent to (\ref{wBarEqn}) since neither set contains an edge
incident to $B$.

Suppose $v\notin B$.  The definition of \b w. shows $w(F_v)=\b
w. (H_{\b v.})$, as well as $\b w. (F_{\b v.})\le w(F_v)$ (by the
optimality of $F_v$).  These relations combine to give
(\ref{wBarEqn}).

\ii
Part  \i shows $w(F_v)$
is $\b w.(F_{\b v.})$ for $v\notin B$
and $\b w.(F_{\b v.}) -w(B_v)$ for $v\in B$.
The two cases of
\ii (depending on whether or not the edge is in $\delta(B)$) follow.
\end{proof}

The following
{\em shrinking procedure} \cite{G}
iterates the construction of \b G..
It also constructs
a tree  that represents the nesting of the contracted blossoms:

\bigskip

{\narrower

{\parindent=0pt

Start by creating a one-node tree for each vertex of $V$.
Then repeat the following step
until the graph consists of one vertex:

\bigskip

Let $B$ be a nontrivial connected component of $E^*$.
Form a tree whose root represents $B$; the subtrees of the root
are the trees that represent the vertices of $B$ in the current graph.
Then change the current graph by contracting $B$ and defining $b$ and $\o w$
as described above for \b G..

}}

\bigskip

The construction can be iterated since as noted each \b G.  is $b$-critical.
Lemma \ref{wBarLemma}\ii shows
$\zeta^*$ never increases from one iteration to the next.

Call the final tree the {\em blossom tree} \B. of $G$, and each nonleaf
a {\em blossom} of $\B.$. %(of \B. or of $G$).
A child $C$ of $B$ is either a {\em blossom-child} or a {\em vertex-child}
(i.e., a leaf of \B.).
Note that a vertex of $V$ can occur in \B. as a singleton blossom as well as
a leaf.
For any node $B$ of \B. let $V(B)$ denote the set of leaf descendants of $B$.

The next goal is to describe how the edges of any $F_v$, $v\in V$, are distributed
in \B.. We begin with several definitions.
Consider the iteration in the  shrinking procedure that ends by contracting $B$.
Let $G_B$ be the graph at the start of the iteration.
Thus each child $C$ of $B$ corresponds to a vertex of $B$ in the graph $G_B$.
%$C$ is either a blossom-child
%formed by contracting the vertex set $V(C)$, or a vertex-child.
Let $\zeta(B)$ be the value of $\zeta^*$ in this iteration.
Let $E^*(B)$ be the corresponding set of edges  (i.e.,
the edges of $\zeta$-value $\zeta(B)$ that join two vertices of $B$ in $G_B$).

We now generalize Lemma \ref{wBarLemma}\i to any graph $G_B$.
Consider
any vertex $v\in V(B)$. As before let
$\b v.$ be the image of $v$ in $G_B$ (so \b v. is
a child of $B$).
$F_{\b v.}$ denotes the unique maximum $b_{\b v.}$-matching of $G_B$.
For each child $C$ of $B$
let $x\in V(C)$ be the vertex that is either $v$ (if $C={\b v.}$) or
is the end of an edge of $F_{\b v.}$.
$x$ is uniquely defined since a blossom $C$ has $b(C)=1$, and
there is no choice for $x$ if $C$ is a vertex of $V$.
We claim

\begin{equation}
\label{FvDefnByBEqn}
F_v\cap \gamma(V(B))=
F_{\b v.}  \cap \gamma(V(B))
\;\cup\; \bigcup_{C\text{ a child of }B} F_x \cap \gamma(V(C)).
\end{equation}

The import of \eqref{FvDefnByBEqn} is that it
defines the entire set $F_v$. Specifically
let \V. be the root node of \B.. Clearly  any $v\in V$ has
$F_v=F_v\cap \gamma(V(\V.))$. So
the entire set $F_v$ is defined by
applying \eqref{FvDefnByBEqn}
to \V. and then
 recursively to the children of \V..

Also note that in  \eqref{FvDefnByBEqn} if $C$ is a vertex-child of $B$ then the expression $F_x\cap \gamma(V(C))$ is empty, by convention.

To prove \eqref{FvDefnByBEqn} first observe
$F_v \cap E(G_B) =F_{\b v.}$. This
follows by repeated applications of
Lemma \ref{wBarLemma}\xi, which shows the edges of $F_v$ in \b G.
form $F_{\b v.}$.

The observation justifies the first term
$F_{\b v.}  \cap \gamma(V(B))  $
in  \eqref{FvDefnByBEqn}.
To complete the proof note that the remaining edges
of $F_v$ are contained in the various blossom-children $C$ of $B$.
We show these edges are given by the terms
$F_x\cap \gamma(V(C))$ in  \eqref{FvDefnByBEqn}.
%recall $F_v$ is a maximum $b_v$-matching.
If $C= \b v.$ this is obvious since $x=v$.
If $C\ne \b v.$ then
$F_{\b v.}\cap \delta(C)$ is a unique edge incident to $x$.
The optimality of $F_v$ implies it agrees with $F_x$ in $\gamma(V(C))$.
\eqref{FvDefnByBEqn} follows.

\eqref{BFE*Eqn} shows that in the graph $G_B$,
$F_{\b v.}\cap \gamma(B)\con E^*(B)$.
Applying this recursively with \eqref{FvDefnByBEqn} shows
$F_v$ consists entirely of $\zeta^*$-edges, more precisely,
in any blossom $B$, the edges of $F_v$ that join children of $B$
(or that are loops incident to a vertex that is a blossom-child of $B$) are in $\zeta(B)$.

\subsection{The efficient ${\boldmath b}$-matching algorithm}
\label{EffBMAlgSec}
An efficient  algorithm
cannot work with the perturbed weight function.
So suppose we execute the  shrinking procedure, starting with the original unperturbed
weight function $w$.
Each iteration will find a connected component $B$ that was found in the construction of \B..
More precisely an iteration that contracts blossom $C$
of \B. will be skipped iff
the parent of $C$ in \B., say $B$,
has $\f {\zeta(C)}=\f {\zeta(B)}$.
So the shrinking procedure will construct a tree
\W. that is a contraction of \B. (more precisely
an edge of \B. from parent $B$ to child $C$ is contracted iff it
satisfies the above relation $\f {\zeta(C)}=\f {\zeta(B)}$).

For any blossom $B$ of \W., define the graph $G^*_B$ as follows.
Consider the iteration of the shrinking procedure that contracts $B$.
The vertices of $G^*_B$ are
the vertices of the current graph that are
contained in $B$. (These vertices are the children of $B$ in \W..)
The edges of $G^*_B$ are the edges $E^*(B)$ in the shrinking procedure.

The following observation is key to our algorithm.
It enables us to find the
edges of any desired maximum $b_v$-matching that occur in each graph
$G^*_B$. Take any vertex $v\in V(B)$. Let \b v. be
the image of $v$ in $G^*_B$.
The edge set $F_v\cap E^*(B)$ is a $b_{\b v.}$-matching of $G^*_B$.
In proof first recall that $F_v$ consists entirely of $\zeta^*$-edges.
Then
apply
\eqref{FvDefnByBEqn} to $B$ and to each descendant of $B$ (in \B.)
that gets contracted into $B$ when we form \W. from \B..

We find a maximum $b_v$-matching using
 a recursive procedure
$b\_match$
that finds the desired edges in each
graph $G^*_B$.
 For a blossom $B$ of \W.
and a vertex $v\in V(B)$,
$b\_match(B,v)$
 finds a
 $b_v$-matching of $\zeta^*$-edges
in the subgraph of $G$ induced by vertices $V(B)$.
It works as follows.

Let \b v. be the child of $B$ that contains $v$.
First find a $b_{\b v.}$-matching of $G^*_B$, say $H_{\b v.}$. ($H_{\b v.}$ exists
by the above key observation.) Add $H_{\b v.}$ to the desired set.
Then  complete the desired set using recursive calls
on the children of $B$. Specifically
for each blossom-child $C$ of $B$ (in \W.),
execute $b\_match(C,x)$
where  $x\in V(C)$ is $v$ (if $C={\b v.}$) or
the end of an edge of $H_{\b v.}$.

As before let  \V. be the root of \W..
Let $H_v$ be the
$b_v$-matching of
the given graph $G$
that is found by $b\_match(\V.,v)$.
In general $H_v$ is not $F_v$.
For instance in a  recursive call $b\_match(C,x)$,
$x$ may differ from the vertex $x$ given by
\eqref{FvDefnByBEqn}. But
because $H_v$ consists of $\zeta^*$-edges,
we can prove $H_v$ is optimum using
duality, as follows.

\def\hW.{\mathy{{\cal W}^-}}

As before
for any blossom $B$ of \W. let $\zeta(B)$
be the value $\zeta^*$ when $B$ is created.
(Now $\zeta(B)$ is an integer.)
For any  blossom $B\ne \V.$  let $p(B)$ be its parent in \W..
Let $\hW.$ be the set of all blossoms of \W..
Define functions $y:V \to \mathbb {Z}$,
$z:\hW. \to \mathbb {Z}$
by
\[
\begin{array}{lcll}
y(v)&=&-w(F_ v)\hspace{58pt}v \in V,\\
%\end{array}\\
%\]
%\[
%z(G)&=&\zeta(G)\\
%z(B) &= &\zeta(B)-\zeta(p(B))&B \mbox{ a blossom }\ne G.
%\begin{array}{lcll}
z(B) &= &
\left\{
\begin{array}{ll}
\zeta(\V.) &B=\V.,\\
\zeta(B)-\zeta(p(B))&B \in \hW.-\V..
%0&B\notin \hW..
%\ne V \text{ or any vertex }v,\\
\end{array}
\right.
\end{array}
\]
%(In the definition of $z$ we have identified a blossom $B$
%with its corresponding set $V(B)$.)
Recall our convention for summing functions (Section \ref{sec:definitions}),
e.g., an edge $e=uv$ has $y(e)=y(u)+y(v)$.
Any blossom $B$ has $\zeta(B) =  z\set{A} {V(B)\con V(A)}$.
The definition of $\zeta$ shows any edge $e$ of $G$ has
$w(e)=y(e) + \zeta(e)$.
So $e\in E^*(B)$ implies
\begin{equation}
\label{eWeightEqn}
w(e) = y(e) + \zeta(B) =
y(e) + z\set {A} {e\con V(A)}.
\end{equation}

\iffalse
$z$ is nonnegative except perhaps on $\V.$
(recall $\zeta^*$ never
increases). Clearly
any blossom $B$ has $\zeta(B) =  z\set{C} {V(B)\con V(C)}$.

Let $e$ be an edge of $G$.
The definition of $y$ implies that
$w(e)=y(e) + \zeta(e).$
Let $B$ be the first blossom created with
$e\con V(B)$. Hence $\zeta(e) \le  \zeta(B)$. This implies
dominance, since
\begin{equation}
\label{eWeightEqn}
w(e) \le  y(e) + \zeta(B) =
y(e) + z\set {C} {e\con C}.
\end{equation}
If $e$ is an edge of a blossom subgraph
the above holds
with equality so $e$
is tight.
\fi

Let $F$ be either $F_v$ or $H_v$. Both sets $F$ weigh
\iffalse
\begin{eqnarray*}
\label{fFactorDualEqn}
%w(F) &=&
\sum \set{w(e)} {e\in F}
&=& \sum \set{y(e) + z\set {A} {e\con V(A),\, A\in \hW.}}      {e\in F} \\
&=&  (b_v y)(V) +\sum \set{z(A)} {e\in F,\, e \con V(A) } \\
& = & (b_v y)(V) + %
\sum \set{ \frac { b_v (V(A)) -1 }  {2} z(A)}  {A\in \hW.}
\end{eqnarray*}
\fi
\begin{eqnarray*}
\label{fFactorDualEqn}
%w(F) &=&
\sum _{e\in F} w(e)
&=& \sum _{e\in F} y(e) + z\set {A} {e\con V(A)}\\
&=&  (b_v y)(V) +\sum_{e\in F} \set{z(A)} {e \con V(A)} \\
& = & (b_v y)(V) +
\sum_{A\in \hW.}   \frac { b(V(A)) -1 } {2}  z(A).
\end{eqnarray*}
For the last line recall that for both sets $F$,
$F\cap \gamma(V(A))$ is a
$b_x$-matching for some $x\in V(A)$,
so
$|F\cap \gamma(V(A))|=\frac {b(V(A))-1}{2}$.
We conclude
$w(H_v)=w(F_v)$, i.e.,
$b\_match$  constructs
 a maximum $b_v$-factor.

To show the algorithm is Las Vegas we only need to observe that
the functions $y,z$
fulfill all the requirements to be
optimum linear programming duals.
This follows from two more properties.
First, the function $z$ is nonnegative except perhaps on $\V.$.
This follows since $\zeta^*$ never increases.
Second any edge $e$ of $G$ satisfies
$w(e) \le y(e) + z\set {A} {e\con V(A)}$.
To see this first note that $\zeta(e) \le  \zeta(B)$ for
 $B$ be the first blossom created with $e\con V(B)$.
Use this relation to compute $w(e)$ as in \eqref{eWeightEqn}.

These two properties make $y,z$ optimum
linear programming duals. So if the properties are satisfied,
and the algorithm finds a $b_v$-matching composed of edges that satisfy \eqref{eWeightEqn},
that matching has maximum weight. (This is easily verified without appealing
to linear programming,
by computing the weight
of an arbitrary $b_v$-matching $F$ similar to the computation of $w(F)$ above.)

To summarize
the algorithm of this section works as follows.
Consider a $b$-critical graph $G$,
 with edge weights $w(e)$.
Assume we are given the weight of every maximum $b_v$-matching, i.e.,
every  value $w(F_v)$, $v\in V$.
Start by
executing the shrinking procedure to  construct the tree \W..
For any vertex $v\in V$,
to find a maximum $b_v$-matching call  $b\_match(\V.,v)$.

The total running time for this procedure is  $O(\phi^\omega)$.
To prove this
we will show
the shrinking procedure
uses $O(m \log n)$ time
and $b\_match$ uses
$O(\phi^\omega)$ time, thus giving the desired bound.
First note
that \W. has $\le 3n$ nodes
($n$ leaves, $n$ blossoms that are singletons, and $n$
larger blossoms).

The shrinking procedure starts by
using the given values $w(F_v)$ to compute the $\zeta$-value of each edge.
Then it sorts the edges on decreasing $\zeta$-value, in time $O(m\log n)$.

We use a set-merging algorithm \cite{CLRS}
to keep track of the contracted vertices.
That is, for any $v\in V$, $find(v)$ gives the contracted vertex
currently containing $v$.
%; we denote this vertex as $\bar v$.
The operation $union(A,B,C)$ merges two contracted
vertices $A$ and $B$ into a new vertex  $C$.

The iteration of the shrinking procedure
that creates $B$
starts by constructing an adjacency structure for the $\zeta^*$-edges
of graph $G_B$.
(An edge $uv$ ($u,v\in V$) joins the vertices
given by
{\tt find}$(u)$, {\tt find}$(v)$.)
Then it finds the nontrivial connected components.
Each edge of $G$ is in at most one graph $G_B$. So the total time
is $O(m)$ plus the time for $n$ {\tt unions} and $O(m)$ {\tt finds}
\cite{CLRS}.

 $b\_match$
runs in
the time to find $b_{\b v.}$-matchings on all the  graphs $G^*_B$.
These graphs contain a total of
$\le 3n$ vertices,
 $\le m$ edges,
and degree constraints $b$ totalling $\le 3\phi$
(the vertices of
$V$ contribute exactly $\phi$
and each
 blossom contributes 1, giving
 $\le 2n+\phi\le 3\phi$). So using an algorithm that finds a $b$-factor in time $O(\phi^\omega)$
gives total running time $O(\phi^\omega)$.
Since the algorithm is randomized there is a
minor point of controlling the error probability when the graphs
get small, this is easily handled.

%%%%%%%%%%%%%%%%%%%%%%%%%%

%% file: shtstpth.tex
\section{Shortest paths}
%TODO omega coincides with omega from matrix multiplication
\label{SPSec}
\def\gm.{\mathy{\o G\,^-}}
\def\gp.{\mathy{\o G\,^+}}
\def\wm.{\mathy{\o w\,^-}}
\def\wp.{\mathy{\o w\,^+}}

This section discusses the single-source shortest-path problem on
conservative undirected graphs.  It defines a ``shortest-path tree''
structure for conservative graphs and proves this structure always
exists.  The structure is a special case of the one for general
$f$-factors. It is used as an example in Sec.\ref{BackgroundSec}.

\subsection{The shortest-path structure}
\label{SPDefnSec}
 Let $(G,t,w)$ denote a connected undirected graph with
distinguished vertex $t$ and conservative
edge-weight function $w:E\to \mathbb {R}$.
%$(G,t,w,W)$ is defined similarly except $w:E\to [-W..W]$.
We wish  to find a shortest path from each
vertex to the fixed sink vertex $t$.

Bellman's inequalities needn't hold
and a shortest-path tree needn't exist
(e.g., the subgraph on $\{a.b.c\}$
in Fig.\ref{SPStrucFig}).  Nonetheless
the distance variables of Bellman's inequalities are optimum for
a related set of inequalities.

\iffalse
\begin{figure}[t]
\begin{center}
\epsfig{figure=triangle.eps, width=1.25in}
%\hw}}
\end{center}
\caption{Shortest paths to $t$ do not form a shortest path tree.}
\label{TriangleFig}
\end{figure}
\fi

We begin by defining the analog of the shortest-path tree.  When there
are no negative edges this analog is a variant of the standard
shortest-path tree (node \V. below).  Figs.\ref{SPStrucFig}--\ref{SPStrucbFig}
illustrate the definition. In Fig.\ref{SPStrucFig}
the arrow from each vertex $v$
gives the first edge in $v$'s shortest path. This edge is $e(v)$
in the definition below. More generally this edge is  $e(N)$ for any
node $N$ that it leaves, e.g.,
$ce= e(c)=e(\{a,b,c\})$.

\begin{figure}[t]
\begin{center}
\epsfig{figure=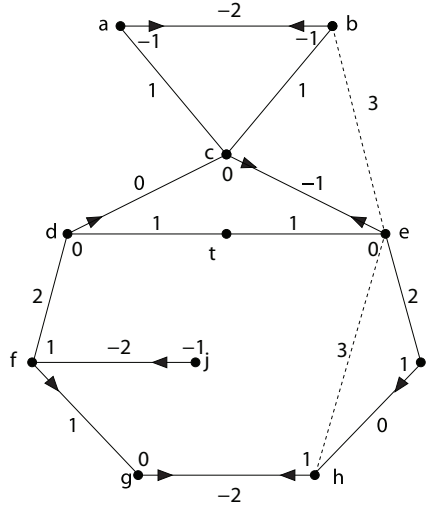, width=2.25in}
\end{center}
%\caption{Shortest path structure.}
\caption{Conservative undirected graph. Vertex labels are
 shortest-path distances; arrows
show the first edge of shortest paths. Dashed edges are not in any shortest path.}
\label{SPStrucFig}
\end{figure}

\begin{figure}[t]
\begin{center}
\epsfig{figure=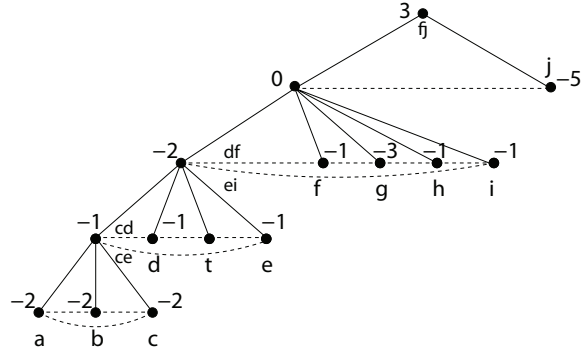, width=3.0in}
\end{center}
\caption{Shortest-path structure.
Node labels are $z$ values. $E(N)$ edges are the dashed edges
joining the children of $N$.}
\label{SPStrucbFig}
\end{figure}

\begin{definition}
\label{SPTreeDfn}
A {\em generalized shortest-path tree (gsp-tree)} \T.  is a tree
whose leaves correspond to the vertices of $G$. For each node
$N$ of \T., $V(N)$ denotes the set of leaf descendants of $N$,
$V(N)\con V(G)$.
Let \V. be the root of \T..
For each node $N$,
$V(N)$ contains a
{\em sink} vertex denoted  $t(N)$; for
$N\ne \V.$, $t(N)$ is the end of an edge
$e(N)\in
\delta(V(N))$.
For $N=\V.$ the sink is $t$, and we take
$e(\V.)=\emptyset$; for  $N\ne \V.$, $t(N)$ and $e(N)$ are
determined by the parent of $N$ as described below.

Consider
an interior node $N$ of \T., with children $N_i$, $i=1,\ldots,k $, $k\ge 2$.
$V(N_1)$ contains $t(N)$ and
$t(N_1)=t(N)$, $e(N_1)=e(N)$.
%For $i>1$, $\delta(V(N_i))$ contains an edge $e_i$and $t(N_i)=e_i \cap V(N_i)$.
$N$ has an associated set of edges
 $E(N)$ with
 $\set {e(N_i)} {1<i\le k} \con E(N)\con \gamma(N)$.
Let $\o N_i$ denote the contraction of
$V(N_i)$ in $G$.

\case {$N\ne \V.$}
$E(N)$ forms
a (spanning) cycle on the vertices $\o N_i,\ i=1,\ldots, k$.

\case {$N= \V.$}
Either \i $E(N)$ gives a cycle
exactly as  in the previous case, or \ii
$E(N)$ is a spanning tree
on the nodes $\o N_i$, rooted at $N_1$,
with each $e(N_i)$  the edge from $N_i$ to its parent.
\end{definition}

Note that
 $\set {e(N_i)} {1<i\le k} =E(N)$
in case \ii above. In contrast
when $E(N)$ is a cycle, fewer than
half its edges may belong to
 $\set {e(N_i)} {1<i\le k}$. See edge $df$ in Fig.\ref{SPStrucFig}.

For any vertex $v$,
a top-down traversal of
\T. gives a naturally defined $vt$-path $p(v)$
that starts with $e(v)$,
that we now describe.
As an example in Fig.\ref{SPStrucFig} $p(j)=j,f,g,h,i,e,c,d,t$;
in Fig.\ref{SPStrucbFig}
this path is composed of pieces in the
subgraphs of 3 nodes, $j,f$; $f,g,h,i,e$;
and $e,c,d,t$.
\ifcase\version %%%%%%%%%%%%%%%%%%%%%%
Details given in the full version.
\else

For any interior node $N$  let $p(v,N)= p(v)\cap
\gamma(V(N))$.
So $p(v)=p(v,\V.)$.
We will specify $p(v)$ by describing
the edge sets
$p(v,N)$.  We  leave it to the reader to add the simple
details that specify the order of these edges in the desired
$vt$-path.

Consider any interior node $N$ and a vertex $v\in V(N)$.
The ends of
$p(v,N)$ are $v$ and $t(N)$. (This is clear for
$N=\V.$, and we shall see it holds for the other nodes by induction.)
Let $v$ belong to $V(N_i)$ for the child $N_i$ of $N$ (possibly
$v=N_i$).
$E(N)$ contains
a unique
$\o N_i \o N_1$-path  $P$ that begins with the edge
$e(N_i)$.
(For $i=1$, $P$ has no edges.)
$p(v,N)$ has the form
\begin{equation}
\label{PvFormEqn}
p(v,N)= E(P) \cup \bigcup \set{p(x_j,N_j)} {\o N_j \in V(P)}.
\end{equation}
%consists of the edges of $P$, plus a subpath $p(v,N_j)$ in
%each $N_j$ belonging to $P$.
Implicit in \eqref{PvFormEqn} is that
$p(v,N_j)$ has
the form $p(x_j,N_j)$, i.e., it has
$t(N_j)$ as one of its ends.
To prove this consider three cases:
If $N_j$ is the last node of $P$ then $p(v,N_j)$ ends
at the end of $p(v,N)$, which is $t(N)=t(N_j)$.
If $N_j$ is the first node of $P$ then
we have chosen the first edge of $P$ as $e(N_j)$,
and it has $t(N_j)$ as an end.
  If $N_j$ is
neither first nor last in $P$
then $P$ contains two edges incident to $\o N_j$,
one of which is $e(N_j)$. For what follows let $f_j$ be the
edge not equal to $e(N_j)$ that is incident to  $\o N_j$. Observe that
 $f_j$ exists
unless $N_j$ is first in $P$.

It remains to specify vertex $x_j$.
This vertex is
  $v$ if $N_j$ is first in $P$, else
it is the vertex  $f_j \cap V(N_j)$.

This completes the definition of $p(v,N)$. Note that
when $N=\V.$ and $E(N)$ is a tree,
$e(N_j)$ is the edge that leaves $N_j$ in (the directed version of) a path $p(v)$;
in all other cases (i.e., $E(N)$ is a cycle)
this needn't hold, e.g., $e(h)$ doesn't leave $h$
in $p(j)$.
\fi %%%%%%%%%%%%%%%%%%%%%%%%%%%

\iffalse
To specify this subpath define
\begin{eqnarray*}
f_j&=&\delta(V(S_j))\cap P -e_j\\
x_j&=&
\begin{cases}
 f_j\cap V(S_j)&j>1\\
v&j=1
\end{cases}\\
P_j&=&p(x_j,S_j).
\end{eqnarray*}

A simple induction shows that this definition makes $p(v,S)$ a path
from $v$ to $t(S)$. For this we use the fact that $P\cap
\delta(V(S_j))=\{e_j,f_j\}$ (where $f_i$ does not exist) and $e_j\cap
S_j= t(S_j)$.
\fi

Next we
specify the numeric values that will prove the $p(v)$'s are shortest paths.
For any $v\in V$ let $P_v$ denote a shortest $vt$-path.
Our approach is based on the subgraphs $P_x\cup P_y\cup xy$.
\iffalse
To give some intuition we will first state our
proof as it specializes to the case of nonnegative
 weight functions. In this case our proof
is just
the following rewriting of the standard
proof.
\fi
\ifcase\version %%%%%%%%%%%%%%%%%%%%%%%%%%%%%
For motivation, the full version shows
that for nonnegative weights, we can
rewrite
Bellman's inequalities as $2d(v)\le d(u)+d(v)+w(uv)$
to prove optimality of paths composed of edges
that have equality.
The generalization to conservative weights
will use a laminar family instead of the lower bounds $2d(x)$.
\or
For motivation we first discuss our proof
 as it specializes to a nonnegative
 weight function.
Recall Bellman's inequality for an edge
$yx$, $d(x)\le d(y)+w(xy)$;
rewrite it, with the above subgraph in mind,
 as $2d(x)\le d(x)+d(y)+w(xy)$.
Consider an arbitrary vertex $v$.
Assume %(as will be the case) that
equality holds in Bellman's inequality for each
edge of $yx$ in $P_v$, $x\ne t$.
We show this implies $w(P_v)\le w(P)$ for any $vt$-path $P$.
Let
 $P=(v=x_0,x_1,\ldots, x_\ell=t)$,
and add the inequalities $d(x_i)+d(x_{i+1})+w(x_ix_{i+1})\ge
2d(x_{i})$, to get
\[
d(v)+ \sum \set{2d(x)}{ x\in P-v,t} +d(t) +w(P)
\ge  \sum \set{2d(x)}{ x\in P-t}.
\]
Equality holds if $P=P_v$.
Add in the identities $2d(x)=2d(x)$ for each $x\notin P$ to get
\[
d(v)+ \sum \set{2d(x)}{ x\in V-v,t} +d(t) +w(P)
\ge  \sum \set{2d(x)}{ x\in V-t}.
\]
Since equality holds for $P=P_v$ we get
$w(P)\ge w(P_v)$ as desired.
\iffalse
We leave it as an exercise to
prove this implies $w(P_v)\le w(P)$ for any $vt$-path $P$.
(Use the identities $2d(x)=2d(x)$ for each $x\notin P$.)
\fi
The generalization to conservative weights
will use a laminar family instead of the lower bounds $2d(x)$.
\fi %%%%%%%%%%%%%%%%%%%%%%%%%

\iffalse
\exercise 1 {Prove this implies $w(P_v)\le w(P)$ for any $vt$-path $P$.
(Use the identities $2d(x)=2d(x)$ for each $x\notin P$.)}

We will prove this implies $w(P_v)\le w(P)$ for any $vt$-path $P$.
Let
 $P=(v=x_0,x_1,\ldots, x_\ell=t)$,
and add the inequalities $d(x_i)+d(x_{i+1})+w(x_ix_{i+1})\ge
2d(x_{i+1})$, to get
\[
d(v)+ \sum \set{2d(x)}{ x\in P-v,t} +d(t) +w(P)
\ge  \sum \set{2d(x)}{ x\in P-t}.
\]
Equality holds if $P=P_v$.
Add in the identities $2d(x)=2d(x)$ for each $x\notin P$ to get
\[
d(v)+ \sum \set{2d(x)}{ x\in V-v,t} +d(t) +w(P)
\ge  \sum \set{2d(x)}{ V\in P-t}.
\]
Since equality holds for $P=P_v$ we get
$w(P)\ge w(P_v)$ as desired.

\noindent
The generalization to allow negative weights
will use a laminar family instead of the lower bounds $2d(x)$.
\fi

A {\em gsp-structure} consists of a gsp-tree
plus two functions
$d$, $z$.
Each vertex $v$ has a value $d(v)$, its distance to $t$.
Each node $N$ of \T. has a value $z(N)$ that is nonpositive
with the exception of $z(\V.)$ which has arbitrary
sign.
Enlarge $E(G)$ to the set $ E_\ell(G)$ by adding a loop $xx$ at every vertex
except $t$,
with $w(xx)=0$.  Also
for any such $x$ define $E(x)$ to be $\{xx\}$
(although $x$ is a node of  \T., its set
$E(x)$ has not been previously defined); set $E(t)=\emptyset$.
Say that a node $N$ of \T. {\em covers}
any edge with both ends in $N$ (including a loop $xx$ at $x\in N-t$)
as well as the edge $e(N)$ (if it exists).
Every edge $xy\in E_\ell(G)$ satisfies
\begin{equation}
\label{DomEqn}
d(x)+d(y)+w(xy)\ge \sum \set {z(N)} {N \text{ covers }xy},
\end{equation}
with equality holding for every edge of
%in every subgraph $E(N)$ and every loop $xx$.
$\bigcup \set{E(N)\cup e(N)}{N \text{ a node of \T.}}$.

\iffalse
Let $\gamma^+$ be the $\gamma$ function for the enlarged graph
(i.e., $\gamma^+(S)$ includes every loop $xx$ for $x\in S$).
Every edge $uv\in  E_\ell(G)$ satisfies
\[
\sum \set {z(N)} {uv \in \gamma^+(N)\cup e(N)}\le
d(u)+d(v)+w(uv),
\]
with equality for every edge of
%in every subgraph $E(N)$ and every loop $xx$.
$\bigcup \set{E(N)}{N \text{ a node of \T.}}\cup \set{xx}{x\in V}$.
\fi

We show
this structure guarantees that each path $p(v)$
is a shortest $vt$-path,
by an argument similar to the nonnegative case
above:
\def\cov{\text{cov}}
Take any $vt$-path $P^-$. Enlarge it to a spanning subgraph
$P$ by adding the loops $xx$, $x\notin V(P^-)$.
Similarly define $P_v$ to be a spanning subgraph formed
by adding loops to $p(v)$.
Let $\cov(P,N)$ be the number of edges of $P$ covered by $N$.
Adding the inequalities \eqref{DomEqn} for each edge of $P$
gives
\begin{equation}
\label{CoverEqn}
d(v)+ \sum \set{2d(x)}{ x\in V-v,t} +d(t) +w(P)
\ge
\sum \set{\cov(P,N)z(N)}{N \text{ a node of \T.}}
\end{equation}
with equality holding for $P_v$.

\xclaim {$P_v$ achieves the maximum value of $\cov(P,N)$,
for every node $N$. Furthermore
every $P$ has the same value of $\cov(P,\V.)$.}

\ifcase\version %%%%%%%%%%%%%%%%%%%%%%%%%%%%%%%%%%%
The easy proof is omitted. The claim implies
the right-hand side of \eqref{CoverEqn} achieves its minimum when $P=P_v$.
This implies $w(P)\ge w(P_v)$ as desired.
\or
\noindent
The claim implies
the right-hand side of \eqref{CoverEqn} achieves its minimum when $P=P_v$
(recall $z(N)\le 0$ for every $N\ne \V.$).
This implies $w(P)\ge w(P_v)$ as desired.

\bigskip
\noindent
{\bf Proof of Claim:}
Every $P$ and $N$ satisfy  $\cov(P,N)\le
|N-t|$. In proof, each vertex of $N-t$ is either
on a loop of $P$ or on a maximal subpath of $P$ from a vertex of $P$
to an edge of $\delta(N)$ or to $t$.
$\cov(P,N)$ equals $|N-t|$ decreased by the number
of the subpaths that leave $N$ on an edge $\ne e(N)$.
This also shows $\cov(P,N)$ equals $|N-t|$ for $P=P_v$
and for $N=\V.$, as desired
\ecproof
\fi %%%%%%%%%%%%%%%%%%%%%%%%%%%

\subsection{Basic facts}

This section derives the basic structure of
shortest-paths to $t$ for conservative graphs.
The key concept,
the ``planted-cycle'', is essentially a special case of
the key concept for $f$-factors, the $2f$-unifactors.
Sec.\ref{BackgroundSec} discusses how
planted-cycles
relate to
general $f$-factor blossoms
(see especially Lemma \ref{sp2.2}).

For convenience perturb the edge weights by adding $\epsilon^i$ to the
$i$th edge. Here $\epsilon>0$ is chosen small enough so that
no two subsets of $ E_\ell$ have the same weight.
Also assume the  edges are ordered arbitrarily except that
the loops $xx$  are the last $n$ edges.
Let $P_v$
denote the unique shortest $vt$-path.
Let $F_v$ be its enlargement to a spanning subgraph,
i.e.,
$P_v$
plus the loop $xx$ for every vertex $x\notin P_v$.
(In the language of $f$-factors, $F_v$ is the unique minimum $f_v$-factor,
see Sec.\ref{BackgroundSec}.)

We start with a simple situation where
conservative weights do not differ from general nonnegative
weights.

\begin{lemma}
\label{SimpleLemma}
$uv\in P_u - P_v$ implies $ P_u=u,v,P_v$.
\end{lemma}

\ifcase\version%%%%%%%%%%%%%%%%%%%%%%
Proof is given in the full version.
This gives a weak analog of the shortest-path tree, proof again is given in the full version:
\begin{corollary}
\label{ForestCor}
$\set {uv} {uv\in P_u \oplus P_v}$ is a forest.
\end{corollary}
\or
\begin{proof}
Let $P_u=u,v,Q_v$, where $Q_v$ is a $vt$-path, and assume $Q_v\ne P_v$.
This assumption implies $w(Q_v)>w(P_v)$, so the $ut$-trail $T=u,v,P_v$
has $w(T)<w(P_u)$, whence $T$
is
not
simple. The assumption $uv\notin P_v$ implies $T=C,R_u$ for
a cycle $C$ through $u$ and a $ut$-path $R_u$. %\con P_v$.
$w(C)>0$ implies
$w(R_u)<w(T)<w(P_u)$, contradiction.
\end{proof}

This gives a weak analog of the shortest path tree:

\begin{corollary}
\label{ForestCor}
$\set {uv} {uv\in P_u \oplus P_v}$ is a forest.
\end{corollary}

\begin{proof}
For contradiction
let $x^0,x^1,\ldots, x^r$ be a cycle of these edges.
Wlog assume  $x^0x^1 \in P_{x^1} - P_{x^0}$. Lemma
\ref{SimpleLemma}
shows
$P_{x^1} =x^1,x^0,P_{x^0}$.
This implies $x^1x^2$ is not in $P_{x^1}$ so it is in $P_{x^2}-P_{x^1}$.
Thus $P_{x^2}=x^2,x^1,x^0,P_{x^0}$.
Continuing this way gives
$P_{x^r} = x^r,x^{r-1},\ldots, x^0,P_{x^0}$.
$P_{x^r}$ a path implies
$x^r\ne x^0$, contradiction.
\end{proof}
\fi%%%%%%%%%%%%%%%%%%%%%%%%%%%%%%%

We will consider subgraphs of the form $ P_u\cup P_v\cup uv$, viewed
as a multigraph contained in $2G$.
Define a {\em p-cycle} (``planted-cycle'') to be
the union of a cycle $C$ and
2 copies of a path $P$ from a vertex $c\in C$ to $t$,
with $V(P)\cap V(C)=\{c\}$. (Possibly $c=t$.)

For motivation first suppose all weights are nonnegative.
It is easy to see (e.g., using the shortest-path tree) that
%\exercise 2 {Show

(a)
$uv\notin P_u\cup P_v$ implies
$ P_u\cup P_v \cup uv$ is a p-cycle.

\iffalse
(b)
$uv\in P_u^- - P_v^-$ implies $ P_u^-\cup P_v^- \cup uv$
consists of 2 copies of $P_u^-$. (Note
this can be viewed as a $uu$-subgraph, if
we add
a loop $uu$ to the graph.
\fi

(b) $uv \in P_u\cap P_v$ never occurs
(assuming our perturbation of $w$).

\noindent

In
the remaining case
$uv\in P_u - P_v$, the lemma shows $ P_u\cup P_v \cup uv$
consists of 2 copies of $P_u$. This can be viewed as a p-cycle  if
we add
a loop $uu$ to the graph.
This motivates an approach similar to the algorithm
of \cite{G} for matching: Enlarge the graph by adding zero-weight loops
$xx$, $x\in V$.
Define the quantity $\zeta(uv) = w(P_u)+ w( P_v)+ w(uv)$.
Repeatedly shrink a cycle of edges with minimum
$\zeta$ value (updating the weights of incident edges
in a natural way). The cycles that get contracted are parts of the desired
paths $P_u$. Once all these pieces are found we can assemble them
into the complete paths.
(The first edge in $P_x$ is revealed in the step that shrinks the loop $xx$.)

A more precise statement of this
``shrinking procedure'' is given in Sec.\ref{BackgroundSec} (or
see \cite{G}); Fig.\ref{SPshrinkFig} in Sec.\ref{BackgroundSec}
will show its execution
on Fig.\ref{SPStrucFig}. Once stated,
it is a simple
exercise
to check that the shrinking procedure
 gets the desired shortest path structure when weights are nonnegative.
In fact the shrinking procedure becomes a variant of Dijkstra's algorithm.

We wish to extend this to conservative weight functions.
Observe that in Fig.~\ref{SPStrucFig}
$P_a\cup P_b\cup ab$ is not a  p-cycle --
it contains an extra copy of $ab$.
We can remedy this by
deleting the extra copy. This suggests the following
definition:%
\footnote{These expressions  are central for general
$f$-factors -- see $\zeta_{uv}$ and $\zeta^{uv}$ defined at
the start of Sec.\ref{VBSec}.}
\begin{equation}
\label{ZetaEqnSP}
\zeta(uv)=
\begin{cases}
w(P_u)+ w( P_v)+ w(uv)&uv \notin P_u \cap P_v\\
w(P_u)+ w( P_v)- w(uv)&uv \in P_u\cap P_v.
\end{cases}
\end{equation}
This definition has a similar failure:
In Fig.~\ref{SPStrucFig} $P_g\cup P_h\cup gh$
contains extra copies of $gh$, $cd$ and $ce$.
But we will see this failure is irrelevant to the shrinking procedure,
and once again it constructs the desired shortest path structure.
The reason is that $\zeta(gh)=3$ is not the smallest $\zeta$ value,
and Fig.\ref{SPshrinkFig} will show it gets "preempted" by  $\zeta(ce)=1<3$.
(Note also that property (a) fails in Fig.~\ref{SPStrucFig}:
$P_d\cup P_f\cup df$ contains extra copies of $cd$ and $ce$.
Again Fig.\ref{SPshrinkFig} will show this is irrelevant.)

\iffalse

Let $\zeta^*$ be the minimum value in
$\set {\zeta(uv)} {uv \notin P_u\oplus P_v}$.
\fi

For the rest of this section assume
\[|\delta(t)|=1.\]
(This  can always be achieved by adding a dummy edge incident to $t$.
The shortest-path structure
for the given graph
can easily be derived from the structure for the enlarged graph.)
The assumption implies any p-cycle  has $c\ne t$.

Let $cc'$ be the multiplicity 2 edge
incident to $C$.
We use the notation $C,c,c'$ throughout the discussion.
When convenient we do not distinguish between the p-cycle, or $C$, or the
pairs $C,c$ or $C,cc'$.
Note that each vertex $x\in C-c$ has exactly
2 $xt$-paths
in the p-cycle, depending on which direction we traverse $C$.

The following lemma gives properties of p-cycles that
are shared with   general blossoms,
%(e.g.,Lemma \ref{b2.2} for $b$-matching)
as well as properties that are specific to
shortest paths (see Sec.\ref{BackgroundSec}).

Let $E^*$ be the set of edges in
$\set {uv} {uv \notin P_u\oplus P_v}$
whose $\zeta$-value is smallest, and let $\zeta^*$ be this smallest
$\zeta$-value.

\begin{lemma}
\label{sp2.2}
\i
There is a p-cycle whose cycle $C$ consists of $E^*$ and possibly other
edges of $\zeta$-value $\le \zeta^*$.
%Also $cc'\in P_c-P_{c'}$ implies $\zeta(cc')\le \zeta^*$.

\ii For every $x\in C$, $P_x$ is one of the $xt$-paths in the p-cycle.
In particular $P_c$ starts with edge $cc'$.

\iii Every $x\in C$ has  $2w(P_x)\le \zeta^*$.
\end{lemma}

\iffalse
n odd circuit $C\con E^*$.
Furthermore \hbox{\rm (i)} $F_v-C_v=F_w-C_w$ and \hbox{\rm (ii)}
$F_v\cap \delta(v)\con E^*$.
\end{lemma}
\begin{lemma}
\label{b2.1}
Any edge $uv\notin $P_u\oplus P_v$
????of a  critical graph %$G$
belongs to a $uv$-subgraph  of edges with
$\zeta$-value
\end{lemma}
\fi

\ifcase\version%%%%%%%%%%%%%%%%%%%%%%%
Proof is in the full version.
\or
\begin{proof}
First note the lemma is trivial if $E^*$ consists of a loop, so suppose not.
The bulk of the argument consists of several claims that
lay the foundation for \xi--\xiii.
%We remark that it seems unclear how to prove the lemma using techniques
%similar to those used for (a) and (b) above.

Take any edge $uv\in E^*$.
Consider the subgraph $S$
of $2G$
formed by modifying
$P_u \cup P_v$ to contain
exactly one copy of $uv$. % (so $S=P_u \cup P_v \pm uv$).
Note that
\begin{equation}
\label{WeightSEqn}
w(S)=\zeta(uv)=\zeta^*.
\end{equation}
We will show $S$ is the desired p-cycle.

\yclaim 1 {\it $E(S)$ can be partitioned into

%(a) a circuit $C$, $uv\in C\con (P_u\oplus P_v)+uv$;
(a) a circuit $C=(P_u\oplus P_v)+uv$;

(b) a multiplicity 2 path $P$ from $t$ to some $c\in C$ ($V(P)\cap V(C)=\{c\}$);

(c) zero or more multiplicity 2 paths
joining two distinct vertices of $C-c$.

}

\bproof
We start by analyzing the multiplicity 1 edges of $S$. We work in the
enlarged graph with edges $E_\ell(G)$.
Let $F_u\con E_\ell(G)$ be $P_u$ enlarged
to an $f_u$-factor by adding loops; similarly for $F_v$.
The subgraph $F_u\oplus F_v$ contains a $uv$-trail $T$
that starts with an edge of $F_v-F_u$ incident to $u$,
ends with an edge of $F_u-F_v$ incident to $v$,
and has edges
alternating between $F_v-F_u$ and $F_u-F_v$.
(Any of these edges including the first and last
may be loops.)
$F_u\oplus T$
is an $f_v$-factor, so
\[
 w(F_u)-w(F_u\cap T)+w(F_v\cap T)\le w(F_v).
\]
Similarly
\[w(F_v)-w(F_v\cap T)+w(F_u\cap T)\le w(F_u).
\]
Adding these inequalities
gives
$w(F_u)+w(F_v) \le w(F_v)+w(F_u)$.
 Thus
all inequalities hold with equality.
The perturbed weight function implies
$F_u -F_u\cap T +F_v\cap T= F_v$.
In particular $F_u-T=F_v-T$.
So
the nonloop edges in $T+ uv$ are precisely the  multiplicity 1 edges in $S$.
Also we get (a) of the Claim.

If $T$ contains a loop $xx$, $x\ne u,v$, then the 2 edges
that immediately precede and follow $xx$
in $T$
are both in $P_u-P_v$ or both in $P_v-P_u$.
This implies  $P_u\oplus P_v$ is a trail
consisting of 1 or more subpaths, each of  1 or more edges,
that alternate between
$P_u-P_v$ and $P_v-P_u$.

$S$ contains a multiplicity 2 path from $t$ to
some $c\in T$ (with no other vertex of $T$). We get (b) of the Claim.

\iffalse
A vertex of $x\in T$ has $d(x,T)\in \{2,4\}$.
(This is clear for $x\ne u,v$.
If $uv\in P_u\cap P_v$ then $d(u,T)=2$. If
$uv\notin P_u\cup P_v$ then it follows from $d(u,P_u\cup uv)=2$
$d(u,P_v)\in \{0,2\}$.)
Say that 2 consecutive edges of $T$ {\em alternate} if
one belongs to
$P_u$, the other to $ P_v$.
It is easy to see that $c$ is on exactly 2 edges of $T$ that alternate.
Let $X$ be the set of all vertices of $T-c$
that are on exactly 2 edges of $T$ that alternate.
\fi

Let $X$ be the set of all vertices of $T-c$
that are on an edge of $P_u\cap P_v-uv$.
If $x\ne u,v$ is the end of an edge $P_u\cap P_v-uv$
then $x$ is
on another edge of $P_u$ and another edge of $P_v$.
It is easy to see this implies that
each vertex of $X$ is joined to another vertex of $X$
by a multiplicity 2 path of $S$.
These are the paths of (c) of the Claim.

Finally note that (a)--(c) account for all edges
of $S$:
An edge of $P_u\oplus P_v$ is in $C$.
An edge of $P_u\cap P_v$ is in a path or cycle of such edges.
A path either has  both ends in $C$ (making it type (c)) or
one end in $C$ (making the other end $t$, so the path is type (b)).
A cycle of $P_u\cap P_v$ edges
cannot exist
since $P_u$ is acyclic.
\ecproof

The above $X$ is a set, not a multiset, i.e., a vertex $x\in X$
is the end of only 1 type (c) path. This follows from $d(x,S)=4$,
which also holds if $x=c$. It is worthwhile to describe the case
$x=u$:
If $uv\in P_u\cap P_v$ then
$d(u,S)=2$, so $u\notin X$.
Suppose $uv\notin P_u\cup P_v$ and $u\in X$.
The first edge of $P_u$, say $f$, must belong to $P_v$; let $g$ be the other
edge of $P_v\cap \delta(u)$. Then $S\cap \delta(u)$
consists of 2 copies of $f$ plus the edges $g,uv$.

\yclaim 2 {\it The edges of type (a) and (c) can be partitioned into
a collection of
 cycles.}

\bproof
Any circuit is the edge-disjoint union of cycles.
So we can assume there are type (c) edges, i.e.,
the above set $X$ is nonempty.

$|X|$ is even, so the vertices of $X$ divide the edges of
$C$ into an even number of segments $C_i$.
Partition the edges of $C$  into 2 sets $\C._1, \C._2$,
each consisting of alternate segments  $C_i$.
Partition the edges of the type (c) paths into 2 sets $\P._1,\P._2$,
each consisting of 1 copy of each type (c) path $P_j$.
Now the edges of type (a) and (c) are partitioned into
the two sets
nonempty  $\P._s\cup \C._s$, $s=1,2$.

{\def\Q.{\P.\C.}
Let $\Q.$ be one of the subgraphs
$\P._s\cup \C._s$.
Observe that each vertex $x$ has $d(x,\Q.)$ even:
If $x$ is interior to a $P_j$ then $d(x,\Q.)=2$.
If $x$  is the end of a $P_j$ then it
is the end of a corresponding $C_i$,
so again $d(x,\Q.)=2$.
%, or 4 if it occurs again in  $C_i$.
(Recall that $X$ is not a multiset!)
\iffalse
(It is worthwhile to  describe the case $x=u$:
We cannot have $uv\in P_u\cap P_v$, since such vertices $u$ have
$d(u,S)=2$ and so $u\notin X$.
So $uv\notin P_u\cup P_v$. Furthermore
the first edge of $P_u$, say $f$, belongs to $P_v$,
and  $\Q.\cap \delta(u)$ consists of  $f$ plus
either the other edge of $P_v\cap \delta(u) $ or $uv$.)
\fi
Any other $x$ has 1 or 2 occurrences in \Q., both  interior to  $C_i$'s.
(2 occurrences may correspond to $x$ occurring twice in some $C_i$,
 or once in two different $C_i$'s.)
Thus $d(x,\Q.)\in \{2,4\}$.

We conclude that each connected component of
$\Q.$ is a circuit (the construction ensures
the circuit is edge-simple, since
any edge of $G$ occurs at most once in
$\Q.$). So \Q. is a union of cycles.
}
\ecproof

\yclaim 3 {\it $C$ is a cycle and there are no type (c) edges.}

\bproof
Suppose the partition of Claim 2 consists of a single cycle.  Then
there are no type (c) edges (since each of the above sets $\P._s\cup
\C._s$ is nonempty).  So Claim 3 holds.

Now for the purpose of contradiction assume the partition of Claim 2
contains at least 2 cycles.  Each cycle has nonnegative weight.  In
fact the perturbation implies the weight is positive.

One of these cycles, call it $B$, contains vertex $c$.  Form a
subgraph $S'$ by using the cycle $B$ and the type (b) path of $S$ from
$t$ to $c$.  Our assumption implies $w(S')<w(S)$.

Let $B$ contain an edge $xy \notin P_x\oplus P_y$ (Corollary
\ref{ForestCor}).  Let $B_x$ ($B'_x$) be the $xt$-trail contained in
$S'$ that avoids (contains) $xy$, respectively.  Define $B_y$ and
$B'_y$ similarly.  Then
\begin{equation}
\label{BxEqn}
w(B_x) +w(B_y)+w(xy)=
w(B'_x) +w(B'_y)-w(xy)=
w(S')<w(S)=\zeta^*.
\end{equation}
(The last equation is \eqref{WeightSEqn}.)  The definition of
$\zeta(xy)$ shows it is at most either the first expression of
\eqref{BxEqn} (if $xy\notin P_x\cup P_y$) or the second expression (if
$xy\in P_x\cap P_y$).  Thus $\zeta(xy)<\zeta^*$, contradiction.
\ecproof

Now we prove \xi--\xiii.
Take any edge $xy \in C$.
The relation \eqref{BxEqn} becomes
$w(B_x) +w(B_y)+w(xy)=
w(B'_x) +w(B'_y)-w(xy)=
w(S)=\zeta^*$.
So $\zeta(xy)\le \zeta^*$.
%, i.e., \i holds for $xy$.
If $xy\notin P_x\oplus P_y$ we get  $\zeta(xy)= \zeta^*$,
and $P_x,P_y$ is either $B_x,B_y$ or $B'_x,B'_y$, i.e.,
\ii holds for $x$.

To prove  \i it remains only to show $C$ contains every edge $u'v'\in E^*$.
Analogous to \eqref{WeightSEqn}, the subgraph $S'$
formed from $u'v'$ the same way $S$ is formed from $uv$ has
weight $w(S')=\zeta(u'v')=\zeta^*$. The perturbed edge weight function
implies $S=S'$. Thus exactly 1 copy of $u'v'$ belongs to $S$, i.e.,
$u'v'\in C$.

\iffalse
Now we prove \i and \xii.
Take any edge $xy \in C$.
The relation \eqref{BxEqn} becomes
$w(B_x) +w(B_y)+w(xy)=
w(B'_x) +w(B'_y)-w(xy)=
w(S)=\zeta^*$.
So $\zeta(xy)\le \zeta^*$, i.e., \i holds for $xy$.
If $xy\notin P_x\oplus P_y$ we get  $\zeta(xy)= \zeta^*$,
and $P_x,P_y$ is either $B_x,B_y$ or $B'_x,B'_y$, i.e.,
\ii holds for $x,y$.
\fi

To prove the first assertion of \ii
%The remaining case of \ii is when
we need only treat the case
$xy\in P_x\oplus P_y$.
(Note the second assertion of \xii is a simple special case of the first.)
These edges are a proper subset of $C$
 (Corollary \ref{ForestCor}). Let $Q$ be a maximal length
path of such edges that does not contain $c$ internally.
At least one end of $Q$, say $r$, is on an edge  $rs \in C -(P_r\oplus P_s)$
%Let $q$ be the other end($q$ may be $c$).
(the other end may be $c$).
If $rs\notin P_r\cup P_s$ then $P_r$ is the $rt$-path
that avoids $rs$
in $C$.  Lemma \ref
{SimpleLemma} shows
any $x\in Q$ has $P_x$ a subpath of $P_r$. Thus \ii holds for $x$.
Similarly if
$rs\in P_r\cap P_s$ then $P_r$ is the $rt$-path
that contains $rs$
 in $C$, and  Lemma \ref
{SimpleLemma} shows
any $x\in Q$ has $P_x$ the $xt$-subpath of $C$ containing $rs$.
Again \ii holds for $x$.

%To prove the second assertion of \xii, the first assertion shows $P_c$
%is the $ct$-path in the p-cycle.  So $P_c$ starts with edge $cc'$.

To prove \iii first assume $x\ne c$.
Thus
$\zeta^*=w(S)=w(B_x)+w(B'_x)>
2w(P_x)$.
For $x=c$ the argument is similar:
Since $G$ is conservative,
$\zeta^*=w(S)= 2w(P_c)+w(C)>2w(P_c)$.
\end{proof}

\fi %%%%%%%%%%%%%%%%%%%%%%%%%%%%%%%

We want the above cycle $C$ to be a cycle node in the gsp-tree.  Lemma
\ref{sp2.2}\ii shows $C$ has the required properties for vertices
$x\in C$. Now we show $C$ has the required properties for $x\notin C$.

A vertex $x$ {\em respects} a p-cycle if either $x\in C$ and $P_x$ is
an $xt$-path in the p-cycle, or $x\notin C$ and $P_x$ either contains
no vertex of $C$ or it contains exactly the same edges of $\gamma(C)+
cc'$ as some $P_y$, $y\in C$.

Let $C,cc'$ be the p-cycle of Lemma \ref{sp2.2}.

\begin{lemma}
\label{spRespectLemma}
%Every shortest path $P_x$ respects $C$.
Every vertex $x$ respects $C$.
\end{lemma}

\remark{We allow $C$ to be a loop $cc\in E_\ell(G)$.
In this case the lemma states that
a shortest path $P_x$ that contains $c$ actually contains $cc'$.}

\begin{proof}
Lemma \ref{sp2.2}\ii shows
we can assume $x\notin C$.
The argument begins similar to Claim 1 of Lemma \ref{sp2.2}.
We work in the
graph with edges $E_\ell(G)$. For any vertex $u$
let $F_u\con E_\ell(G)$ be $P_u$ enlarged
to an $f_u$-factor.
The subgraph $F_x\oplus F_c$ contains an $xc$-trail
that starts with an edge of $\delta(x)\cap F_c-F_x$,
ends with an edge of $\delta(c)\cap F_x-F_c$, and
has edges
alternating between $F_c-F_x$ and $F_x-F_c$.
Let $e$ be the first edge of the trail
that belongs to $\delta(C)$.
 Let $T$ be
the subtrail  that starts at $x$ and ends with edge $e$.

Consider two cases:

\case {$e\ne cc'$} Clearly $e\in F_x-F_c$.
Let $e$ be incident to  vertex $a\in C$.
$T$ is a subgraph of $F_x\oplus F_a$ (Lemma \ref{sp2.2}\xii).
Now follow the argument of Lemma \ref{sp2.2} Claim 1:
$F_x\oplus T$ is an $f_a$-factor
and
$F_a\oplus T$ is an $f_x$-factor,
so we get
$F_x -F_x\cap T +F_a\cap T= F_a$ and
$F_x-T=F_a-T$.
The latter implies $F_x$ and $F_a$ contain the same edges of
$\gamma(C)+ cc'$, i.e., $x$ respects $C$.

\case {$e= cc'$} Clearly
$cc'\in F_c-F_x$. Define a degree-constraint function $f^c$ by
\[f^c(x)=\begin{cases}
1&x=t\\
3&x=c\\
2&x\ne t,c.
\end{cases}
\]
An $f^c$-factor (in $E_\ell(G)$)
consists of
a $ct$-path plus a cycle through $c$, plus loops
at the remaining vertices.
Since $G$ is conservative and $w$ has been perturbed, $cc$ is the smallest cycle through $c$.
So $F^c$,
the minimum-weight $f^c$-factor,
consists of $P_c$ plus a loop at every vertex except $t$
(in particular $cc\in F^c$).
Thus $F^c=F_c + cc$.

Now the argument follows the previous case:
$T$ is a subgraph of $F_x\oplus F^c$.
$F_x\oplus T$ is an $f^c$-factor,
$F^c\oplus T$ is an $f_x$-factor,
so
$F_x -F_x\cap T +F_c\cap T= F^c$ and
$F_x-T=F^c-T$.
We have already noted $cc'\notin F_x$ and the last equation shows
$F_x$ contains the same edges of
$\gamma(C)$ as $F^c$, i.e., every loop $aa, a\in C$.
In other words $P_x$ does not contain a vertex of $C$,
so $x$ respects $C$.
\end{proof}

\subsection{Construction of the shortest-path structure}

The last two lemmas show how to construct the first node
of the gsp-tree. We
construct the remaining nodes by iterating the procedure.
This sections first shows how to construct the gsp-tree; then it
completes the gsp-structure by constructing $z$.

For the gsp-tree we first state the algorithm
and then prove its correctness.
{\em Shrinking a p-cycle}  means
contracting its cycle $C$; $C$ is the {\em shrunken cycle}.
Let $\o G$ be a graph formed by starting with $G$ and repeatedly
shrinking a p-cycle. (So the collection of
shrunken cycles forms a laminar family.)
We treat $\o G$ as a multigraph that contains parallel edges
but not loops. It is convenient to refer to vertices and
edges of $\o G$ by indicating the corresponding objects in $G$.
So let
$\o E(G)$ denote the set of edges of $G$ that correspond to (nonloop)
edges in $\o G$.
Thus writing
$xy\in \o E(G)$
implies
$x,y\in V(G)$. We do not distinguish between
$xy$ and its image in $\o G$.
Similarly,
writing $C,cc'$
for a shrunken p-cycle implies $cc'\in \o E(G)$,
and a $vt$-path in $\o G$ has $v\in V(G)$ and its first edge incident to $v$.
An overline denotes quantities in $\o G$, e.g., $\o w$, $\o \zeta$.
%Finally, for uniformity  we will view every vertex of $\o G$ as a set
%of vertices of $G$. Thus it makes sense to speak of a vertex of $G$ belonging
%to vertex $v$ of $\o G$, even if $v$ is not a contracted set.

The following algorithm  constructs the
gsp-tree \T.. We assume the shortest paths $P_x$ are known.
(The function $z$ is constructed below.)

\bigskip

Initialize
 $\o G$ (the current graph) to the graph
$(V,E_\ell(G))$, and
\T. to contain each vertex of $G$ as a singleton subtree.
Then repeat the following step
until %the only cycles in $\o G$ are loop nodes:
$\o G$ is acyclic:

\bigskip

{\narrower

{\parindent=0pt

Let $C,cc'$ be the p-cycle of weight $\zeta^*(\o G)$ given by Lemma
\ref{sp2.2}.  Shrink $C$ in $\o G$. %If $C$ is a loop,
Set $e(C)=cc'$. Unless $C$ is a loop,
 create a
node in \T. whose children correspond to the vertices of $C$.

}}

\bigskip

\noindent
When the loop halts create a root
node of \T. whose children correspond to the vertices of
the final graph $\o G$.%
\footnote{In contrast with the general definition,
our assumption that $t$ is on a unique edge
ensures the root of \T. is always a tree node.}

\bigskip

To complete the description of this algorithm
we must specify the weight function $\o w$ for $\o G$.
Let $\C.$ be the collection of maximal shrunken cycles that
formed $\o G$.
An edge $cc'\in \o E(G)$ may be associated
with two cycles of \C., one at each end.
For a p-cycle $C,cc'$ in \C., as in
Definition \ref{SPTreeDfn}
$V(C)$ denotes the set of vertices of $G$ that belong to
$C$ or a contracted vertex of $C$ .
$V(\C.)$ is the union of all the $V(C)$ sets.
For $C,c\in \C.$ and $x\in V(C)$,  $C_x\con E(G)$ denotes
%$xc$-subpath of $P_x$.
%(Property P2 below will show $c$ is actually in $P_x$.)
the minimum-weight $xc$-path contained in $\gamma(C,G)$.
%($C_x$ is defined by the contractions that form $C$.)

Let $G$ denote the given graph, and
let
%$\omega=2\sum\set{|w(e)|}{e\in E}$.
\[\omega=2|w|(E).\]
The weight  of an edge $e=xy\in \o E(G)$ in $\o G$ is defined to be
$\o w(e)=w(e)+\Delta(e,x)+\Delta(e,y)$, where
\[\Delta(e,x)=\begin{cases}
0&x \notin V(\C.)\\
-2\omega&\text{\C. contains p-cycle  }C,e \text{ with } x\in V(C)\\
2\omega+w(C_x)&\text{\C. contains p-cycle  }C,cc' \text{ with }x\in V(C), cc'\ne e.
\end{cases}
\]

We will use a variant
of the ``respects'' relation.
Let $\o S$ be a cycle or path in $\o G$. Let
$C,cc'$ be a p-cycle of \C., and  $\o C$ the contracted vertex
for $C$ in $\o G$.
\iffalse
%that is a cycle or a $vt$-path. Here we allow $v$ to be a vertex of $G$,
%and for $v\in V(\C.)$ a $vt$-path begins at the image of $v$.
\fi
$\o S$ {\em respects} $\o C$  if
$\o S\cap \delta(\o C)$ is either empty
or consists of $cc'$ plus $\le 1$ other edge.
This definition corresponds to the previous definition if we view $\o C$
as a p-cycle whose cycle is  a loop $\o C\,\o C$.
More importantly we shall use this fact:
If a path $P$ respects $C,cc'$ and $\o G$ is derived from $G$
by contracting $C$ to a vertex $\o C$, then the image of
$P$ in $\o G$ respects $\o C$.
Also,
for any $C,cc'$ in $\o{\C.}$,
 if the image of a  shortest path $P_u$
respects $\o C$ and contains $\o C$ internally,
then $P_u$ traverses $C$ along the path $C_x$ that joins
the 2 edges of $P_x\cap \delta(C)$.
This follows from the optimality of $P_u$.

$\o S$ {\em respects} $\o {\C.}$  if it respects each $\o C\in \C.$.
\iffalse
call the subgraph $S$ of $G$ We require $S$ to be a $vt$-path if $\o
S$ is.  We say this subgraph $S$ is
\fi
When $\o S$ respects $\o {\C.}$,
the {\em preimage} of $\o S$ is the subgraph
$S$ of $G$ that completes $\o S$ with minimum weight, i.e.,
$S$
consists of the edges of $\o S$
plus,
for each p-cycle $\o C$  on 2 edges of $\o S$,
the minimum-weight path $C_x$ in $C$ that joins
the 2 edges.
This preimage
is unique.
(Any $C,cc'$ of \C. that is on 2 edges of $\o S$ is on
$cc'$ and another edge that determines the vertex $x$
in the definition of the preimage. Note
that
if $\o S$ is a $vt$-path that respects $\o {\C.}$ and $v\in V(\C.)$
then $v=c$ for some p-cycle $C,cc'$ of \C.; so in this case
the preimage of $\o S$
does not contain any edges
of $C$.)

\iffalse
and
with $x\notin V(\o C.)$ or $x=c$ for some
 then $P_x$ is the preimage.
\fi

Lemma \ref{spContractLemma} below shows
the following properties always hold.

\bigskip

{\parindent=0pt

P1:  $\o G$ is conservative.

%P2: For any vertex $x\in V(G)$ with image $\o x$ in $\o G$, $P_{\o x}$
%is the image of $P_x$.  In fact

P2: For any vertex $x\notin V(\C.)$, $P_x$ is the
preimage of $P_{\o x}$.  For any p-cycle $C,cc'\in \C.$,
$P_c$ is the
preimage of $P_{\o C}$, both of which start with edge $cc'$.
For any other $x\in V(C)$, the image of $P_x$ in $\o G$ is $P_{\o C}$.
}

\iffalse
P3: Let $C,cc'$ be a nonloop  p-cycle of some iteration, with children $B_i,
i=1,\ldots, k$ in \T..  For any $i$, $\o E(P_{\o B_i})$ is one of the
$\o B_i\,t$-paths in the p-cycle.
\fi

\begin{lemma}
Assuming P1--P2 always hold, \T. is a valid
gsp-tree for $G$.
\end{lemma}

\begin{proof}
First observe that every vertex $x\in V(G)$ gets
assigned a  value $e(x)$.
This holds as long as some  iteration chooses $xx$ as the minimum weight
p-cycle. So for the purpose of
contradiction, suppose
an iteration contracts a p-cycle $C$
where  $x\in C$ but
the loop $xx$ has not been contracted.
Since $xx\notin P_x\oplus P_x=\emptyset$,
$xx$ is considered
in the definition of $E^*$. So we get the desired contradiction
by proving
%This makes \eqref{xxEqn} a contradiction.
\begin{equation}
\label{xxEqn}
\zeta(xx)=    w(xx)+2w(P_x)<\zeta^*.
\end{equation}
Lemma \ref{sp2.2}\iii shows
$2w(P_x)\le \zeta^*$ (in any iteration).
We can assume the perturbation of $w$
gives every subgraph of $2G$  a distinct weight.
Thus $2w(P_x)<\zeta^*$.
Furthermore we can assume the perturbation enforces a lexical ordering
of the edges. Since
the loop $xx$
is ordered after any edge of $G$, it cannot reverse
this inequality, i.e., \eqref{xxEqn} holds.

The algorithm sets $e(C)$ correctly for
each node $C$ of \T., by definition.
Recall
the path $p(v,N)$ from \eqref{PvFormEqn}.
It is easy to see the lemma amounts to proving that
for
any interior node $C$ of \T. and any vertex $x\in V(C)$, the edges of
$E(P_x)\cap \gamma(C)$ correspond to $p(x,C)$ as defined in \eqref{PvFormEqn}.
Let p-cycle $B,bb'$ be the child of $C$ with $x\in B$.
($B$  is a loop when the child of $C$ is a leaf of \T..)
We consider two similar cases for node $C$.

{\def\gm.{$\o G$}
Suppose  $C,cc'$ is a cycle node of \T..
Let \gm.  be the graph immediately before $C$ is contracted.
P2 in \gm. shows
$P_b$ is the preimage of $P_{\o B}$.
It also shows
$P_{\o B}$ starts with $bb'$.
Lemma \ref{sp2.2}\ii shows $P_{\o B} $ is one of the $\o
B\, t$-paths in the p-cycle.
We have already observed that $C$ and its children have
the correct $e$-values. Thus
\T. specifies  the desired path $p(b,C)$ as defined in  \eqref{PvFormEqn}.
This extends to any $x\in V(B)$ by the last assertion of P2.
}

%  Using the notation of
%P3, suppose $x\in V(B_i)$ for p-cycle $B_i,b_ib'_i$ (possibly a loop).

The remaining case is for the root node \V. of \T., with \V.
 a tree node. In the final acyclic graph,
%let the image of $x$ be the contracted cycle $B,bb'$.
P2 shows $P_b$ is the preimage of $P_{\o B}$.
Thus \T. specifies the desired path $p(b,\V.)$.
As before
this extends to any vertex $x\in V(B)$.
\xecproof
\end{proof}

We complete the gsp-structure by specifying $z$.  For each node $N$ of
\T. let $\zeta_N$ be the weight $\zeta^*(\o G)$ of its corresponding
p-cycle.  For a leaf $x$ the corresponding p-cycle is the loop $xx$;
for the root \V. of \T., which is a tree node, the corresponding
p-cycle is the last p-cycle to get shrunk. Let $p$ be the parent
function in \T..  For each node $N$ define

\[
z(N)=
\begin{cases}
\zeta_N&N=\V.\\
\zeta_N-\zeta_{p(N)}&N\ne \V.\\
\end{cases}
\]

Lemma \ref{spContractLemma} shows this additional property always holds:

\bigskip

{\parindent=0pt

P3:  For any $uv\in \o E(G)$,
\[
\o\zeta(uv)=\begin{cases}
\zeta(uv)&uv \notin P_u\oplus P_v\\
\zeta(uv)&uv \in P_u-P_v, u\notin V(\C.)\\
\zeta(uv)-4\omega&uv \in P_u-P_v, u\in V(\C.).
\end{cases}
\]
}

\noindent
The notation in P3 is unambiguous since
$uv\in P_u$ iff $uv\in P_{\o u}$ by P2.

\begin{lemma}
Assuming P1--P3 always hold, \T. with the above function  $z$ is a
gsp-structure for $G$.
\end{lemma}

\begin{proof}
 The definition of $z$ clearly implies
that for any node $N$ of \T.,
the $z$-values of all the ancestors of $N$ (including $N$)
sum to $\zeta_N$. We claim $\zeta^*(\o G)$ increases every iteration.
The claim implies that
for every $N\ne \V.$,
$z(N)=\zeta_N-\zeta_{p(N)}\le 0$, i.e., $z(N)$ is nonpositive as desired.
To prove the claim note that
for a fixed edge $uv$,
$uv$ belongs to $P_u\oplus P_v$
in one iteration iff it does in the next iteration,
as long as it is not contracted (by P2).
Hence
the only change in $E^*$ from one iteration to the next is that
contracted edges leave $E^*$. P3 shows the edges in $E^*$ retain their original
$\zeta$-values. Thus $\zeta^*(\o G)$ never decreases.

To complete the proof we must show
\eqref{DomEqn} for every $uv\in E_\ell(G)$, with equality for
$e(N)$ edges and other edges of subgraphs $E(N)$.
Consider two cases.

\case {$uv \notin P_u\cap P_v$}
By definition
\[d(u)+d(v)+w(uv)=\zeta(uv).\]
Let $N$ be the deepest node of \T. that covers $uv$.
Observe that
\[
\zeta(uv) \ge \zeta_N,
\]
with equality holding if $uv\in E(N)\cup e(N)$.
In proof,
if $uv\notin P_u\cup P_v$ this follows from
$uv\in E^*$.
If  $uv\in P_u- P_v$,
then Lemma \ref{SimpleLemma} shows $uv=e(u)$.
So the loop $uu$ covers $uv$, $N=\{u\}$,
and equality holds.
Since the sum of the right-hand side of \eqref{DomEqn} equals $\zeta_N$,
combining the two displayed relations gives the desired conclusion
for  \eqref{DomEqn}.

\case {$uv \in P_u\cap P_v$}
By definition
\[d(u)+d(v)-w(uv)=\zeta(uv).\]
$uv$ is covered by the nodes of \T. that are ancestors
of $u$ or $v$. The ancestors of $u$ have $z$-values summing
to $\zeta(uu)=2d(u)$ and similarly for $v$.
 Let node $A$ be the least common ancestor
of $u$ and $v$ in \T..
Since $uv$ is in the cycle of node $A$, $\zeta_A =\zeta(uv)= d(u)+d(v)-w(uv)$
(by definition of $\zeta$).
The sum of the right-hand side of \eqref{DomEqn} equals
\[\zeta_u +
\zeta_v -\zeta_A = 2d(u)+2d(v)-(d(u)+d(v)-w(uv))=
d(u)+d(v)+w(uv)\]
as desired.
\end{proof}

The development is completed by establishing P1--P3:

\begin{lemma}
\label{spContractLemma}
P1--P3 hold in every iteration.
%In each iteration the current graph $\o G$ satisfies P1--P2.
\end{lemma}

\begin{proof}
Consider two edges
$e_i=x_iy_i \in \o E(G)$, $i=1,2$,
where both $x_i$ have the same image in $\o G$.
If $x_i\notin V(\C.)$ then
\begin{equation}
\label{DelVEqn}
\Delta(e_1,x_1)+\Delta(e_2,x_2)=0.
\end{equation}
Suppose the $x_i$ belong to the p-cycle $C$ of \C..
If
$e_1,e_2$ respects $\o C$ then wlog the p-cycle corresponds to $C,e_2$, and
\begin{equation}
\label{DelCResEqn}
\Delta(e_1,x_1)+\Delta(e_2,x_2)=w(C_{x_1}).
\end{equation}
In the remaining case,
i.e., $e_1,e_2$ does not respect $\o C$,
\begin{equation}
\label{DelCNoResEqn}
\Delta(e_1,x_1)+\Delta(e_2,x_2)\ge 2(\omega+|w|(E-\gamma(C))).%=\omega.
\end{equation}

\yclaim 1
{\em Let $\o S$ be a cycle of $\o G$.  If $\o S$ respects $\o {\C.}$ then
$\o w(\o S)=w(S)$ for $S$ the preimage of $\o S$.  If $\o S$ does not respect
$\o {\C.}$  then $\o w(\o S)\ge 2\omega$.}

\bproof
Let $\C.\R.$ ($\C.\N.$) contain the p-cycles of \C.
that are respected (not respected) by
$\o S$, respectively.
Let $CX=\bigcup\set{C_x}{x = x_1 \text{ in } \eqref{DelCResEqn} \text{ for }
C \in \C.\R.}$.
Then
\eqref{DelVEqn}--\eqref{DelCNoResEqn}
imply
\begin{equation}
\label{CWeightEqn}
\o w(\o S) \ge w(\set{e}{e\in \o S \cup CX} ) +
\sum \set{2(\omega +|w|(E-\gamma(C)))} {C\in \C.\N.}.
\end{equation}

\noindent
Here we use the fact that the $2\omega$ terms cancel
on the 2 edges incident to a cycle of \C.\R. (this property actually depends on
our assumption that $t$ has a unique incident edge).

When $\o S$ respects $\o{\C.}$, i.e., $\C.\N.=\emptyset$,
\eqref{CWeightEqn} holds with equality, and we get the claim.
Suppose $\o S$ does not respect $\o {\C.}$. Then choosing $B$ as any
p-cycle of $\C.\N.$, \eqref{CWeightEqn} implies
$\o w(\o S) \ge w(\set{e}{e\in \o S \cup CX} ) +
2(\omega +|w|(E-\gamma(B))) \ge 2\omega$,
giving the claim.
\ecproof

Claim 1 implies  property P1.

\yclaim 2 {\em Let $\o S$ be a $vt$-path in $\o G$. If $\o S$ does not
  respect $\o {\C.}$ then $\o w(\o S)\ge \omega$.  If $\o S$ respects
  $\o {\C.}$ let $S$ be the preimage of $\o S$.  If $v\notin V(\C.)$
  then $\o w(\o S)=w(S)$.  If $v\in V(\C.)$ then $\o w(\o
  S)=w(S)-2\omega$.}

\bproof
$\o S$ satisfies a version of (\ref{CWeightEqn}) that accounts for the
term $\Delta(e_1,v)$ for edge $e_1=\delta(v,\o S)$.  We examine
several cases.

Suppose $\o S$ respects $\o{\C.}$.  If $v\notin V(\C.)$ then
(\ref{CWeightEqn}) holds with equality, giving the claim.  If $v\in
V(\C.)$
then \C. contains a p-cycle $C,e_1$ with $v\in V(C)$.
$\o w(\o S)$ contains an extra term
$\Delta(e_1,v)=-2\omega$, again giving the claim.

Suppose $\o S$  does not respect $\o{\C.}$.
If $v\notin V(\C.)$ then
 (\ref{CWeightEqn}) holds unmodified. As in Claim 1,
$\o w(\o S)\ge 2\omega$, giving the current
 claim.
So suppose
$v\in V(B)$ for $B\in \C.$. If $e_1$ does not respect $\o B$
$\o w(\o S)$ contains an extra term
$\Delta(e_1,v)=2\omega +w(B_v)\ge \omega +|w|(E-\gamma(B_v))$,
giving the claim. In the remaining case
$e_1$ respects $\o B$
and
$\o S$ does not respect some
$\o A\ne \o B$. The right-hand side of  (\ref{CWeightEqn}) contains the extra term
$\Delta(e_1,v)=-2\omega$, and the term for $A$ is at least
$3\omega + 2|w|(E-\gamma(A))$. These two contributions
sum to $\ge \omega+|w|(E-\gamma(A))$. Thus the right-hand side
of (\ref{CWeightEqn}) is $\ge \omega$, as desired.
\ecproof

\yclaim 3 {P2 holds every iteration.}

We argue by induction on the number of iterations.
Consider an iteration for the p-cycle $C,cc'$.
%, with $B_i$ as defined in P3.
Let \gm.,\wm. (\gp.,\wp.) be the graph and weight function
 immediately before (after) $C$ is contracted, respectively.
For greater precision, if $H$ is
\gm. or \gp. and $z$ is a vertex of $H$,
let $P(z,H)$ denote the shortest $zt$-path in $H$.
Take any $x\in V(G)$. Let $P= P(\o x, \gp.)$.
Note that $P_x$ continues to denote the shortest $xt$-path in $G$, and
for P2 we want to establish the relationship between $P$ and $P_x$.
\C. denotes the family of contracted vertices in \gp., i.e., it includes $C$.

First assume either $x\notin V(\C.)$ or $x=b$ for some p-cycle $B,bb'$ of
$\C.$ ($x=c$ is a possibility).
Let $\delta$ be 0 ($2\omega$)
if $x\notin V(\C.)$ ($x=b$) respectively.
Let $Q=P(\o x, \gm.)$.
Lemma \ref{spRespectLemma} shows $Q$ respects $C$.
Thus  $Q$ has an image $Q^+$ in \gp.
that respects $\o C$.
The inductive assumption of P2 in \gm. shows
$P_x$ is the preimage of $Q$. (For $x=c$,
$\o x$ in \gm. may be a vertex or a contracted cycle, and we use
the appropriate assertion of P2.)
Thus
the optimality of $P_x$ implies
it is the preimage of $Q^+$. Thus Claim 2 shows
\[\o w^-(Q)=\o w^+(Q^+)=w(P_x)-\delta.\]

Since $\omega> w(P_x)$, Claim 2 shows $P$ respects $\o {\C.}$ and
\[\wp.(P)=w(P^-)-\delta\] for  $P^-$ the preimage of $P$.
Since $P^-$ is an $xt$-path,
\[w(P^-)\ge w(P_x).\]
Combining the inequalities gives $\wp.(Q^+)\le \wp.(P)$. Thus
$Q^+=P$. So as asserted by the first part of P2, $P_x$
is the preimage of $P_{\o x}$.

Also, as in the second assertion of P2, every
p-cycle $B,bb'$ of $\C.$
has
$P_b$ and $P({\o B},\gp.)= P({\o b},\gp.)$
both
 starting with edge $bb'$.
This follows since
$P({\o b},\gm.)$ starts with $bb'$ (by definition if $b$
is not in a contracted vertex of \gm., else by  P2 in \gm.)
and $P_b$ is the preimage of
both $P({\o b},\gm.)$ and $P({\o b},\gp.)$.

%Lemma \ref{sp2.2} , and

Finally  the last assertion of P2 follows since
any $y\in V(B)$ has $\o y=\o b$.
\ecproof

\yclaim 4 {P3 holds in every iteration.}

\bproof
By definition \[\o\zeta(uv)=
\o w(P_{\o u})+ \o w( P_{\o v})\oplus \o w(uv)=
(\o w(P_{\o u})\oplus\Delta(uv,u))+
(\o w( P_{\o v})
\oplus\Delta(uv,v))
\oplus w(uv)\]
for some $\oplus\in \{+,-\}$.
Since $uv$ is an edge of $\o G$, P2 shows
$uv\in P_u$ iff $uv\in P_{\o u}$. Hence the version of this formula
for $\zeta(uv)$ uses the same sign $\oplus$ as $\o \zeta(uv)$.
%The rest of the argument uses the same logic in all cases:
We evaluate  $\o\zeta(uv)$ using the formula, with
Claim 2 %of Lemma \ref{spContractLemma}
providing $\o w(P_{\o u})$
 and the definition of $\Delta(uv,u)$
giving its value, as follows.

Suppose $u\notin V(\C.)$.  $P_u$ is the preimage of $P_{\o u}$ by
P2. % and the definition of preimage.
 So
\[\o w(P_{\o u})\oplus\Delta(uv,u)=w(P_u)\]
since $\o w(P_{\o u})=w(P_u)$, $\Delta(uv,u)=0$.

Suppose $u\in V(\C.)$, say $u\in V(C)$ for the p-cycle $C,cc'$.
$P_c$ is the preimage of $P_{\o u}$ by
P2. % and the definition of preimage.
Claim 2 shows $\o w(P_{\o u})=w(P_c)-2\omega$.
If $uv \notin P_u$ then
\[\o w(P_{\o u})+ \Delta(uv,u)=w(P_u)\]
since $\Delta(uv,u)=2\omega+w(C_u)$, and $P_u=C_u\cup P_c$ by
P2 and the optimality of $P_u$.
If $uv \in P_u$ then
\begin{eqnarray*}
\o w(P_{\o u})- \Delta(uv,u)&=&w(P_u),\\
\o w(P_{\o u})+ \Delta(uv,u)&=&w(P_u)-4\omega
\end{eqnarray*}
since P2 implies $u=c$ and $\Delta(uv,u)=-2\omega$.

%$\o w(P_{\o u})+\Delta(uv,u)=(w(P_u)-2\omega)+(-2\omega)=w(P_u)-4\omega$.

The alternatives of P3 all follow, by
combining the equations with $-$ signs for $uv\in P_u\cap P_v$
and $+$ signs in all other cases.
\xecproof
\end{proof}

%++++++++++++++++++++++++++++++++++++++++++++++++++++++

%%%%%%%%%%%%%%%%%%%%%%%%%%%%%%%%%%%%%%%%%%%%%%%

%% file: background.tex
\section{Background on \boldmath {$f$}-factors}
%, and shortest-path tree algorithms}
\label{BackgroundSec}
\def\AfTheorem{\cite[Theorem 4.17]{G}}
\def\NewWeights{\cite[Fig.13]{G}}
\def\AWeightLemma{\cite[Lemma 4.13]{G}}
\def\ARespectLemma{\cite[Lemma 4.2]{G}}
\def\ARespectCor{\cite[Corollary 4.11]{G}}
\def\BTreeDfn{\cite[Definition 4.14]{G}}
\def\ConstructionSec{\cite[Section 4.4]{G}}
\def\AEndLemma{\cite[Lemma 4.12]{G}}

\iffalse
f odd pair
???  $f$-odd sets $(V(B),I(B))$ of (\ref{DualDefnEqn})).

$y(e)  + z\set {B,\ V}$ notation
defn f-factor, b matching
maximum f factor
d(v,H)
Las Vegas
perfect to critical reduction

Recall that a real-valued edge weight function is {\em conservative}
if there are no negative weight cycles \cite{LP}.
If $G$ is a graph, $2G$ is the multigraph consisting of
two copies of every edge of $G$. Similarly if $G$ is a multigraph.
$f_v$ ?

possible intro text:
The three sections
are all based on the idea of a shrinking procedure proposed in
\cite{G}. The details differ slightly: The b-matching procedure uses
the same $\zeta$ values as used in ordinary matching in \cite{G}.
The shortest path procedure uses two types of $\zeta$ values.
The f-factor procedure uses unifactors (equivalently, all three types of $\zeta$-values.

{\bf FINAL CONCLUSION:}
Assume
the above weights
$w(F_v)$ can be found in time
$\tilde{O}(W \phi^\omega)$ and a $b$-matching can be found in
time $\tilde{O}(\phi^\omega)$, both  with high probability.
Then the algorithm of this section for finding $F_v$ can be implemented
in total time $O(W\phi^\omega)$.
So a maximum $b$-matching can be found in time
$\tilde{O}(W \phi^\omega)$ with high probability.

\fi

The bulk of this section reviews
the approach of \cite{G}  based on critical graphs.
The review ends by illustrating the shrinking procedure.
Then we show that procedure efficiently
constructs a  generalized shortest-path structure.

To define critical graphs,
for each vertex $v\in V$ define $f_v$,
the {\em lower perturbation} of $f$ at $v$,
by decreasing $f(v)$ by 1.
Similarly define $f^v$,
the {\em upper perturbation} of $f$ at $v$,
 by increasing $f(v)$ by 1.
Every $f_v$, $f^v$, $v\in V$
is a {\it perturbation} of $f$.
\fs v. stands for a fixed  perturbation
that is either
$f_v$ or $f^v$.
A graph is {\it $f$-critical} if it has an $f'$-factor for every
perturbation $f'$ of $f$.%
\footnote{The usual notion of criticality for
matching
only assumes existence of the factors for lower perturbations.
It is easy to see this implies existence for all the upper
perturbations. This holds for
$b$-matching too, but not
for general $f$-factors \cite{G}.}

It is easy to see that
a maximum $f$-factor of $G$ can be found by working on
the critical graph $G^+$ formed by adding a vertex $s$ with
edges $sv$, $v\in V(G)$ and loop $vv$, all of weight 0,
and $f(s)=1$: A maximum $f_s$-factor is the desired $f$-factor.

\iffalse %%%%%%%%%%%%%%
As with $b$-matching we shall find the weighted blossoms $B$ in order
of decreasing value $\bar z(B)$.  For $f$-factors a blossom is a pair
of sets $(V(B), I(B))$. We find these pairs in two steps
using a graph $G(\zeta)$: Having found the blossoms of $\bar z$-value
$>\zeta$, we construct $G(\zeta)$.  Its 2-edge-connected components
constitute the $V$-sets for all weighted blossoms of $\bar
z$-value $\zeta$.  The edges of $I$-sets that are still unknown are found
amongst the bridges of $G(\zeta)$ or in a related computation.

We assume the quantities, $w(F_v), w(F^v)$, $v\in V$ are known.  The
following development is based on \cite[Section 4]{G}, which is
included as an appendix. It is important to note that the construction
of \cite{G} modifies the given weight function $w$ in two ways: First,
the given edge weights are perturbed slightly so that every maximum
factor $\Fs v.$, every maximum $2f$-unifactor, and every maximum
irredundant $2f$-unifactor is unique.
%%%%CHECK THIS
Second, each time a blossom is contracted the weights
of its incident edges are modified (by large quantities, including
a value $W$ greater than the sum of all previous edge weights).
The reduction of this section
has no access to these modifications, and it works entirely with the given
edge weight function $w$. But our references to \cite{G} must
respect the differences in weight functions.

\fi %%%%%%%%%%%%%%

Recall the linear programming formulation for maximum weight
$f$-factors, especially the dual problem (e.g.~\cite[Ch.~32]{Schrijver}).
We summarize  the slight modification
used  in \cite{G} for critical graphs.
Simply put we want the dual variables to be optimum for every perturbation
of $f$.

Dual variables for a graph with integral weights are functions
$y:V\to \mathbb {Z}$ and $z:2^V\times 2^E\to \mathbb {Z}$.
Pairs $B=(V(B),I(B))$ with nonero $z$-value
are called (weighted) ``blossoms''.
$z(V)$ is a shorthand for $z((V,\emptyset))$.
The optimum dual function
$y$ is defined by
$y(v)=-w(F_v)$ \AfTheorem, as expected from Sec.\ref{SPSec}.

%We start by describing the structure of these blossoms
%(in the optimum dual functions that we  use).

The $z$ function has support given by a forest that we now describe.
It generalizes the shortest-path structure
of Section \ref{SPSec}. The usual version corresponds exactly
to that structure; it also corresponds to
the blossom tree for matching \cite{G}.
Our algorithms  use a weighted version of this structure -- it is the
shortest-path structure/blossom tree with every cycle node/blossom
of 0 $z$ value contracted into its parent. This
{\em weighted blossom forest} is defined as follows
(we  give a minor modification of \BTreeDfn,
using \cite[Lemma \AEndLemma]{G}).
Let \M. be the set of all maximal blossoms of $G$.
$\M.=\{V\}$  for matching and $b$-matching but not generally
(for shortest paths \M. is the set of maximal cycle nodes).

\bigskip

\i Each $B\in \M.$ is the root of a
{\em weighted blossom tree} \W.. Each interior node of \W. is a weighted
blossom
and each leaf is a vertex of $G$.
The
children of any node $B$ are the maximal weighted
blossoms properly contained in $B$
plus all vertices of $G$ contained in $B$ but no smaller weighted blossom.
$V(B)$ is the set of all leaf descendants of $B$, i.e., the vertices in $B$.
Each vertex of $G$ belongs to exactly one blossom of \M..

\ii The support of $z$ is $\set {B,\ V} {B \text{ a blossom of
a \W.-tree}}$.
$z(B)>0$  for blossoms $B$ of \W.,
while $z(V)$ may have arbitrary sign.

\iii Each blossom $B$  has a set of edges $I(B)\con \delta(V(B))$.

\iv The blossoms of \M.
are nodes of
a tree \T..
 Every edge $AB$ of \T. belongs
to $I(A)\oplus I(B)$.
%(When $\M.=\{V\}$, \T. consists of one vertex and gives no information.)

\bigskip

%\ifcase\version%%%%%%%%%%%%%%%%
%For shortest paths,
%$I(B)$ consists of the edge $e(N)$.
%%special topology of the $e(N)$ edges comes from ``respectfulness'' below.
%As in shortest paths and matching, a blossom $B$ {\em covers} any edge of
%$\gamma(V(B))\cup I(B)$.
%An edge $uv$ of $G$ has a  value
%$$\H{yz}(uv) = y(u) +y(v)  + z\set {B,\ V}
%{B \text{ a blossom of \W. that covers } uv}. $$
%\or

For shortest paths,
$I(B)$ consists of the edge $e(N)$.
%special topology of the $e(N)$ edges comes from ``respectfulness'' below.
As in shortest paths and matching, a blossom $B$ {\em covers} any edge of
$\gamma(V(B))\cup I(B)$.
An edge $uv$ of $G$ has a  value
$$\H{yz}(e) = y(e)  + z\set {B,\ V}
{B \text{ a blossom of \W. that covers } e}. $$
(Here we use the above convention; recall $e=\{u,v\}$.)
%\fi %%%%%%%%%%%%%%%%%%%%
$e$ is {\em underrated} if
\[w(e)\ge \H{yz}(e);\]
$e$ is {\em strictly underrated} if the inequality is strict and
{\em tight} if equality holds.
An $\fs v.$-factor
$F$ {\em respects} a blossom
$B$
 iff
\begin{equation}
\label{RespectEqn}
F \cap  \delta(V(B))=
\left\{
\begin{array}{ll}
I(B)&v\in V(B)\\
I(B)\oplus e&v\notin V(B), \text{ $e$ is some edge in $\delta(V(B))$}
\end{array}
\right.
\end{equation}

\noindent
(For shortest paths, the first alternative says
a  shortest $vt$-path leaves  $B$ on $e(B)$ if $v\in V(B)$;
otherwise the  second line  says it contains
$e(B)$ and one other edge ($t\notin B$) or one edge $\ne
e(B)$ ($t\in B$)).
An $\fs v.$-factor has maximum weight if it contains
every strictly underrated edge, its other edges are tight, and it
respects every blossom with positive $z$-value.
(The optimum $f$-factors we use satisfy this criterion.)

As in matching,  blossoms
are built up from odd cycles, defined as follows \cite[Definition 4.3]{G}:
An {\em elementary blossom} $B$ is a 4-tuple $(VB, C(B), CH(B), I(B))$, where
$VB\con V$,
$C(B)$ is an odd circuit on $VB$,
$CH(B)\con \gamma(VB)-E(C(B))$,
$I(B)\con \delta(VB)$, and
every $v\in VB$ has
$f(v)= d(v, C(B))/2 +d(v,CH(B) \cup I(B))$.
\iffalse
\begin{equation}
\label{EBlossomEqn}
f(v)= d(v, C(B))/2 +d(v,CH(B) \cup I(B)).
\end{equation}
\noindent
\fi
%\ifcase\version%%%%%%%%%%
%Blossoms for shortest pairs  have $|I(B)|\le 1$
%and $CH(B)=\emptyset$.
%\or
 $C(B), CH(B)$, and $I(B)$ the {\em circuit, chords}, and
{\em incident edges}
of $B$, respectively.
We sometimes use
``blossom'' or ``elementary blossom'' to reference
the blossom's circuit or its odd pair.
Blossoms for shortest pairs  have $|I(B)|\le 1$
(as mentioned)
and $CH(B)=\emptyset$ (this is the import of
Lemma \ref{sp2.2}, especially part \ii and Claim 2 respectively).
%\fi %%%%%%%%%%%%%%%%%%%%%%

The blossoms for optimum duals are found
by repeatedly finding the next blossom and shrinking it.
Fig.\ref{SPshrinkFig}
illustrates the
shrinking of the 3 cycle nodes of Fig.\ref{SPStrucbFig}, e.g.,
$ \{a,b,c\}$ gets shrunk in  Fig.\ref{SPshrinkFig}(a).
In general the  vertices of the blossom
are contracted and a loop is added to the contracted vertex.
Edges incident to the contracted vertex get their weights adjusted
to account for contracted edges.
%\ifcase\version%%%%%%%%%%%
%(in   Fig.\ref{SPshrinkFig}(a)
%$w(be)$ decreases by 1 to account for path $b,a,c$.)%
%\footnote{The weight adjustment of \NewWeights\ also involves a large quantity
%$J$. We omit it since $J$ is needed only for the proof of \cite{G}, not for
%explaining the result.}
%\or
(For shortest paths these edges are the paths $p(v,N)$.
In   Fig.\ref{SPshrinkFig}(a)
$w(be)$ decreases by 1 to account for path $b,a,c$.)%
\footnote{The weight adjustment of \NewWeights\ also involves a large quantity
$J$. We omit it since $J$ is needed only for the proof of \cite{G}, not for
explaining the result.}
%\fi %%%%%%%%%%%%%%%%%%%%%%%

\begin{figure}[t]
\begin{center}
\epsfig{figure=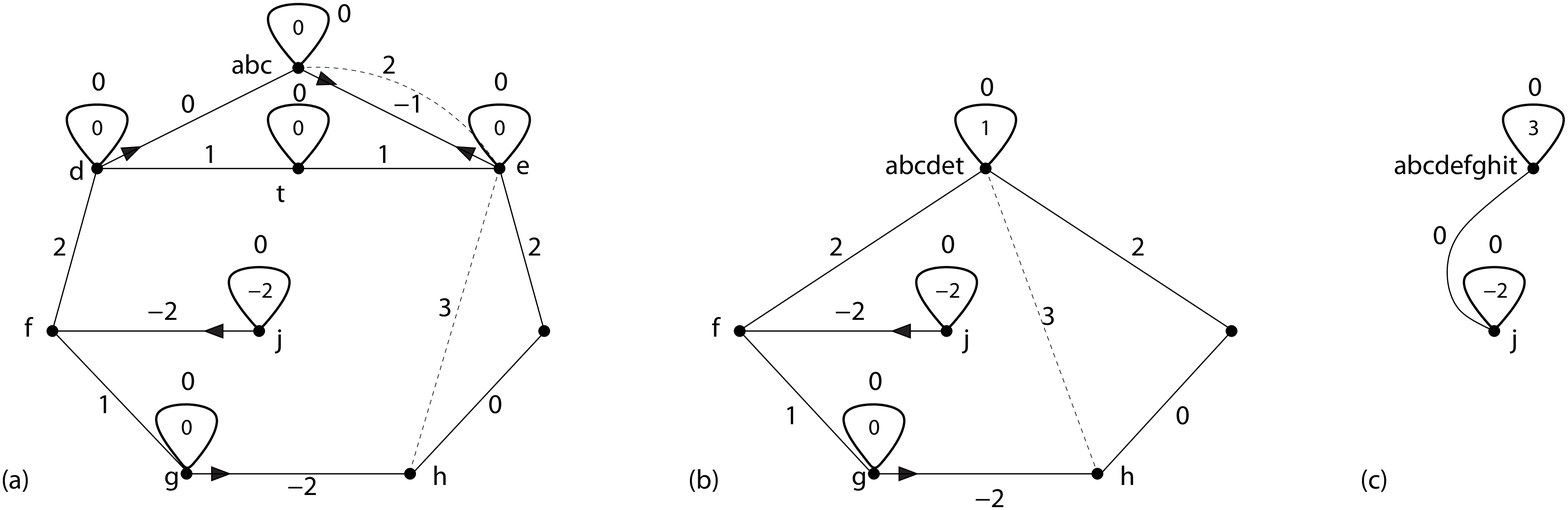, width=\textwidth}%2in}
%\hw}}
\end{center}
\caption{Contracted graphs for nonloop unifactors.}
\label{SPshrinkFig}
\end{figure}

%\ifcase\version%%%%%%%%%%%%%
%Perturb the edge weights slightly, as in
%Sec.\ref{SPSec}.
%Let $F_v$ and $F^v$ denote the unique
%maximum weight {$f_v$} and {$f^v$} -factors
%respectively.
%\or
For simplicity perturb the edge weights slightly so that no two sets of edges
have the same weight.
That is, number the edges from 1 to $m$ and increase the weight of the
$i$th edge by $\epsilon^i$ for some $\epsilon\ge 0$.
For small enough $\epsilon>0$,
no two sets of edges have the same weight.
Thus
any such
perturbation has a unique maximum factor which is also maximum for the
original weights.
Let $F_v$ and $F^v$ denote the maximum weight {$f_v$} and {$f^v$} -factors
respectively.
(Eventually (in \eqref{DualDefnEqn})
we set $\epsilon$ to 0
and define the dual function $z$ using the original weights.)
%\fi %%%%%%%%%%%%%%%%%%%%%%
%
Recall (Sec.~\ref{sec:definitions}) the  multiset notation
$2S$, $S/2$, $2G$.
Assume that even in $2G$, no two sets of edges have the same weight.

\iffalse
. denotes $S$ with every multiplicity doubled.
Similarly for a multigraph $G=(V,E)$, $2G$ denotes the multigraph
$(V,2E)$.
If every multiplicity of $S$ is even then
$S/2$ denotes $S$ with every multiplicity halved.
\fi

We choose the
blossom to shrink next using
a subgraph that generalizes the $uv$-subgraphs of Section \ref{SPSec}:
Any subgraph of $2G$ consists of edges
of $G$ at multiplicity 0,1 or 2.
A {\em {$2f$}-unifactor} is a $2f$-factor of $2G$
whose multiplicity 1 edges form
an odd circuit. (When  possible we
abbreviate ``$2f$-unifactor'' to ``unifactor''.)
The elementary blossom for a unifactor  $U$ with odd circuit $C$ is the
 4-tuple $B=(V(C), C, CH(B), I(B))$
where
$CH(B)=(U\cap \gamma(C)-E(C))/2$ and
$I(B)=(U\cap \delta(C))/2$.
%\ifcase\version%%%%%%%%%%%%%%%%%
%For shortest paths note how this defines $I(B)$ to be $e(N)$.
%\or
For shortest paths note how this defines $I(B)$ to be $e(N)$
(the type (b) path to $b$ ends in multiplicity 2 edge $e(N)$).
%\fi %%%%%%%%%%%%%%%%%%%%%%%%%

The {\em shrinking procedure} of \cite{G} constructs the
blossoms
as follows:

\bigskip

{\narrower

{\parindent =0pt

Let the next
blossom $B$
be the elementary blossom of the maximum weight proper  $2f$-unifactor.%
\footnote{``Proper'' means (a) the unifactor's circuit
is not a loop created previously when a blossom was contracted;
(b) the unifactor respects each previous blossom. Here ``respects''
is the generalization of \eqref{RespectEqn} to unifactors, e.g., a shortest path
that enters and leaves a blossom must contain $e(N)$. \cite{G}
enforces (b) using
the previously mentioned quantity $J$ that we omit.}
Shrink $B$. Repeat this step until no unifactor exists.

}}

\bigskip

Let us describe how this procedure constructs the blossoms and tree
of Fig.\ref{SPStrucbFig}. We use
Fig.\ref{SPshrinkFig}.
The shrinking procedure  repeatedly finds the minimum weight proper unifactor.
Let $\zeta^*$ denote its weight.
$\zeta^*$ never decreases from step to step, so we describe
the blossoms found at the various values of $\zeta^*$.
Recall that we treat the shortest-path problem by adding a 0 weight
loop at every vertex $\ne t$.

\bigskip

{\narrower

{\parindent =0pt

$\zeta^*=-2$: The loops at $a, b$ and $j$ are blossoms.
(For instance the unifactor for $j$ corresponds to
the shortest $jt$-path -- it consists of loop
$jj$ (multiplicity 1) and
the $jt$-path of weight $-1$ plus
a loop at every vertex not on the path (all multiplicity 2).
The weight of the unifactor, $-2$, is drawn inside the loop at $j$
in Fig.\ref{SPshrinkFig}(a). Fig.\ref{SPshrinkFig} does this for all blossoms.)

$\zeta^*=0$: The loops at $d,e,g,t$ are blossoms, as is
 cycle $a,b,c$.  Fig.\ref{SPshrinkFig}(a) shows
the graph after all these blossoms have been shrunk.

$\zeta^*=1$: Contracted vertex $\{a,b,c\}$ plus
vertices $d,e,t$ form a blossom. Fig.\ref{SPshrinkFig}(b)
shows the graph after it has been shrunk.

$\zeta^*=2$: The loops at $f, h$, and $i$ are blossoms.

$\zeta^*=3$: All vertices but $j$ form a blossom.
 Fig.\ref{SPshrinkFig}(c) shows the graph after it is shrunk.
The new weight on $jt$ reflects the contracted edge $fd$.

}}

\bigskip

Note that when a loop $xx$ becomes a blossom, the unifactor's
weight is $2d(x)$. This corresponds to the bound $2d(x)$ in Sec.\ref{SPSec}
mentione for the Bellman inequality argument. Also, the description
for loops in $\zeta^*=-2$  might seem to imply we need to know
the shortest paths to execute this procedure -- but see the
 implementation in Sec.\ref{fFactorSec}!

The optimum dual function $z$ is defined as follows.
For any blossom $B$ let
$U(B)$ be its corresponding unifactor.
If $B$ is not a maximal blossom  let
$p(B)$ denote the blossom $B$ gets contracted into.
Let $w$ be the original (unperturbed) weight function.

\begin{equation}
\label{DualDefnEqn}
\begin{array}{lcll}
%y(v)&=&-w(F_ v)\hspace{102pt}v \in V,\\
%\end{array}\\
%\]
%\[
%z(G)&=&\zeta(G)\\
%z(B) &= &\zeta(B)-\zeta(p(B))&B \mbox{ a blossom }\ne G.
%\begin{array}{lcll}
z(B) &= &
\left\{
\begin{array}{ll}
w(U(B)) &\text{$B$ a maximal blossom},\\
w(U(B))-w(U(p(B)))&B \mbox{ nonmaximal.}%\\
%0&\mbox{otherwise.}
\end{array}
\right.
\end{array}
%\]
\end{equation}

\noindent
For instance in Fig.\ref{SPStrucbFig}
blossom $aa$ is created at $\zeta^*=-2$,
 $\{a,b,c\}$
is created at $\zeta^*=0$, and $z(\{a,b,c\})=-2 - 0=-2$.

The
following characterization of
maximum weight unifactors is central to the analysis:

\begin
{lemma}[{\cite[Lemma 4.5 and Cor.~4.6]{G}}]
\label{UnifactorLemma}
For any vertex $v$, the maximum weight $2f$-unifactor containing $v$
in its circuit is $F_v + F^v$.
For any edge $uv$ consider 3 cases:

$uv \notin F_u\cup F_v$:  The maximum weight $2f$-unifactor containing $uv$
in its circuit is $F_u + F_v+uv$.

$uv \in F_u\cap F_v$:  The maximum weight $2f$-unifactor containing
$uv$ in its circuit is $F^u + F^v-uv$.

 $uv \in F_v -F_u$:
$F_u=F^v-uv$.
Furthermore
the maximum weight $2f$-unifactor containing $v$
in its circuit
is $F_u + F_v+uv$, and this unifactor
contains $uv$ as an edge incident to its circuit.
\end{lemma}

For shortest paths \i -- \ii
characterize the maximum $uv$-subgraph.
\iii
 corresponds
to Lemma \ref{SimpleLemma}.

We close this review
by reiterating some notation from \cite{G} that we use in the next two sections:

$w$: the given, unperturbed,  weight function.

\B.:
the forest whose nodes are the elementary blossoms  found by the
shrinking procedure, plus the vertices of $G$.
 The children of a blossom $B$ are the vertices that
get contracted to form $B$.

\iffalse
If edge $uv\in E(G)$
has its image in the circuit
of $U$, and $uv \notin F_u\oplus F_v$, then $U$ is the image of the maximum
weight $2f$-unifactor of (the original graph) $G$ that contains $uv$ in its circuit.
This follows from repeated applications of Lemma \ref{AWeightLemma}\xiii.
(Recall that by definition there are no redundant unifactors in $G$.)
\fi

$\bar z(B)$:
the sum of the dual values $z(A)$ for every ancestor $A$ of $B$ in \B..
If $B$ does not appear in
the weighted blossom forest \W. (defined above) then $\bar z(B)=
\bar z(A)$  for the blossom  $A$ of \W. that absorbs
$B$ by contraction.

$B_v$, for any vertex $v\in V(G)$:
the smallest weighted blossom  containing $v$.

%%%%%%%%%%%%%%%%%%

%%%%%%%%%%%%%%%%%%%%%%%%%%%%%%%%%%

%% file: fred.tex
\section{Weighted \boldmath {$f$}-factors}
\label{fFactorSec}

This section presents the reduction of maximum weight
$f$-factors to unweighted $f$-factors.
As usual we use the shrinking procedure to
find the weighted blossoms $B$ in order
of decreasing value $\bar z(B)$.  A blossom is a pair
of sets $(V(B), I(B))$. We find these pairs in two steps
using a graph $G(\zeta)$: Having found the blossoms of $\bar z$-value
$>\zeta$, we construct $G(\zeta)$.  Its 2-edge-connected components
constitute the $V$-sets for all weighted blossoms of $\bar
z$-value $\zeta$.  The edges of $I$-sets that are still unknown are found
amongst the bridges of $G(\zeta)$ or in a related computation.

In our references to \cite{G} we are careful to
recall that
\cite{G} modifies the given weight function $w$ in two ways: First,
the given edge weights are perturbed so every maximum
factor $\Fs v.$ and every maximum
$2f$-unifactor is unique.
%%%%CHECK THIS
Second, in \cite{G} each time a blossom is contracted the weights
of its incident edges are modified (by adding large quantities, including
a value called $W$ greater than the sum of all previous edge weights).
The reduction of this section
has no access to these conceptual modifications, and it
must works entirely with the given
edge weight function $w$.

 Our overall strategy is similar to the $b$-matching algorithm of
 Section \ref{GeneralbMatchingSec}: We
 assemble the
  desired maximum weight $\fs v.$-factor from its
 subgraphs that lie in the various nodes of the weighted blossom tree.
 The details are similar to Section \ref{GeneralbMatchingSec}, and
 it is implicit in \cite{G} that these details work. But for completeness
 the next lemma proves the necessary properties.

 Any blossom $B$ in a \W.-tree has an associated graph $\o G(B)$.
 Its vertices are
 the children of $B$, with every blossom child
 contracted.
 Its edges are the underrated edges that join any two of its vertices
 (this includes
 underrated loops).
 Informally the lemma states that the edges of any
 maximum $\fs v.$-perturbation
 in $\o G(B)$
 satisfy all relevant constraints.

 \begin{lemma}
 \label{GbarLemma}
 For
  any $v\in V(G)$,
 any perturbation $\fs v.$,
 and any $B$ in \W.,
 let $\o B \con \delta(B)$ be a set of edges that respects $B$.
 If $\o B$ has the form $I\oplus e$ assume $e$ is tight.
 Then $\o G(B)\cup \o B$ has a subgraph
 $F$ containing $\o B$
 wherein

 (a) every child of $B$ that is a vertex $x\in V(G)$
 has $\fs v.(x)=d(x,F)$;

 (b) $F$ respects every child of $B$ that is a
 contracted blossom.
 \end{lemma}

 \begin{proof}
 The two possible forms for $\o B$ are $I(B)$ and $I(B)\oplus e$.
 We first show that $I(B)$ is unique and the lemma holds when
 $\o B=I(B)$.

 For any $v\in V(B)$,
 any maximum
 $\fs. v$-factor $F$ respects $B$, i.e., it
  has
 $F\cap \delta(B)=I(B)$.
 (This follows from the optimality of the duals.)
 So $I(B)$ is unique.
 Also this factor $F$ proves the lemma
 when $\o B=I(B)$, i.e., $F$ satisfies (a) and (b).

 For $v\notin V(B)$ any set respecting $B$ has the form
 $I(B)\oplus e$, $e\in \delta(B)$. By assumption $e$ is tight.
 %Furthermore $e$ must be tight, since it belongs to $F_x\oplus \F s. v$.
 Let
 $e=xy$ with  $y\in V(B)$.
 If $e\in I(B)$ then $F=F^y -xy$ is a maximum $f_x$-perturbation,
 with $F\cap \delta(B)=I-e$.
 (This follows from the optimality criterion of the duals.)
 Similarly if $e\in \delta(B)-I(B)$ then $F=F_y + xy$ is a maximum
 $f^x$-perturbation with  $F\cap \delta(B)=I+e$.
 As before, the maximum factor $F$ proves the lemma.
 \end{proof}

 \subsection{Finding the {\boldmath {$V(B)$}}-sets}
 \label{VBSec}
 The reduction is given
 the quantities, $w(F_v), w(F^v)$, $v\in V(G)$.
 So it  can use these quantities:

 \bigskip

 For each vertex $v\in V$, $\zeta_{v} = w(F_v)+w(F^v)$.

 For each edge $uv\in E$, $\zeta_{uv} = w(F_u)+w(F_v)+w(uv)$ and
 $\zeta^{uv} = w(F^u)+w(F^v)-w(uv)$.

 \bigskip

 For a vertex $v$ let $UNI(v)$ be the maximum $2f$-unifactor containing
 $v$ in its circuit.
 %, and let $B_v$ be the smallest blossom containing
 %$v$ in its circuit.
 Lemma \ref{UnifactorLemma} shows
 $\zeta_v=w(UNI(v))$.  Recalling the definition of dual variables
 (\ref{DualDefnEqn}) we get
 \[\zeta_v= \bar z (B_v).\]

 For any edge $uv$ of $G$ let $UNI(uv)$ be the maximum $2f$-unifactor
 containing $uv$ in its circuit.  We classify $uv$ as type 0, 1 or 2
 depending on the number of sets $F_u, F_v$ that contain
 $uv$. Specifically $uv\in E$ is

 \bigskip

 {\em type 0} if $uv \notin F_u\cup F_v$;

 {\em type 2} if $uv \in F_u\cap F_v$;

 {\em type 1} if $uv \in F_u \oplus F_v$.  Additionally a type 1 edge
 $uv$ is {\em type $1_\zeta$} if $\zeta_u, \zeta_v >\zeta$.
 %$\zeta_u, \zeta_v \ge \zeta$.

 \bigskip

 \noindent
 Of course the type of an edge is  unknown to the
 reduction!
 Lemma \ref{UnifactorLemma} extends to give the following
 combinatoric
 interpretations of $\zeta_{uv}$ and $\zeta^{uv}$.
 In contrast to $\zeta_v$ these interpretations are also
 unknown to the reduction.

 \begin{lemma}
 \label{NUnifactorCor}
 Consider any edge $uv$ of $G$.% and any maximum factors $F_u, F_v$.

 \i If $uv$ is type 0 then $UNI(uv)=F_u + F_v+uv$.  Thus $w(UNI(uv))
 =\zeta_{uv}$.  Furthermore $\zeta_{uv}\le \zeta_u, \zeta_v \le
 \zeta^{uv}$.

 \ii If $uv $ is type 2 then $UNI(uv) =F^u + F^v-uv$.  Thus
 $w(UNI(uv))= \zeta^{uv}$.  Furthermore $\zeta^{uv}\le \zeta_u,
 \zeta_v\le \zeta_{uv}$.

 %CHECK iii!!!!!!!!!!!!!!!!!!!!!!!!!!!!!!!!!!!!!!!!!!!!!!!!!!!!!!!!
 \iii Suppose $uv$ is type 1 with  $uv \in F_v -F_u$.
 %then %$F_u=F^v-uv$.Furthermore
 Then $UNI(v)=F_u + F_v+uv$ and $uv$ is incident to the circuit of
 $UNI(v)$. So $w(UNI(v))=\zeta_v =\zeta_{uv}$.  Furthermore
 $\zeta^{uv}=\zeta_u$.
 \end{lemma}

 \begin{proof}
 Lemma \ref{UnifactorLemma} gives the characterization of the various
 unifactors. The relations between the various $\zeta$ quantities all
 follow easily from this simple identity: Any edge $uv\in E$ satisfies
 \begin{equation}
 \label{ZetaEqn}
 \zeta_{uv} +\zeta^{uv} =\zeta_{u} +\zeta_{v}.
 \end{equation}

 \noindent
 For instance to prove part \xi, the relation $w(U_{uv})\le w(U_u), w(U_v)$
 translates to $\zeta_{uv} \le \zeta_u, \zeta_v$. Now (\ref{ZetaEqn}) implies
  $\zeta_u, \zeta_v\le \zeta^{uv}$.
 \end{proof}

 We define the graph $G(\zeta)$, for any real value $\zeta$: Its
 vertices are the vertices of $G$ with all blossoms of $\bar z$-value
 $>\zeta$ contracted. Its edge set is
 \[E(G(\zeta))=\set{uv} {uv \in E(G), \ \min \{\zeta_{uv},\ \zeta^{uv}\}
 \ge \zeta}.\]
   The reduction can
 construct $G(\zeta)$, since previous iterations  have
 identified the blossoms of $\bar z$-value $>\zeta$.

 Recall a blossom $B$ of \B. is constructed as an elementary blossom
 in a contracted graph $\o G$. As such it has a circuit $C(B)$.
 $C(B)$ consists of edges that are images of edges of $G$, as well
 as blossom loops (resulting from contractions). In the
 lemma statement below, $E(C(B))\cap E(G)$ denotes
 the edges of $G$ whose images belong to $C(B)$.

%%%%%%%%%%%%%%%%Q: could C(B) be in an earlier graph?

 \begin{lemma}
 \label{BlossomLemma}
 For any blossom $B$ of \B., $E(C(B)) \cap E(G) \con E(G(\bar z (B)))$.
 \end{lemma}

 \begin{proof}
 Consider any edge $uv \in C(B)$.
 The desired relation $uv \in   E(G(\bar z (B)))$ amounts to
 the inequality
 $\min \{\zeta_{uv},\ \zeta^{uv}\} \ge \bar z (B)$. (Note that
 $uv$ cannot be a loop in $G(\bar z (B))$ since $uv$ joins
 distinct blossoms of $C(B)$.)
 We will prove the  desired inequality using the fact
 that
  every edge $uv \in C(B)$ is tight, proved in
 \AfTheorem.%
 \footnote{In matching and $b$-matching tightness is forced by the fact that
 every edge of a blossom circuit is in some maximum perturbation.
 The analogous statement fails for  $f$-factors --
 an ''exceptional'' circuit edge  may
 belong to
 no
 maximum perturbation at all or to every maximum perturbation \cite{G}.}
 Consider three cases:

 \case{$uv$ is type 0}
 Tightness means $y(u)+y(v)+\bar z (B) = w(uv)$.
 Equivalently $\bar z (B) = w(F_u)+w(F_v) +w(uv)$. Thus
 $\bar z (B) = \zeta_{uv}$ and Lemma
 \ref{NUnifactorCor}\i shows
 $uv \in G(\bar z(B))$.

 \case{$uv$ is type 2} Tightness means $y(u)+y(v)+\bar z (B_u)+
 \bar z (B_v)-\bar z(B)
  = w(uv)$.
 As noted above
 \[w(F_v)+w(F^v)= \zeta_v=\bar z(B_v)\]
  and
 similarly for $u$. Sustituting this relation gives
  $w(F^u)+w(F^v) -w(uv)= \bar z(B)$.
 Thus
 $\bar z (B) = \zeta^{uv}$ and Lemma
 \ref{NUnifactorCor}\ii shows
 $uv \in G(\bar z(B))$.

 \case{$uv$ is type 1}
 Wlog $uv\in F_v-F_u$. Tightness means
 $y(u)+y(v)+\bar z (B_v) = w(uv)$.
 Equivalently
 \[
 \bar z (B_v) = w(F_u)+w(F_v)+w(uv).\]
 With
 $\bar z (B_v) \ge \bar z (B)$ (since $B$ is an ancestor of $B_v$) this gives
 $\zeta_{uv}\ge \bar z (B)$.

 The (last) displayed equation is equivalent to
 $w(F_v)+w(F^v)=  w(F_u)+w(F_v)+w(uv)$. Rearranging gives
 $w(F^u)+ w(F^v)-w(uv)=  w(F^u)+w(F_u)=\bar z(B_u) \ge \bar z (B)$,
 i.e., $\zeta^{uv}\ge \bar z (B)$.
 \end{proof}

 In $G(\zeta)$ every type 0 edge is in a blossom circuit (Lemma
 \ref{NUnifactorCor}\xi, which refers to the given graph, and
 \AWeightLemma,
 which shows the relation of Lemma
 \ref{NUnifactorCor}\i is preserved as blossoms are contracted).  The
 same holds for every type 2 edge.  (However note that an arbitrary
 type 0 or 2 edge needn't belong to a blossom circuit -- it may not
 appear in any $G(\zeta)$ graph because of blossom contractions.)

 In contrast a type 1 edge of $G(\zeta)$ may or may not be in a blossom
 circuit.  The $1_\zeta$ edges obey the following
 generalization of Corollary \ref{ForestCor} for
  shortest paths.

 \begin{lemma}
 \label{AcyclicLemma}
 In any graph $G(\zeta)$ the $1_\zeta$ edges are acyclic.
 \end{lemma}

 \begin{proof}
 We start with  a relation between the set $I(A)$ of a blossom $A\in \B.$
 and the same set when $A$ is contracted, i.e., set $I(a)$
 for  blossom vertex
 $a$.
 In $G(\zeta)$ suppose $a$ is a blossom vertex and vertex $b\ne a$.
  Recall that
 $F_a\cap \delta(a)=I(a)$ and $F_b$ respects blossom $a$ (\ARespectCor,
 which says that as expected, every
 maximum perturbation  \Fs v. respects every maximum proper
 unifactor's blossom).
  Thus
 \begin{equation}
 \label{IaEqn}
 |(F_a \cap\delta(a)) \oplus (F_b\cap \delta(a))|=
 |I(a) \oplus (F_b\cap \delta(a))|=
 1.
 \end{equation}

 We use (\ref{IaEqn}) to prove the following:

 \xclaim {In $G(\zeta)$ consider a blossom vertex $c$ and distinct edges $e,f\in
   \delta(c)$.  Suppose vertices $b,c,c'$ have $e\in F_b\oplus F_c$ and
   $f\in F_c\oplus F_{c'}$.  Then $f\in F_b\oplus F_{c'}$.  }

 \bproof
 We can assume $b\ne c$ since otherwise the lemma is tautologous.
 Since $F_b$ respects blossom $c$ and $e,f \in \delta(c)$,
 $e\in F_b\oplus F_c$ implies $f\notin  F_b\oplus F_c$
 (by (\ref{IaEqn})).
 Combining with  the hypothesis $f\in F_c\oplus F_{c'}$
 gives $f\in   ( F_b\oplus F_c) \oplus (F_c\oplus F_{c'})=
  F_b\oplus F_{c'}$.
 \ecproof

 Now consider a cycle of $1_\zeta$ edges in $G(\zeta)$, say
 $a,b,\ldots, y, z$ with $z=a$.  Type 1 means $ab\in F_a\oplus F_b$.
 This immediately shows the cycle has $\ge 2$ edges.
 We will show edge $yz=ya$ also belongs to $F_a\oplus F_b$.
 $ya$ may be parallel to $ab$ (a length 2 cycle) but
 $ya$ is not the same edge as $ab$ (i.e., we do not have a length
 1 cycle, since a loop $aa$ is not type 1).
 So we get 2 distinct edges in  $\delta(a)\cap (F_a\oplus F_b)$.
 This contradicts (\ref{IaEqn}).

 It is convenient to also denote the cycle as $a, c^0=b, c^1,\ldots,
 c^{r-1}=y, c^r=z$.  Type 1 means $c^{i-1}c^{i}\in F_{c^{i-1}}\oplus
 F_{c^{i}}$.  Inductively assume $c^{i-1}c^i\in F_b\oplus F_{c^{i}}$.
 (This holds for $i=1$.)  Since $c^{i}c^{i+1}\in F_{c^{i}}\oplus
 F_{c^{i+1}}$ the claim (with $c=c^i$, $c'=c^{i+1}$) shows $c^ic^{i+1} \in F_b\oplus F_{c^{i+1}}$.
 Thus induction shows $c^{r-1}c^{r} \in F_b\oplus F_{c^{r}}$, i.e.,
 $yz\in F_b\oplus F_z$, as desired.
 \end{proof}

 The reduction processes the graphs $G(\zeta)$ for
 $\zeta$ taking on the distinct values in
  \[
 \Omega=\set{\,\min \{\zeta_{uv},\,\zeta^{uv}\}\,} {uv \text { an  edge of $G$}}\]
 in descending order.
 For any $\zeta\in \Omega$ let
 $\zeta^-$
 be any value strictly between $\zeta$ and
 the next largest  value in $\Omega$.
 Observe that $G(\zeta^-)$ is the graph $G(\zeta)$ with every blossom
 of $\bar z$-value $\ge \zeta$ contracted.

 \begin{corollary}
 \label{AcyclicCor}
 For any $\zeta\in \Omega$, $G(\zeta^-)$ is a forest.
 \end{corollary}

 \begin{proof}
 In $G(\zeta)$ every type 0 or 2 edge is in a blossom
 circuit (as indicated after Lemma \ref{BlossomLemma}).
 So
 $G(\zeta^-)$, which has all these blossoms contracted, has
 only type 1 edges $uv$.
 Furthermore since $\min \{\zeta_{uv},\ \zeta^{uv}\} \ge \zeta$,  Lemma
 \ref{NUnifactorCor}\iii shows $\zeta_u,\zeta_v\ge \zeta$.  In other
 words $uv$ is type $1_{\zeta^-}$.
 Thus
 Lemma
 \ref{AcyclicLemma} shows $G(\zeta^-)$ is
 acyclic.
 \end{proof}

 The next lemma shows how the reduction finds the vertex sets of the
  blossoms of \W. that have $\bar z$-value $\zeta$.

 \begin{lemma}
 As vertex sets, the 2-edge-connected components of $G(\zeta)$
 are precisely the weighted blossoms of $\bar z$-value $\zeta$.
 \end{lemma}

 \begin{proof}
 We first show that each blossom of \B. is 2-edge-connected. More precisely
 let $B$ be a node of \B. with $\bar z(B)=\zeta$.
 %The edges of $C(B)$ belong to $G(\zeta)$ by Lemma \ref{BlossomLemma}.

 \xclaim {In $G(\zeta)$ the subgraph of edges $\bigcup \set {C(A)}
   {\text{node $A$ of \B.  descends from $B$ and } \bar z (A) = \zeta}$
   is 2-edge-connected.}

 \bproof
 The argument is by induction on the number of descendants $A$.  Recall
 \B.
 is
  constructed by repeatedly finding the next elementary blossom
 $B$ and contracting it.

 When $B$ is found, each vertex of its circuit is either \i an original
 vertex of $G$, or \ii a contracted blossom of $\bar z$-value $>\zeta$,
 or \iii a contracted blossom of $\bar z$-value $\zeta$. (Recall the
 definition of $z$, (\ref{DualDefnEqn}). The
 type \iii blossoms are blossoms with $z$-value 0.)  Vertices of
 type \i or \ii are vertices of $G(\zeta)$.  Vertices of type \iii have
 2-edge-connected subgraphs in $G(\zeta)$ by induction.  Each original
 edge of $C(B)$ is contained in $G(\zeta)$ (Lemma
 \ref{BlossomLemma}). Since $C(B)$ is a circuit when $B$ is formed, it
 completes a 2-edge-connected subgraph of $G(\zeta)$.  This completes the
 induction.
 \ecproof

 Now starting with $G(\zeta)$, contract each of the above
 2-edge-connected subgraphs that corresponds to a maximal blossom of $\bar
 z$-value $\zeta$. We get the  acyclic graph   $G(\zeta^-)$ (Corollary
 \ref{AcyclicCor}).
 So the contracted subgraphs are precisely the
 2-edge-connected components of $G(\zeta)$.
 \end{proof}

 In summary we find all the $V(B)$-sets as follows.
 \bigskip

 {\narrower

 {\parindent=0pt

 Compute all values $\zeta_{uv},\zeta^{uv}$, $uv$ an edge of $G$.
 Then repeat the following step for $\zeta$ taking on the distinct
 values in $\Omega$
 %$\min \{\zeta_{uv},\zeta^{uv}\}, {uv\in E(G)}$
 in decreasing order:

 {\parindent=30pt

 \narrower

 {\parindent=0pt

 Construct $G(\zeta)$, contracting all $V(B)$-sets
 of blossoms of $\bar z$-value $>\zeta$.
 Find the  2-edge-connected components of $G(\zeta)$
 and output them as the $V(B)$-sets of blossoms of $\bar z$-value $\zeta$.

 }}

 In addition, output the graph $G(\zeta^-)$ for the final value of $\zeta$.

 }}

 \bigskip

 It is easy to modify the output to get most of the
 weighted blossom forest
 (\xi--\iv of Sec.\ref{BackgroundSec}) for $z$:

 \i The
 weighted blossom trees \W. are constructed from the
 containment relation for the 2-edge-connected components.

 \ii $z(V)$ is
 the final value of $\zeta$. Consider
 a weighted blossom $B\ne V$. Let it be formed in the graph $G(\zeta)$.
 Then
 $z(B)= \zeta -\zeta'$,
 where $G(\zeta')$
 is the graph in which the parent of $B$ (in \W.) is formed,
 or if $B\in \M.-V$,
 $\zeta'=z(V)$.
 Clearly $z(B)>0$.

 \iii The  $I(B)$-sets are computed in the next section.

 \iv
 The tree \T. is the final graph $G(\zeta^-)$.
 In proof $G(\zeta^-)$ is a forest (Corollary \ref{AcyclicCor}).
 Its vertices are
 the contractions of the blossoms of \M.,
 since every vertex
 in a critical graph is in a blossom
 (Lemma \ref{UnifactorLemma}).
 $G(\zeta^-)$ is a tree since
 a critical graph is connected.
 Finally
 every edge $AB$ of $G(\zeta^-)$
  belongs
 to $I(A)\oplus I(B)$
 since it is type 1 (and
 any  blossom $C$ has $I(C)= F_C \cap\delta(C)$).

 \bigskip

 We conclude the section by  estimating the time for this procedure.
 The
 values in $\Omega$ are sorted into decreasing order in time $O(m\log n)$.
 We use a set-merging algorithm
 to keep track of two partitions of $V(G)$, namely the connected components of
 the $G(\zeta)$ graphs, and the 2-edge connected components.
 The total time for set merging is  $O(m+n^2)$.

 The total time for the rest of the procedure is $O(m+n^2)$.
 In proof, there are $\le m$ iterations ($|\Omega|\le m$).
 An iteration that does not change either partition (because its new edges
 are contracted) uses $O(1)$ time for each new edge.
 There are $\le 2n$ other iterations (since each of these iterations decreases
 the number of connected components or 2-edge connected components).
 Each such iteration
  uses linear time, i.e., $O(1)$ time per vertex or edge
 of $G(\zeta)$.
  We complete
 the proof by showing that $<n$ edges belong
 to $>1$ graph $G(\zeta)$.

 A type 0 or 2 edge is in one graph $G(\zeta)$.  A type 1 edge first
 appears in $G(\zeta)$ for $\zeta=\min\{\zeta_u, \zeta_v\}$.  Either it
 gets contracted in this iteration (and so is in just one $G(\zeta)$
 graph) or it is a $1_\zeta$ edge in any future $G(\zeta)$ graph that
 contains it. Any $G(\zeta)$ has $<n$ $1_\zeta$ edges (Lemma
 \ref{AcyclicLemma}).

 \subsection{Finding the {\boldmath {$I(B)$}}-sets}
\label{IBSec}

 A {\em $\zeta$-blossom}
 is a blossom of \W.
 with $\bar z$-value $\zeta$.
 The iteration of the reduction for $\zeta$
 finds $I(B)$ for each $\zeta$-blossom
 $B$.
 For any  $y\in V(B)$ define
 the set
 \[I(y)= I(B)\cap \delta(y).\]
 %(Recall $B_y$ denotes the smallest blossom of \B. containing $y$.)
 So any blossom $B$ has $I(B)=\bigcup \set {I(y)} {y\in V(B)}$
 (this is part of the definition of blossoms
 \BTreeDfn; it also  easily follows  from
 \iv of
 the definition of the weighted blossom tree,
 since every maximum \Fs v. respects every blossom).
 The iteration for $\zeta$  computes the sets $I(y)$
 for all vertices $y$ with
 $\zeta_y =\zeta$.
 Clearly we can combine these sets with sets $I(B')$
 known from previous iterations (for values $\zeta'>\zeta$)
 to find the $I$-set of each $\zeta$-blossom.

 We compute the $I(y)$-sets in two steps.
 For $y$ and $B$ as above (i.e.,
 $B$ a $\zeta$-blossom, $y\in V(B)$,
 $\zeta_y =\zeta$) define the set
 \[I_0(y)=
 \set {xy} {xy\in \delta(V(B)),\
 \zeta_y>\zeta^{xy} \text{ or }
 %\zeta_y=\zeta_{xy}<\zeta_x}.\]
 \zeta_y=\zeta_{xy}\ne \zeta_x}.\]

 \begin{lemma}
 \label{IySubsetLemma}
 $I_0 (y)\con I(y)$.
 \end{lemma}

 \begin{proof}
 Consider an edge $xy\in I_o(y)$.
 Since $xy\in\delta(V(B))$ we have $xy \in I(y)$ iff
 $xy\in F_y$ ($F_y$ respects $B$).  The latter certainly
 holds if $xy$ is type 2.  So it suffices to show $xy$ is not type 0,
 and $xy\in F_y$ if $xy$ is type 1.

 Suppose $xy$ is  type 0.  Lemma \ref{NUnifactorCor}\i shows
 $\zeta_y\le \zeta^{xy}$, i.e., the first alternative in the set-former
 for
 $I_0(y)$ does not hold.  If the second alternative holds we have
 $\zeta=\zeta_y=\zeta_{xy}$, so  Lemma \ref{NUnifactorCor}\i  implies $x\in V(B)$,
 contradicting $xy\in\delta(V(B))$.

 Suppose $xy$ is type 1.
 The first  alternative in the set-former
 for $I_0(y)$ implies
 $\zeta_y\ne \zeta^{xy}$ and the
 second alternative has  $\zeta_x\ne \zeta_{xy}$.
 Both relations imply $\zeta_y= \zeta_{xy}$ and
 $xy\in F_y$ (Lemma \ref{NUnifactorCor}\xiii).
 \end{proof}

 To find the remaining $I$-edges for $G(\zeta)$,
 define
 the set
 \[
  IE=\set{uv} {uv \text{ joins distinct $\zeta$-blossoms, }
 \zeta=\zeta_u=\zeta_v=\zeta_{uv}}.
 %type 1$_\eta$ with }\zeta=\zeta_u=\zeta_v=\zeta_{uv}}.
 \]

 \begin{lemma}
 \label{IELemma}
 Any vertex $y$ with $\zeta_y =\zeta$ has
 $I(y)-I_0(y)\con IE.$
 \end{lemma}

 \begin{proof}
 Let $B$ be the $\zeta$-blossom containing $y$.
 Suppose $xy\in I(y)$. Thus $xy\in F_y$, making
 $xy$  type 1 or 2. If type 2, Lemma \ref{NUnifactorCor}\ii
 shows  $\zeta^{xy}\le \zeta_y$. Furthermore with $\zeta_y=\zeta$
 and $x\notin V(B)$
 it shows the inequality
 is strict. Thus $xy\in I_0(y)$.
 If type 1,  Lemma \ref{NUnifactorCor}\iii shows $\zeta_y=\zeta_{xy}$.
 If $\zeta_x \ne \zeta_y$
 then $xy\in I_0(y)$.
 If  $\zeta_x =\zeta_y$
 then $xy\in IE$, since $xy\in \delta(V(B))$ shows $x$ and $y$
 are in different
 $\zeta$-blossoms.
 \end{proof}

  $IE$ is a subgraph of $G(\zeta^-)$ ($uv\in IE$
 has $\zeta^{uv}= \zeta_{uv}=\zeta$ by (\ref{ZetaEqn})).
 Now observe that
 every nonisolated vertex of $G(\zeta^-)$
 is a contracted blossom.
 (In proof,
 any edge $xy$ of $G(\zeta^-)$ belongs to $G(\zeta)$,
 so $\min\{\zeta_{xy},\zeta^{xy}\}\ge \zeta$. Furthermore
 $xy$ is type 1, so Lemma \ref{NUnifactorCor}\iii shows
 $\min\{\zeta_x,\zeta_y\}=\min\{\zeta_{xy},\zeta^{xy}\}$.
 Thus
 $\min\{\zeta_x,\zeta_y\}\ge \zeta$.)
 We conclude that  every edge
 $xy$ of $G(\zeta^-)$ has $xy\in I(x)\oplus I(y)$ (since every edge of
 $G(\zeta^-)$ is
 type 1).

 So to complete the computation of the $I(B)$-sets we need only
 decide which alternative ($xy\in I(x)$ or $xy\in I(y)$) holds
 for each edge $xy$ of $G(\zeta^-)$. (Note that
 an edge $xy$ of $G(\zeta^-)$ needn't be
 in $IE$ -- the blossoms containing $x$ and $y$ may not be $\zeta$-blossoms.
 But this causes no harm.) We accomplish this classification using the
 acyclicity of  $G(\zeta^-)$, as follows.

 Let $T$ be a nontrivial tree of  $G(\zeta^-)$.
 Let $B$ be a leaf of $T$,
 incident to edge $xy$ of $T$ with $y\in B$.
 Since $F_y$ respects $B$, $f_y(B)+|I(B)|$ is even (recall the definition
 of respect). Thus

 \noindent
 \hbox to \textwidth{$(*)$\hfill $xy\in I(B)$ iff $f_y(B) + |I(B)-xy|$
   is odd.\hfill}

 \noindent
 All edges of
  $I(B)-xy$ are known. This follows from Lemma \ref{IELemma}
 if $\zeta=\zeta_y$ (recall $I(B)-IE$ is known).
 If  $\zeta<\zeta_y$ it holds since all of $I(B)$ is actually known.
 So $(*)$ can be used to add $xy$ to exactly one of the
 sets $I(x)$ or $I(y)$.
 Thus the following algorithm correctly classifies each edge
 of $G(\zeta^-)$.

 \bigskip

 {\narrower

 {\parindent=0pt

 In the (current) forest $G(\zeta^-)$, let $B$ be a leaf of a nontrivial
 tree $T$.  Halt if no such $T$ exists. Let $xy$,
 $y\in B$, be the unique edge of $T$ incident to $B$.  Assign $xy$ to
 exactly one of $I(x)$ or $I(y)$, using $(*)$. Then delete $B$ from $T$
 and repeat.

 }}

 \bigskip

 The total time to compute $I(B)$-sets is $O(m+n^2)$.
 Specifically, a set $I_0(y)$ is computed in the iteration
 where $\zeta=\zeta_y$,
  using $O(1)$ time on each edge incident to $y$.
 The algorithm for processing $IE$ is executed in each iteration
 where some $y$ has $\zeta_y=\zeta$. Each such execution
 uses $O(n)$ time, giving time $O(n^2)$ in  total.

 \subsection{Finding a maximum factor}
 \label{TopDownAlgSec}

 Assume  we are given the dual functions
 $y$, and $z$ in the form of  its weighted blossom forest.
 We show how to find a
 maximum lower or upper perturbation $\Fs v.$
 for any given $v\in V$,
 in total time $O(\phi^\omega)$, with high probability.
 Let $F$ be the desired maximum perturbation ($F=F_v$ or $F=F^v$).

 The procedure is in three steps.
 It halts with the set $F_0$ equal to $F$.
 Initially $F_0$ is empty and we add edges to $F_0$ as they become known.

 \iffalse
 We shall add edges known to be in $F$ to a set $F_0$.
 Let $r$ denote  the residual degree constraint function, i.e.,
 \[r(x)= \fs v. (x)-d(x,F_0).\]
 Initially $F_0=\emptyset$ and $r=\fs v.$.
 \fi

 \subsubsection*{Edges of \T.}

 The first step determines $F\cap \T.$
 and adds  these edges to $F_0$.
 Consider any edge $e$ of $\T.$. Let $X$ and $X'$ be
 the sets of the partition of $V(G)$
 induced by the
 connected components of $\T.-e$.
 Since $F$ is an $\fs v.$-factor,
 $f_v(X)$ counts every edge of $F\cap \gamma(X)$ twice. So

 \noindent
 \hbox to \textwidth{\hfill $e\in F$ iff $f_v(X)$
   is odd.\hfill}

 Clearly we can implement this test to find all edges of
  $F\cap \T.$  in time $O(n^2)$.
 % (or $O(n)$ with greater care).

 \subsubsection*{Strictly underrated edges}

 The second step calculates $\H{yz}(e)$ for each edge $e\in
 E(G)$. If $e$ is strictly underrated it is added
 to $F_0$.

 To do the calculation efficiently
 assume each blossom is labelled with its $\bar z$-value.
 Let  $e=uv$.
 Let
 $B$ be the nearest common ancestor of $u$ and $v$ in the blossom
 tree containing $u$ and $v$.
 Then
 \[\H{yz}(uv)= y(u) + y(v)+
 \begin{cases}
 \bar z(B)&uv \notin I(B_u)\cup I(B_v)\\
 \bar z(B_u)&uv \in I(B_u)- I(B_v)\\
 \bar z(B_u) +\bar z(B_v)- \bar z(B)&uv \in I(B_u)\cap I(B_v).
 \end{cases}
 \]

 This step uses total time $O(m)$.
 (Nearest common ancestors are found in $O(1)$ time.)

 It is convenient to ignore these strictly underrated edges in the rest
 of the discussion. So assume
 %these edges have been saved and
 the degree-constraint
 function $\fs v.$ has been decreased to account for
 the strictly underrated edges, i.e., from now on
 $\fs v.$ denotes the residual degree constraint.

 \subsubsection*{Edges of blossoms}

 We turn to the third step of the procedure.
 As in shortest paths and $b$-matching, we process the weighted blossoms  in a top-down
 fashion. Consider a weighted blossom $B$. Assume the set
 $F\cap \delta(B)$ is known. The first step ensures this
 for a root $B\in \M.$.
 We will find the edges of $F$ that belong to $\o G (B)$
 and  add them  to $F_0$.
 Note that these edges
 complete the sets $F\cap \delta(A)$, $A$ a child of $B$.
 So
 if $A$ is a blossom we can process it  the same way. Thus
  we can process every weighted blossom  this way.

 Form a graph $H$ as
 $\o G (B)$ with
 the strictly underrated edges deleted.
 %\footnote{Note that every edge
 %of $I(B)$ is tight, Theorem \ref{}.}
 Define $e\in E(G)\cup \{\emptyset\}$ by the relation
 \[F\cap \delta(B)=I(B)\oplus e.\]
 If $v\notin B$ this defines $e$ as an edge, and
 if $v\in B$ it defines $e$ as $\emptyset$ (recall $F$ respects $B$).
 %Note that when $e$ is an edge, it does not appear in $H$.

 Note that if $e$ is an edge it is tight.
 In proof,
 take any $y\in V(B)$. Then
 $(F\cap \delta(B)) \oplus (F_y\cap \delta(B))= (I(B)\oplus e)) \oplus I(B)=e$.
 So $e\in F\oplus F_y$ implies $e$ is underrated but not
 strictly underrated, whence $e$ is tight.

 \iffalse
 $e\in F$ implies $e$ is underrated (we only add underrated edges
 to $F$). So we need only show any strictly underrated edge $g\in F\cap \delta(B)$ is not $e$,

 This follows since we
 only add underrated edges $g$ to $F$. If $g\in F\cap \delta(B)$
 is  strictly underrated
 $g$ cannot be in $\delta(B)-I(B)$ (such an edge is not
 in $F_y$ for $y\in V(B)$. so such an edge is not strictly
 underrated) So $g\in I(B)$.
 Furthermore $g\in F$ by step two. So $g \ne e$, i.e.,
 $e$ is tight.
 \fi

 %(We delete the strictly underrated edges from $\o G (B)$
 %since they are already in $F_0$.)
 The desired subgraph of $H$ is specified in Lemma \ref{GbarLemma}
 and we find it as follows.
 Let $f'$ be the degree constraint function for $H$.
 Let $x$ be a vertex of $H$.
 If $x$ is a vertex of $G$ then it has degree constraint
 \[
 f'(x)=\fs v.(x)-|\delta(x,F\cap \delta(B))|.
 %\fs v.(x)-|\delta(x,\{e\})|&x\in V(B)\\
 \]

 \iffalse
 \[
 f'(x)=
 \fs v.(x)-|\delta(x, I(B)\oplus e)|&x\in V(B)\\
 |I(B)\cap \delta(x)|&x\in V(G)-V(B).
 \end{cases}
 \]
 \fi

 If $x$ is a contracted blossom $A$ then
 the edges of $F$ that are incident to $x$ are governed
 by
 the fact that $F$ respects $A$. Specifically
 %Lemma \ref{ARespectLemma} shows that
 $F\cap \delta(A)$
 is $I(A)$ if $v\in A$, and $I(A)\oplus g$ for some edge
 $g\in \delta(A)$ if $v\notin A$.

 \ConstructionSec\
  models all the above constraints
 on the desired subgraph  so that it corresponds to a maximum
 weight
 $f$-factor on $H$. It does this by redefining
 the edge weights (in fact it uses
 weights that are much larger than the given ones -- see
 \NewWeights).
 This is inappropriate for the current context, since we wish to
 find the desired subgraph using
 a routine for  unweighted $f$-factors.

 \iffalse
 of Fig.\refs on the subgraph as follows.
 As before
 a vertex $x\in V(G)$ that belongs to $H$
 has degree constraint  $r(x)$.
 Consider a blossom $A$ that is a vertex of $H$.
 If $v \in V(A)$ then
 we add the edges of $I(A)\cap \gamma(V(B))$ to  $F_0$ and delete
 $A$ from $H$ (updating $r$).

 Any other blossom $A$ that is a vertex of
 $H$ is
 represented  using the subgraph of
 \fi

 We model the constraint for $A$ using the blossom substitute of
 Fig.\ref{NSubstituteFig}.
 As indicated, the new vertices $a,c$ and each $a_k$ all have degree
 constraint 1, and
 \[
 f'(b)=\begin{cases}
 0& v \in A \text{ or $e$ an edge in }\delta(A)\\
 1& \text{otherwise.}
 \end{cases}
 \]
 The following claim shows this substitute faithfully models  the constraints
 on $A$.
 Let $S$ be a set of edges
 in the substitute.
 Let
 the images of these edges in $H$ be the set $SB$
 (e.g., $i_ka_k$ in $S$ corresponds to $i_k A$ in $SB$).
 Let
 \[S \o B =F\cap \delta(B) \cap \delta(A)=(I(B)\oplus e) \cap \delta(A).\]

 \xclaim{$S$ satisfies the degree constraints of
 the blossom substitute for $A$ iff
 $SB\cup S\o B $ respects $A$.}

 \bproof
 Consider the two possible values of  $f'(b)$.

 \case{$f'(b)=0$} $S$ satisfies the degree constraints iff it
 consists of edge $ac$ and every edge $a_k i_k$  (but no edge $bj_k$) i.e.,
 $SB=I(A)\cap \gamma(B)$. This is equivalent to
 %\[SB\cup \delta(B,I(A))=SB \cup S \o B=I(A).\]
 \[SB \cup S \o B=(I(A)\oplus e)\cap \delta(A).\]
 If $v\in A$ then
 $e=\emptyset$, and
  the displayed equation becomes
  $SB\cup S\o B=I(A)$ which is equivalent to
 % $SB\cup (\delta(B,I(A))\oplus e)=SB\cup \delta(B,I(A))=I(A)$
  $SB\cup S\o B$ respecting $A$.
 Similarly if  $v\notin A$ and $e\in \delta(A)$
  the displayed equation becomes
 $SB\cup S \o B=I(A) \oplus e$  which is  equivalent to
 $SB\cup S \o B$ respecting $A$.

 \iffalse
 %=SB \cup S \o B=I(A)$ respects $A$.

 $SB\cup (\delta(B,I(A))\oplus e)$
 $CF$ has $CF\cap \delta(A)=I(A)$, as desired.

 %$SB\cup \delta(B,I(A))=SB \cup S \o B$ respects $A$.

 is
 %$SB \cup \delta(B,I(A))+ e=I(A)+e$ if $e\in \delta(A)-I(A)$,
 $SB \cup S \o B+ e=I(A)+e$ if $e\in \delta(A)-I(A)$,
 and it is
 %$SB \cup (\delta(B,I(A)) -e)=I(A)-e$ if $e\in I(A)$.
 $SB \cup S\o B -e=I(A)-e$ if $e\in I(A)$.
 In both cases it respects $A$.
 \fi

 \case{$f'(b)=1$, $v\notin A$, and $e\notin \delta(A)$}
 The assumption
 $e\notin \delta(A)$ implies
 $S \o B = I(A)\cap \delta(B)$.

 \iffalse
 $SB\cup (\delta(B,I(A))\oplus e)=SB\cup \delta(B,I(A)=
 SB\cup S\o B$, i.e., all edges of $H$ chosen by the substitute.
 \fi

 The degree constraint
 $f'(c)=1$ means
 $S$ contains either
 edge $ac$ or $bc$.  The former makes $S$ contain
 each edge $a_ki_k$ and one edge $b j_h$, so
 $SB\cup S\o B=I(A)+ Aj_h$.
 The latter makes $S$ contain
 one edge $a a_h$ and the edges $a_ki_k, k\ne h$,
 so $SB\cup S\o B=I(A)-Ai_h$.
 The assumption  $v\notin A$ in this case shows that
 the two possibilities combined are equivalent to $SB\cup S\o B$
  respecting $A$.
 \ecproof

 %\bigskip

 In summary the algorithm for the third step forms the above variant of $H$
 (note the tight edges have been identified in the second step).
 It finds an $f'$-factor on $H$
  and adds these edges to $F_0$.

 Lemma \ref{GbarLemma} shows that the
 desired factor on $H$ exists. (Note the hypothesis is satisfied, i.e., $e$ is
 tight.) Now
 it is easy to see that the totality of edges added to $F_0$
 (by all three steps) achieves the optimality condition
 given in Section \ref{BackgroundSec}
 for
 the desired maximum weight $\fs v.$-factor.

 {\def\hw{3in}
 \begin{figure}
 \twofigs
 {\hfill
 \subfigure[Contracted blossom $A$ in $\o G(B)$. $I$ and $\bar I$
 denote $I(A)\cap \gamma(B)$ and $\delta(A,H)-I(A)$ respectively.]
 {\epsfig{figure=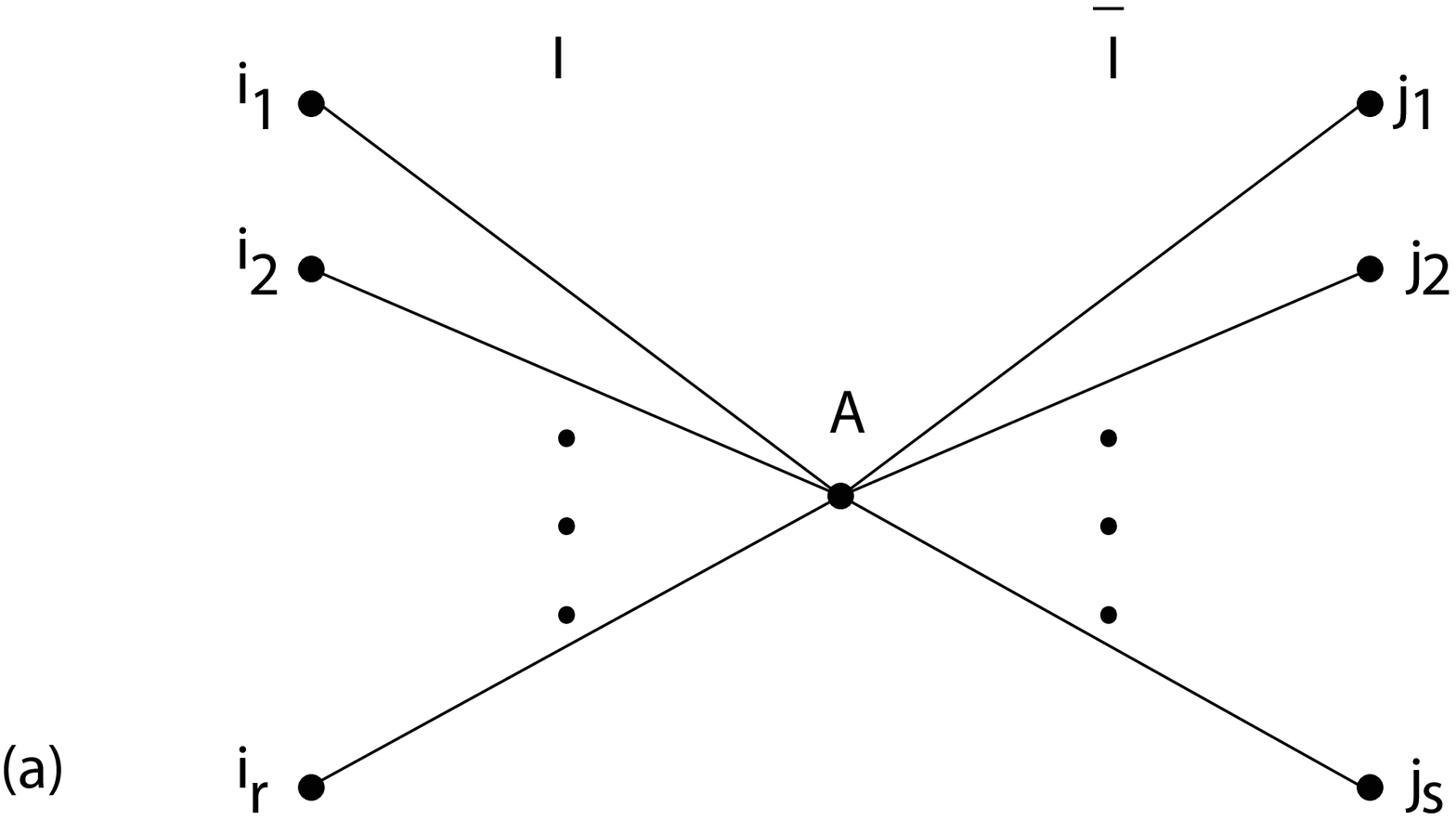, width=2.5in}}\hfill\hfill
 \subfigure[Substitute for $A$: vertices $a,c$, and each $a_k$
 all have $f'$-value 1; $b$ has $f'(b)\in \{0,1\}$.]
 {\epsfig{figure=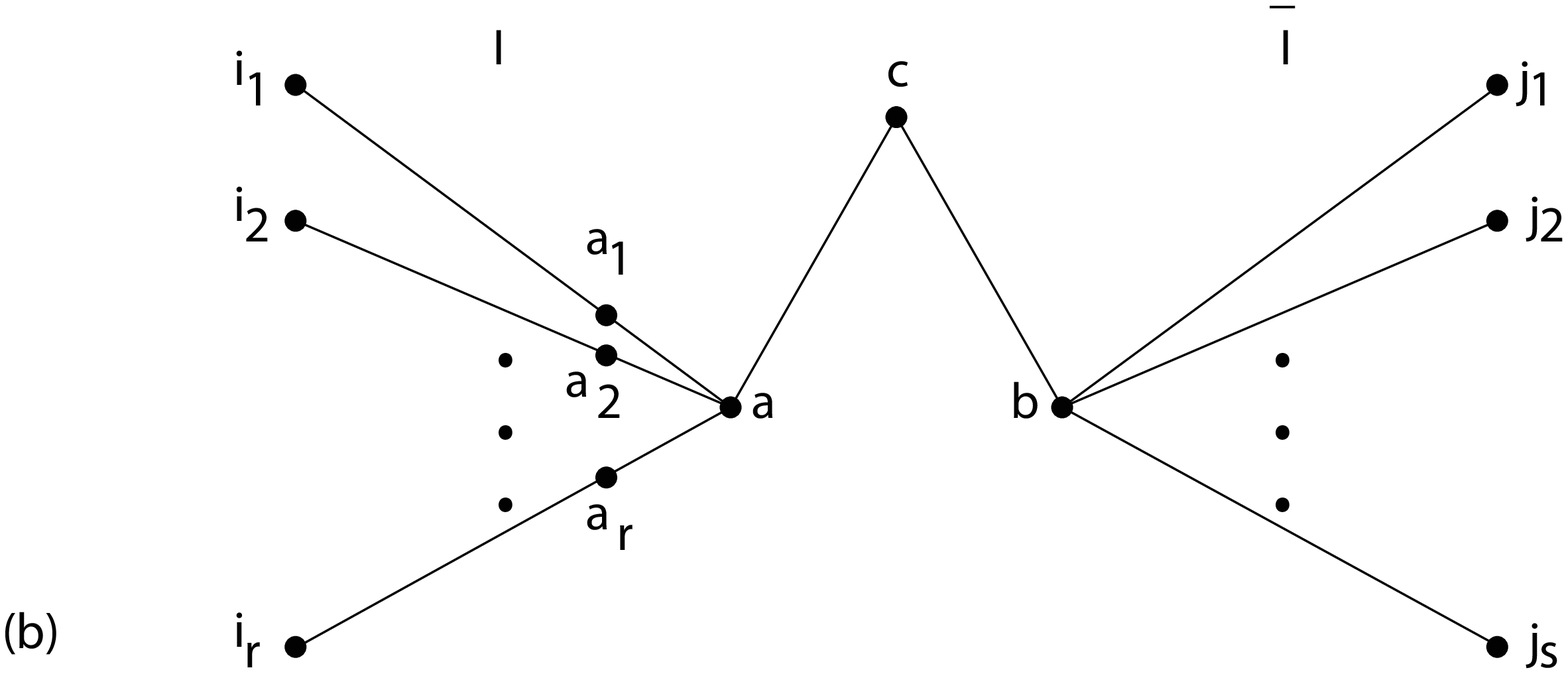, width=3.5in}}
 \hfill}
 \caption{Blossom substitute.}
 \label{NSubstituteFig}
 \end{figure}
 }

 The time for the third step is $O(\phi^\omega)$.
 To prove this it suffices to show that the total number of vertices
 in all
 $H$ graphs is $O(\phi)$.
 There are $O(n)$ weighted blossoms, and hence $O(n)=O(\phi)$ vertices
 of type $a,b$ or $c$ in blossom substitutes.
 An edge $e$ in an $I(B)$-set occurs in only one $H$ graph --
 the graph corresponding to  $\o G (p(B))$, for $B$
 the maximal set with $e\in I(B)$
 and $p(B)$ the parent of $B$ in \W..
 $e$ introduces $1$ extra vertex $a_k$ in $H$.
 %%CHECK !!!!!!!!!!
 A vertex $v$ is on $\le f_v(v)$ edges of sets $I(B)$
 (since these edges are in $F_v$)
 so the total number of edges $e\in
 \bigcup \set{I(B)} {B\in \W.}$ is $\le \phi$.
% \fi %%%%%%%%%%%%%%%%%%%%%%%%%%%%%%%%%%%%%%%%%%%%%%%

%% file: spalg.tex
\iffalse
NO!!!!!!!!!!!!!!!!!!
add
\label{BlossomEdgeAlgSec}
to fred 9.3 ssection
\fi
\section{Shortest-path tree algorithms}
\label{SPAlgSec}

\begin{figure}[t]
\begin{center}
\epsfig{figure=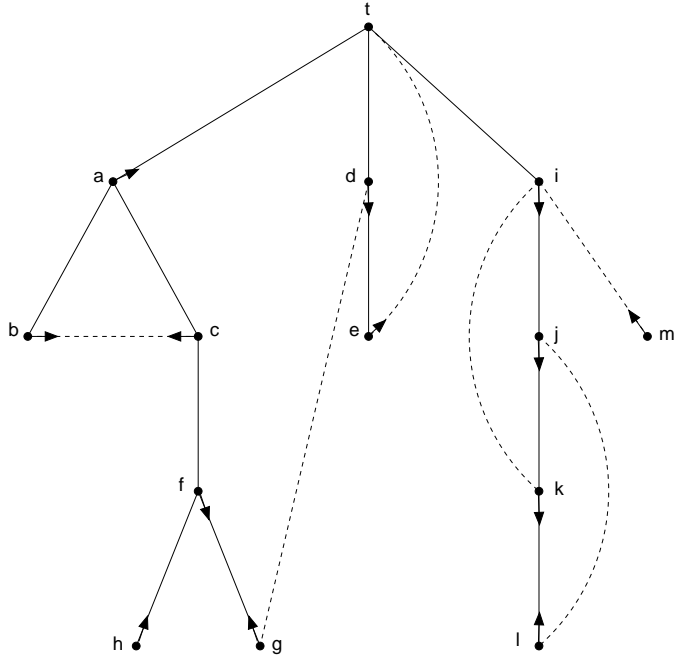, width=3.5in}
\end{center}
\caption{Illustration for permissible paths:
Arrows indicate $e(v)$ edges. Solid
edges belong to a search tree.
The 3 leftmost dashed edges cause contractions;
the 3 rightmost cause no action.}
\label{SAStrucFig}
\end{figure}

\begin{figure}[t]
\begin{center}
\epsfig{figure=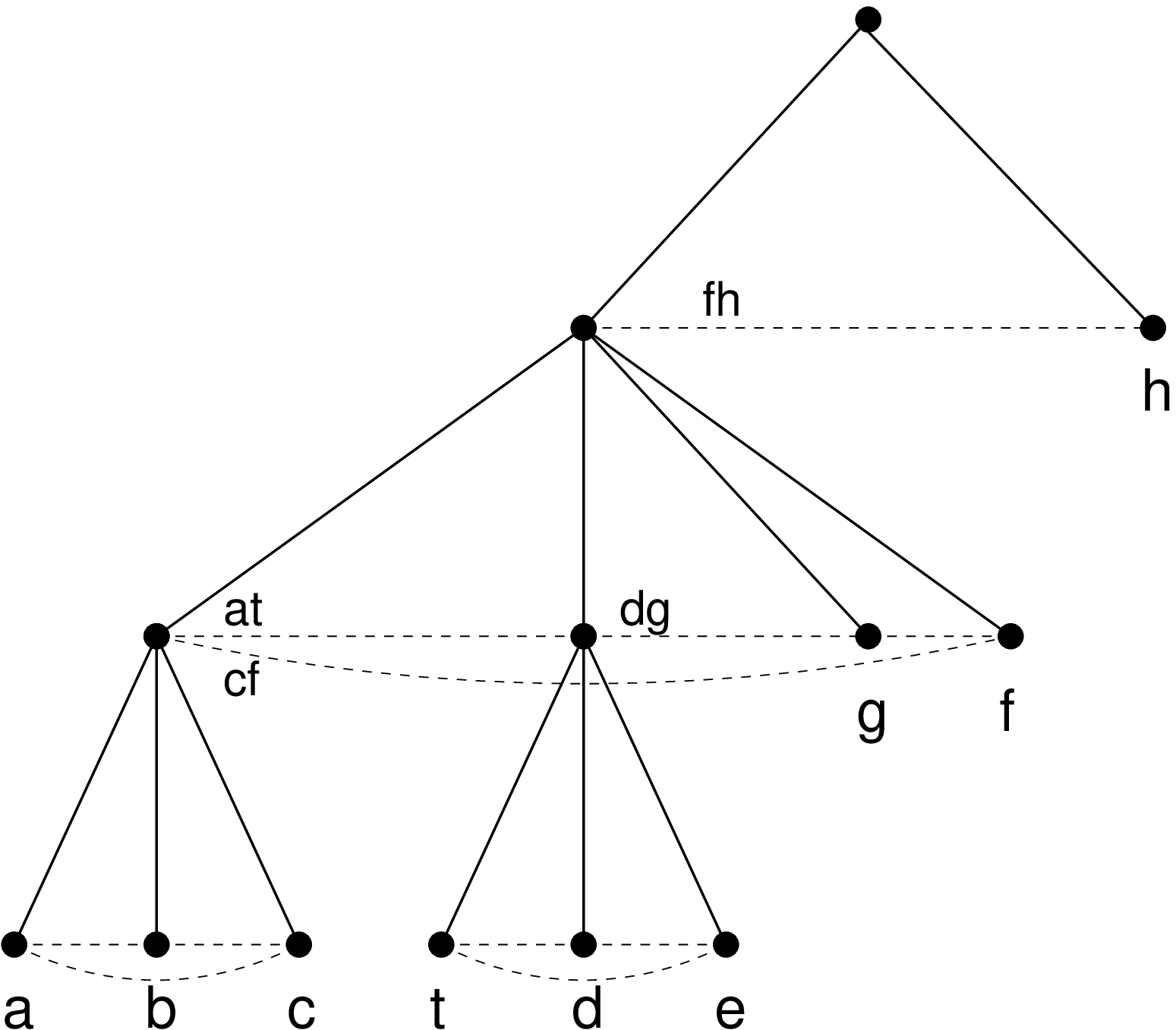, width=3.0in}
\end{center}
\caption{Gsp-tree for vertices $a,\ldots, h$ and $t$ of Fig.\ref{SAStrucFig}.
$e(N)$ edges are the dashed edges
joining the children of $N$.}
\label{SAStrucbFig}
\end{figure}

%%%??????Following Sec.\ref{BlossomEdgeAlgSec},
We construct the gsp-tree
from the weighted blossom forest \W.
found in Section \ref{fFactorSec}. %%%%%%%%%?????????????????????????
This section shows
how to find a
gsp-tree for the graph $\o G(B)$, $B$ a weighted blossom.
It is a simple matter to join these gsp-trees together to get
the entire gsp-tree for the given graph. (Note the
root of the gsp-tree for any blossom $B$ is a cycle node, not a tree node.)
Let us restate the properties of $\o G(B)$.
Every edge is tight.
There is a sink vertex $t$. Some vertices $v$
have a distinguished edge $e(v)\in \delta(v)$, known to be
on the shortest $vt$-path. (These vertices are children
of $B$ that are contracted blossoms, possibly singleton blossoms.) The remaining vertices of $\o G(B)$
have no such edge
(these vertices are children of $B$ that are vertices in the given graph).
We wish to find a gsp-tree for this graph $\o G(B)$.

%%%%%%%%%%????????????????????????????????????MORE DETAILS OF REDUCTION???

We solve a slightly more general problem.
Consider a graph with a sink vertex $t$ and
with every vertex $v$ having a value
$e_0(v)\in \delta(v)\cup \{\emptyset\}$,
$e_0(t)=\emptyset$.
A $vv'$-path $P$ is {\em permissible} if
every vertex $x\in P-v'$ has $e_0(x)\in E(P)\cup \{\emptyset\}$.
In particular $P$  starts with $e_0(v)$ if it is nonnull.
Observe that any $vt$-path specified by a gsp-tree
is permissible for $e$. %A {\em permissible path tree}
We will present an algorithm that, for
a  given $G,t,e_0$, finds a gsp-tree
\T. whose function $e$ agrees with $e_0$ on vertices
where $e_0 \neq  \emptyset$.
In other words we will find a gsp-tree that specifies
a
permissible $vt$-path
for every vertex $v$ of $G$.
For convenience we assume the desired gsp-tree exists, as it does in our application.

To illustrate the discussion Fig.\ref{SAStrucFig}
shows a given graph (all $e_0$ values are edges). %By intent,
Vertices $i,j,k,m$ have no permissible path to $t$.
Fig.\ref{SAStrucbFig} gives a gsp-tree for the subgraph of
Fig.\ref{SAStrucFig} induced by $t$ and $a,\ldots, h$.

We begin with some definitions and facts
that lay the foundations of the algorithm. The discussion refers
to both the given function $e_0$ and the final function $e$. In situations
where
either function can be used we try to use the more informative $e_0$.
Also note that the notion of permissibility is slightly different for
the two functions. We will specify the function $e_0, e$
when it may not be clear.

A cycle $C$ is {\em permissible} (for $e$) if for some vertex $b\in C$,
every vertex $x\in C-b$ has $e(x)\in E(C)$; furthermore,
$b=t$ if $t\in C$ and
$b\ne t$ implies $e(b)\in \delta(C)$.
Observe that
for a permissible cycle $C$,
any $x\in C-b$ has a permissible path to $b$
using edges of $C$.

We will find the desired gsp-tree  \T. by growing a ``search tree'',
repeatedly contracting a permissible cycle $C$, and making
$C$ a cycle node of \T..
If $\o C$ denotes the contracted vertex,
the contracted graph has  $e(\o C)=
e(b)$. Also if $b=t$ then
the contracted
graph has $t= \o C$.

As an example in Fig.\ref{SAStrucFig} we will contract $a,b,c$;
this gives the leftmost cycle node
of Fig.\ref{SAStrucbFig}; the contracted vertex gets $e$-value $at$.

Let $\o G$ be a graph formed from $G$ by zero or more
such contractions.
A {\em search tree}\footnote{The reader will recognize
the resemblance to search trees used in cardinality matching
algorithms.}
$T$ is a tree in $\o G$ rooted at $t$,
with each of its nodes  $v$ having one of two types,
``u'' for ``up'' or ``d'' for ``down'', depending on where
we find the edge $e(v)$. In precise terms
let $f$ be the edge from the parent of $v$ to $v$.
If $f=e(v)$, or if $v=t$,  then
$v$ is a {\em u-vertex}.
 In the remaining case (i.e.,
$v\ne t$ and $f\ne e(v)$)
$v$ is  a {\em d-vertex}. We require that
a d-vertex
$v$ has at most one child, namely $c$ where $e(v)=vc$.

As an example in Fig.\ref{SAStrucFig}
a search tree might
initially contain u-vertices $t,a$ and d-vertices $b,c$.
Contracting $a,b,c$
makes
the contracted vertex a u-vertex. Edge $cf$ can be added to the search
tree after this contraction, but not before.
Adding $cf$ makes $f$ a d-vertex.

A search tree has the following property:

%\bigskip

%\noindent
%(a)

\begin{lemma}
\label{PCycleLemma}
Consider a nontree edge $x_1x_2$ where for $i=1,2$,
$x_i$
is a u-vertex or $x_i$ is  a d-vertex having $e(x_i)=x_1x_2$.
Let $b$ be the nearest common ancestor of $x_1,x_2$ and let
$C$ be the fundamental cycle of $x_1x_2$. Then $C,b$ is a permissible
cycle.
\end{lemma}

Fig.\ref{SAStrucFig} illustrates  the lemma.
In the given graph (i.e., before any cycles  have been contracted)
edges $bc$ and $et$ give permissible cycles (for an appropriate search tree).
$bc$ ($et$)  illustrates the lemma when $\{x_1,x_2\}$ contains
2 (1) d-vertices, respectively.
After these 2 permissible cycles have been contracted,
$gd$ illustrates the lemma when $\{x_1,x_2\}$ contains 2 u-vertices.
The nontree edge $ik$ ($jl$) illustrates a nonpermissible cycle,
when $\{x_1,x_2\}$ contains 2 (1) d-vertices, respectively.

\bigskip

\begin{proof}
First note $b$ is a u-vertex. (If $b$ is a d-vertex,
it has just one child $c$, with $e(b)=bc$.
$x_1x_2$ must be a back edge with some $x_i=
b$. The hypothesis  makes $e(x_i)=e(b)=x_1x_2$, contradiction.)
Now it is easy to check the properties for permissibility
(e.g., a d-vertex $v \in C$ has $e(v)\in C$
since $v$ is either an $x_i$ or the edge from $v$ to its unique child
is in $C$).
\end{proof}

The algorithm will contract such cycles $C$, and only these. In particular
this implies that  in $\o G$,
any d-vertex is actually a vertex of $G$.

Also note the algorithm will choose values $e(v)$ for vertices
with $e_0(v)=\emptyset$ by making $v$ a
u-vertex.
In particular this implies that $e(v)=e_0(v)$ for every d-vertex $v$.

\iffalse
\bigskip

\noindent
(b) If $C,b$ is a contracted cycle, any vertex $v\in C$ has
a permissible $vb$-path in $G$.
The proof is a simple induction on the number of contracted
cycles, using the definition
 of permissible cycle.

\bigskip

\noindent
(c)
Let $v$ be a vertex of $G$ whose image $\o v$ in $\o G$ is
a {\em u}-vertex.
Let $\o P$ be the $\o v t$-path in $T$. $\o P$
gives a permissible
$vt$-path $P$ in $G$. In proof let $x$ be a  vertex  of $\o P$.
If $x$ is a {\em d}-vertex, $\o P$ enters $x$ from its unique
child; the corresponding edge is $e_0(x)$. If $x$ is a u-vertex,
it may actually be a contraction of vertices of $G$.
Let $x_0$ be the vertex of $G$ with $e_0(x_0)=e(x)$. ($x_0$ is $x$
unless $x$ is a contraction, in which case $x_0$ is a $b$ vertex
of a contracted cycle.)
$\o P$ enters $x$ along an edge that, by
property (b), has a permissible path to
$x_0$. $\o P$ leaves $x$ on edge $e_0(x)=e_0(x_0)$.

\bigskip
\fi

\paragraph*{Algorithm}
We will build the desired gsp-tree
by growing a search tree $T$. Initialize $T$ to consist of just a  root $t$
(a {\em u}-vertex) and let $e$ be the given function $e_0$.
Then repeat Algorithm~\ref{algorithm:scan} until
every u-vertex $x$ in the current graph has every incident edge
either in $T$ or scanned from $x$,
and every d-vertex has been scanned.

\begin{algorithm}[H]
\caption{A procedure that scans vertices and builds the gsp-tree.}
\begin{algorithmic}[1]
\If{some u-vertex $x$ has an unscanned edge $xy\notin T$} \comment{scan $xy$ from $x$}
\If{$y\notin T$}
\State add $xy$ to $T$
\ElsIf{$y$ is a u-vertex {\bf or} $xy=e_0(y)$} \comment{$y$ is a d-vertex leaf}
\State contract the fundamental cycle of $xy$
\EndIf
\ElsIf{there is an unscanned d-vertex $x$}\comment{scan $x$}
\State let $e_0(x)=xy$
\If{$y\notin T$} 
\State add $xy$ to $T$
\ElsIf{$xy=e_0(y)$} \comment{$y$ is a d-vertex leaf}
\State contract the fundamental cycle of $xy$
\EndIf
\EndIf
\end{algorithmic}
\label{algorithm:scan}
\end{algorithm}

%
%\bigskip
%
%{\narrower
%
%{\parindent = 0pt
%
%\def\i{{\bf if\ }}
%\def\t{{\bf then\ }}
%\def\e{{\bf else\ }}
%\def\hi{\advance \parindent by 20pt}

%\i  some u-vertex $x$ has an unscanned edge $xy\notin T$ \t
%$/*$ scan $xy$ from $x$ $*/$
%
%{\hi
%
%\i  $y\notin T$ \t add $xy$ to $T$
%
%\e \i
%$y$ is a u-vertex {\bf or}
%$xy=e_0(y)$ $/*$ $y$ is a d-vertex leaf $*/$ \t
%
%{\hi
%
%contract the fundamental cycle of $xy$
%
%}
%}
%
%\e \i there is an unscanned d-vertex $x$  \t $/*$ scan $x$ $*/$
%
%{\hi
%
%let $e_0(x)=xy$
%
%\i  $y\notin T$ \t add $xy$ to $T$
%
%\e %\i $y$ is a d-vertex {\bf and } $xy=e_0(y)$ $/*$ $y$ is a leaf $*/$ \t
%\i $xy=e_0(y)$ $/*$ $y$ is a d-vertex leaf $*/$ \t
%
%{\hi
%
%contract the fundamental cycle of $xy$
%
%}
% }
%
%
%}}

%\bigskip

\paragraph*{Addendum}
Several aspects of this algorithm are stated at a high level and
deserve further elucidation:

{%\parindent =0pt

An edge $xy$ that is scanned from a u-vertex $x$
is considered scanned from any u-vertex $\o x$ that contains $x$
by contractions.

As mentioned, when the algorithm  adds an edge $xy$ to $T$,
if $e_0(y)=\emptyset$ it sets $e(y)=xy$.
Furthermore in all cases it makes $y$ type {\em u} or {\em d} as appropriate.

When the algorithm contracts a fundamental cycle $C$,
it creates a node in \T. whose children
correspond to the vertices of $C$.
If the algorithm halts with $T$ containing
nodes other than $t$, it creates a root node of \T. whose children
correspond to the vertices of $T$.

}

%\bigskip

\paragraph*{Examples}
Note that over the entire algorithm an edge may get scanned twice,
once from each end, e.g.,
in Fig.\ref{SAStrucFig} $gd$ may get scanned from $g$ when
$d$ is still a d-vertex.
In a similar vein, if we change $e_0(i)$ to a new edge $id$,
the algorithm might scan $i$ (and edge $id$) before $d$ becomes a u-vertex
and $id$ gets scanned from it.

%\bigskip

\paragraph*{Analysis of the algorithm}
The algorithm takes no action for the following
types of edges $xy$:

\i $x$ a u-vertex, $y$  a d-vertex, $xy\ne e_0(y)$,

\ii $x$ and $y$  d-vertices, $xy$ is not both $e_0(x)$ and $e_0(y)$.

\iii $x$ a d-vertex, $y\notin T$, $xy \ne e_0(x)$.

\noindent
\ii and \iii  include edges that are never even scanned (e.g.,
$ik, im$ in  Fig.\ref{SAStrucFig}).
When the algorithm halts every nontree edge of $\o G$ with an end in
$T$ is of type \xi, \ii or \xiii.

Now we show the algorithm is  correct.
First note Lemma \ref{PCycleLemma} shows every contracted cycle is
permissible.

Assume every vertex $v$ has a permissible path to $t$ in the given graph $G$
with function $e_0$.
Let $\o G$ be the final graph of the algorithm.
Let $T$ be the final tree.

\begin{lemma}
$\o G$ consists entirely of u-vertices.
\end{lemma}

\begin{proof}
Let $v$ be a non-u-vertex, with $P$ a permissible $vt$-path in $G$.
The reader should bear in mind the possibility that
the image of $P$ in $\o G$ needn't be
permissible (for $e$). For instance suppose $\o G$ contains a contracted node
like $X=\{a,\ldots,g,t\}$ in Fig.\ref{SAStrucFig}.
If the graph has two  edges $b'b,c'c\in \delta(X)$,
$P$ might contain a subpath $b',b,c,c'$, making  it  nonpermissible in $\o G$.
In fact similar edges incident to
 $fg$ might make the image of $P$ nonsimple in $\o G$.

First observe that we can assume $v$ is a d-vertex and $V(P)\con V(T)$.%
\footnote{Here we commit a slight abuse of notation: The inclusion
is meant to allow the possibility that
a vertex of $P$ is contained in a contracted vertex of $T$.}
In proof, suppose $P$ contains a nontree vertex (if not we're clearly done).
It
is eventually followed by an edge $rs$, where $r\notin V(T)$ and
$V(Q)\con V(T)$ for $Q$
the $st$-subpath of $P$.
\xi--\iii show $s$ is not a u-vertex. So $s$ is a d-vertex. \iii shows
$rs \ne e_0(s)$, so $Q$ %the  $st$-subpath
is permissible.

% contained in $V(T).  %$\ni s$.

\iffalse
Choose $v$ so $P$ has the fewest number of edges possible.
First observe $V(P)\con V(T)$. In proof, if $P$ contains a nontree vertex it
is eventually followed by an edge $rs$, $r\notin V(T)\ni s$.
\i shows $s$ is not a u-vertex. So $s$ is a d-vertex. \iii shows
$rs \ne e_0(s)$, so the  $st$-subpath is permissible.
It  contradicts the choice of $v$.

So assume $v$ is a d-vertex and $V(P)\con V(T)$.
\fi

Imagine traversing the
edges of $P$, starting from $v$.
Some edges of $P$ will be in contracted vertices of $\o G$, others will be
edges of $T$, and all
others will be of type \i or \xii, not \iii (since $V(P)\con V(T)$).
We assert that whenever we reach a vertex $r$ of $\o G$, either

\bigskip

(a) $r$ is a d-vertex, and $e_0(r)$ has not been traversed, or

(b) $r$ is a u-vertex, reached by traversing edge $e(r)$.

\bigskip

\iffalse
$x$ is a d-vertex, $y$ is the child of $x$, $xy=e_0(x)$,
 and $xy= e_0(y)$ implies $y$ a u-vertex,

$x$ and $y$ are d-vertices, $xy=e_0(x)\ne e_0(y)$,

$x$ is a u-vertex, $P$ enters $x$ from its parent in $T$,
i.e., along edge $e_0(x)$,
$y$ is a u-vertex child of $x$ or $y$ a d-vertex, $xy\ne e_0(y)$.

\bigskip
\fi

We prove the assertion by induction.
Note that the assertion completes the proof of the lemma:
(b) shows  $P$  always enters a  u-vertex from its parent, so it never reaches
the u-vertex $t$, contradiction.

For the base case of the induction, $r=v$ obviously satisfies (a).
For the inductive step assume the assertion holds for
$r$ and let $s$ be the next vertex of $\o G$ that is reached.
So $rs$ is the next edge of $\o G$ that is traversed.
 ($r$ or $s$ may be  contracted vertices of $\o G$.)

\case {$r$ satisfies (a)}
Permissibility in $G$ implies
the next edge of $P$ in $\o G$ is $rs=e_0(r)$.
% (even if $r$ is a contracted vertex).

Suppose $e_0(r)\in T$. This implies $s$ is the child of $r$.
The definition of search tree implies $s$ satisfies
(a) or (b), depending on whether or not
$rs=e_0(s)$.

\iffalse
If $rs=e_0(s)$. then $s$ is  a u-vertex, and  (b) holds for $s$.
If $rs\ne e_0(s)$ then  (a) holds for $s$.
\fi

Suppose $e_0(r)\notin T$.
$e_0(r)$ is not  type \i above (even if we take $r=y$).
So it is type \xii. This implies (a) holds for $s$.

\case {$r$ satisfies (b)}
The argument is similar.

Suppose $rs\in T$. (b) shows $r$ is reached from its parent in $T$.
So $s$ is a child of $r$.
As in the previous case,
$s$ satisfies
(a) or (b).

Suppose $rs\notin T$.
So $rs$ is type \i above and
(a) holds for $s$.
\end{proof}

The lemma implies
$T$ with its contractions gives the desired gsp-tree \T..
(E.g., if $T$ consists of just one vertex $t$, the
root of \T.
is a cycle node, specifically the last cycle to be contracted;
otherwise the root is a tree node and its tree is $T$.)

It is easy to implement the above procedure in time $O(m\log n)$ or better
using an algorithm for set merging \cite{gabow-76}.

%%%%%%%%%%%%%%%%%%%%%%%%

%% file: shtstpcomb.tex
\section{Combinatoric algorithms for shortest paths}
\label{SPcombinatoric}
Let $(G,t,w)$ denote a connected undirected graph with
a distinguished vertex $t$ and a conservative
edge-weight function $w:E\to \mathbb {R}$. Let $E^{-}$ be the set of edges with negative weights.
In this section we show how to use combinatoric algorithms
for finding maximum perfect matchings to compute the gsp-tree.
We will define a graph $\ddot{G}_t$ that models paths in $G$ by {\em almost perfect matchings}, i.e., matchings that miss exactly one vertex.
Moreover, we define {\em a $v$-matching} to be an almost perfect matching in $G$ that avoids $v$.
We believe that the construction of the split graph $\ddot{G}$ is essentially due to Edmonds~\cite{Edmonds67}.
We define the {\it split graph} $\ddot{G}=(\ddot{V},\ddot{E})$ with weight function $\ddot{w}$ in the following way
\[
\ddot{V} = \{v_1, v_2: v \in V\} \cup \{e_1,e_2: e \in E^{-}\},
\]
\begin{eqnarray}
\ddot{E} = \{v_1 v_2: v\in V\} &\cup& \{u_1v_2, u_2v_1, u_1v_1, u_2v_2 : uv\in E \setminus E^{-}\} \nonumber \\
&\cup& \{u_1 e_1, u_2e_1, e_1 e_2, v_1 e_2, v_2 e_2 : e=uv \in E^{-}, u<v\},\nonumber
\end{eqnarray}
\[
\ddot{w}(u_i v_j) =
 \left\{
\begin{array}{rl}
 -w(uv) & \textrm{if  } uv \in E\setminus E^{-}, \\
 -w(e) & \textrm{if  } u_i=e_1 \textrm{ and } v_j\neq e_2 \textrm{ and } e \in E^{-},\\
 0 & \textrm{otherwise.}
\end{array}
\right.
\]

%\begin{figure}[h]
%\begin{center}
%  \includegraphics{figures/fig5.pdf}
%%\subfloat{
%%  \includegraphics[width=0.33\textwidth]{figures/fig3.pdf}
%%}
%%\subfloat{
%%\hspace{-0.7cm}
%%  \includegraphics[width=0.33\textwidth]{figures/fig4.pdf}
%%}
%\caption{An undirected graph $G$ and its graph $\ddot{G}$.
%  In $\ddot{G}$ zigzag edges weigh $-1$, dashed edges weigh $1$
%    and the remaining edges weigh $0$.
%  Vertices corresponding to negative edges of $G$ are white squares.
%  The far right  shows a matching $M(a_2c_1)$ of weight $-2$, which corresponds
%  to a shortest path between $a$ and $c$.
%}
%\label{fig:ddot-example}
%\end{center}
%\end{figure}
%Fig.\ref{fig:ddot-example} gives a complete example.

%Note how a length-two path in $G$, say $a,b,c$ with $w(ab)\ge 0>w(bc)$ and $e=bc$, corresponds to
%a matching  in $\ddot{G}$ such as $a_1b_1, b_2e_1,e_2c_1$, having the same total weight.
An important property is that we can assume
$\ddot{n}=|\ddot{V}| \le 4n$.  This follows since we can assume
 $|E^{-}| < n$, as otherwise the set of negative edges contains a cycle.
The following observation is essentially given in~\cite{Ahuja93} in Chapter~12.7, where
the reduction is explained on a clear example.

\begin{lemma}
\label{lemma-lawler}
Let $u,v\in V$, let $M$ be the maximum perfect matching, and let $M(u_2v_1)$ be the maximum weight perfect matching in $\ddot{G}- u_2-v_1$.
If $G$ does not contain negative weight cycles then $\ddot{w}(M)=0$ and
the shortest path weight from $u$ to $v$ in $G$ is equal to $-\ddot{w}(M(u_2v_1))$.
\end{lemma}

Note also that it is easy to detect a negative cycle in $G$ -- it corresponds to a perfect matching in $\ddot{G}$ with positive weight.
On the other hand, as described in~\cite{Ahuja93}, in order to find a shortest path from $u$ to $v$ we need to find the maximum weight
perfect matching $M(u_2v_1)$. However, here we want to compute the whole gsp-tree and hence require the distances from all vertices
in $G$ to $t$. We will show that in order to find all these distances it is essentially enough to find one maximum perfect matching.
Let us define $\ddot{G}_t$ to be graph $\ddot{G}$ with both vertices $t_1$ and $t_2$ unified to one vertex $t$. We observe
that the resulting graph is {\em critical}, i.e., for each vertex $v$ there exists a $v$-matching.

\begin{lemma}
\label{lemma-critical}
Graph $\ddot{G}_t$ is critical. Moreover, let $M(v_1)$ be the maximum $v_1$-matching in $\ddot{G}_t$ then
$\ddot{w}(M(t_2v_1)) = \ddot{w}(M(v_1))$.
\end{lemma}
\begin{proof}
We need to show the existence of  $v$-matchings for all vertices $v$ in $\ddot{G}_t$. Consider the following cases

$\bullet$
$v=t$  then $M(t_2t_1)$ in $\ddot{G}$ corresponds matching that avoids $t$,

$\bullet$
$v=v_2$  then $M(t_1v_2)$ in $\ddot{G}$ corresponds to matching that avoids $v_2$,

$\bullet$
$v=v_1$ then in this case we take $M(t_2v_1)$ in $\ddot{G}$.

\noindent Hence, $\ddot{G}_t$ is critical. The second part of the lemma follows by the
above correspondence of matchings.
\end{proof}

We need to relate the above observation to the definitions from previous sections. You might observe
that $M(v_1)$ corresponds to $P_v$ from Section~\ref{SPSec} and in language of $f$-factors
to $F_v$ as discussed in Section~\ref{fFactorSec}. In order to construct the gsp-tree we use the shrinking
procedure for general $f$-factors from Section~\ref{SPAlgSec}. This requires us to know $w(F^v)$ as well.
However, we can observe that $F^v$ can be obtained from $F_v$ and vice versa by adding or removing the zero weight loop $vv$.

\subsection{Matching duals for factor critical graphs}
Consider an arbitrary graph $G = (V, E)$ that contains a perfect matching. Let $w : E \to \mathbb {R}$ be the
edge weight function. The dual variables in Edmonds' formulation~\cite{e} are assigned
to vertices $y:V \to \mathbb {R}$ and to odd-size subsets of vertices $z:2^V\to \mathbb {R}$.
The function $z$ can be negative possibly only on $V$. We define the value of $uw$ (with respect
to the dual $y,z$) as

\[
\H{yz}(uv) = y(u) +y(v)  + z\set {B}{e\subseteq B}.
\]
We require the duals to {\em dominate} all edges $uv\in E$, i.e., we require
\[
\H{yz}(uv) \ge w(uv).
\]
We say that an edge is {\em tight} when the above inequality is satisfied with equality.
The {\em dual objective} is defined as
\[
(y,z)V = y(V) +\sum\left\{\lfloor|B|/2\rfloor z(B) : B\subseteq V\right\}.
\]
By the duality we know that for any perfect matching $M$ and duals $y,z$ we have $w(M)\le (y,z)V$.
We say that a matching {\em respects} a set $B$ if it contains $\lfloor |B|/2\rfloor$ edges in $\gamma(B)$.
As shown by Edmonds~\cite{e} a perfect matching is maximum if and only if all its edges are tight and
it respects all sets with positive $z$ for a pair of dominating duals.

A {\em blossom} is a subgraph $B$ of $G$ defined as follows. Every vertex
is a blossom and has no edges. Otherwise, the vertices $V(B)$ are partitioned
into an odd number $k$ of sets $V(B_i)$, for $1 \le i \le k$, where each $B_i$
is a blossom. Each blossom $B$ contains edges in $E(B_i)$ and $k$ edges that form a cycle on $B_i$, i.e.,
consecutive edges end in $V (B_i)$ and $V(B_{i+1})$, where $B_1=B_{i+1}$.
Blossoms $B_i$ are called subblossoms of $B$. The set of blossoms can be represented as forrest called {\em blossom forrest}.
The parent-child relation in this forrest is defined by the blossom-subblossom
relation.

Edmonds' algorithm for finding maximum perfect matchings constructs a {\em structured matching}, i.e.,
a matching $M$ and a dual solution. The dual solution is composed out of a blossom forrest $F$ and functions $y$ and $z$.
The structured matching satisfies the following conditions

\i $M$ respects blossoms in $F$,

\ii $z$ is nonzero on blossoms in $F$,

\iii all edges in $M$ are tight,

\iv all edges in blossoms in $F$ are tight.

If one finds a structured matching that is perfect then it is a maximum perfect matchings.
In the classical view Edmonds' algorithms algorithm operates on a graph with even number of
vertices. However, as shown in~\cite{gabow-tarjan-89} it can be seen to work on a critical
graph. The {\em optimal matching structure} of a critical graph $G$ consists a blossom tree $B$
and dual functions $y,z$ such that every vertex is a leaf in $B$ and properties \ii and \iv
above are satisfied. Observe that there is no matching in this definition. As argued
in~\cite{gabow-tarjan-89} Edmonds' algorithm computes an optimum matching structure when it is
executed on a critical graph.  Moreover, the optimal matching structure allow us to relate weights of maximum $v$-matchings to the dual $y,z$
in the following way.

\begin{lemma}[\cite{gabow-tarjan-89}]
\label{lemma-v-matching}
Let $y,z$ and $B$ the optimal matching structure for the critical graph $G$. Then
the weight of maximum $v$-matching is to equal to $(y, z)V - y(v)$.
\end{lemma}

\subsection{The algorithm}
Let us now join all the ingredients to compute the gsp-tree in the undirected graph $(G,w,t)$ with
the conservative weight function $w$.

\begin{algorithm}[H]
\caption{A combinatoric algorithm for computing gsp-tree for the undirected graph $(G,w,t)$.}
\begin{algorithmic}[1]
\State{Construct $\ddot{G}$ from $G$}
\State{Construct $\ddot{G}_t$ from $\ddot{G}$ by identifying $t_1$ and $t_2$}
\State{Compute optimal matching structure $y,z$ and $B$ for $\ddot{G}_t$} \comment{$\ddot{G}_t$ is critical by Lemma~\ref{lemma-critical}}
\For{$v\in V$}
\State{Let $w(M_v) = (y, z)V - y(v)$} \comment{by Lemma~\ref{lemma-v-matching}}
\State{Let $w(F_v) = -w(M_v)$} \comment{by Lemma~\ref{lemma-critical} and Lemma~\ref{lemma-lawler}}
\State{Let $w(F^v) = w(F_v)$} \comment{by adding zero length loop $vv$}
\EndFor
\State{Using $w(F_v)$ and $w(F^v)$ find the blossom forrest}\comment{using shrinking procedure from Section~\ref{fFactorSec}}
\State{Construct gsp-tree from the blossom forrest}\comment{using procedure from Section~\ref{SPAlgSec}}
%\EndProcedure
\end{algorithmic}
\label{algorithm:combinatoric-shortest-paths}
\end{algorithm}

We note that optimal matching structure can be found using fast implementations of Edmonds' algorithm in $O(n(m + n \log n))$ time~\cite{gabow-90} or in $O(\sqrt{n \alpha(m,n)\log n}\ m \log (nW))$
time~\cite{gabow-tarjan-89}. The other steps of the above algorithm take only less time, e.g., the shrinking procedure can be implemented to work in $O(m \log n)$ time. This way we obtain
the $O(n(m + n \log n))$ time and the $O(\sqrt{n \alpha(m,n)\log n}\ m \log (nW))$ time combinatoric algorithms for computing gsp-tree.

%% file: determinant.tex
\section{Determinant formulations for general graphs}
\label{sec:determinant}
Let $G$ be a simple graph with vertices numbered from 1 to $n$.
Let $\phi=\sum_i f(i)$. We
define a skew-symmetric $\phi \times \phi$ matrix $B(G)$ representing $G$ in the following
way. A vertex $i\in V$ is associated with
$f(i)$ consecutive rows and columns of $B$, both indexed by the pairs $i,r$ for $0\le r<f(i)$.
Call a tuple $(i,r,j,c)$ representing the entry $B(G)_{i,r,j,c}$
{\it permissible} if either
\[i<j, %
\text{ or } i=j \text{ and } r<\f{f(i)/2}\le \c{f(i)/2}\le c<f(i).
\]
The set of permissible entries is denoted by $\mathbb P$.
Note that the permissible entries of $B(G)$ are all above the diagonal.
The permissible entries corresponding to a fixed edge $e$
form  a
rectangular submatrix for a nonloop $e$ and 
an $\f{f(i)/2} \times \f{f(i)/2}$ submatrix for a loop $ii$.

Using indeterminates $x^{ij}_{r}$, $y^{ij}_{c}$ we define an entry of $B(G)$ as
\begin{equation}
\label{SimpleGenMatrixEqn}
B(G)_{i,r,j,c} = % \!=\!
\begin{cases}
x^{ij}_{r} y^{ij}_{c}& %
ij\!\in\! E \text{ and } (i,r,j,c)\in \mathbb P,\\ % \text{ is permissible,}\cr
-x^{ji}_{c} y^{ji}_{r}&  %
ij\! \in\! E \text{ and } (j,c, i,r)\in \mathbb P,\\ % \text{ is permissible,}\cr
0&\text{otherwise.}
\vspace{-0.1cm}
\end{cases}
\end{equation}
Clearly $B(G)$ is skew-symmetric. Note that a loop $ii\in E$ is represented by
an $\f{f(i)/2}\times \f{f(i)/2}$ submatrix that is empty
if $f(i)=1$. This is fine since $ii$ is not in any $f$-factor.

The analysis for multigraphs is almost identical to simple graphs.
So we will concentrate on simple graphs, but also point out how it extends to  multigraphs.
Towards this end we extend the definition of $B(G)$ as in Section \ref{section:biparite:formulations}:
Let $\mu(e)$ denote the multiplicity of any edge $e$. The copies of $e=ij$ get
indeterminates $x^{ij,k}_{r}$, $y^{ij,k}_{c}$ ($1\le k\le \mu(e)$) and we define $B(G)$ by
\begin{equation}
\label{MultiGenMatrixEqn}
B(G)_{i,r,j,c} =
\begin{cases}
{\sum _{k=1}^{\mu(e)} x^{ij,k}_{r} y^{ij,k}_{c}}& %
\minibox{$ij\in E$  and $(i,r,j,c) \in \mathbb P$,}\vspace{3pt}
 \\\vspace{3pt} %\text{ is permissible,}\cr
-\sum _{k=1}^{\mu(e)} x^{ji,k}_{c} y^{ji,k}_{r}&  %
\minibox{$ij \in E$  and  $(j,c, i,r) \in \mathbb P$,}\\% \text{ is permissible,}\cr
0&\text{otherwise.}
\end{cases}
\end{equation}

\subsection{Review of the Pfaffian}

\def\pf{{\rm pf}}
\def\sgn{{\rm sgn}}

Our goal is to prove an analog of Theorem \ref{BipSimpleDetThm}. 
It is easiest to accomplish this using
the Pfaffian, so we begin by reviewing this concept.
Let $A$ be a skew-symmetric matrix 
of order $2h\times 2h$. Its Pfaffian is defined by
$$\pf(A)=\sum \sgn
\left(\begin{array}{ccccc}
1&2&\ldots&2h-1&2h\\
i_1& j_1&\ldots& i_h& j_h
\end{array}
 \right)
a_{i_1 j_1}\ldots a_{i_h j_h}.
$$
The above sum is over all partitions
of the integers $ [1..2h]$ into pairs
denoted as $\{i_1, j_1\},\ldots. \{i_h, j_h\}$.
Also {\em sgn}  denotes the sign of the permutation.
(It is easy to check that each term $\sigma$ in the summation is well-defined.
For example interchanging $i_1$ and $j_1$ does not change the 
partition. It flips both the sign of the permutation and 
the sign of $a_{i_1 j_1}$, so $\sigma$ is unchanged.)
The central property of the Pfaffian is \cite{lp} 
\begin{equation}
\label{PfEqn}
\det(A)=(\pf(A))^2.
\end{equation} 

Recall that the Tutte matrix $T$ for a graph $G$ is the skew-symmetric matrix
obtained from the adjacency matrix of $G$ by replacing the entry
1 for edge $ij$ by an indeterminate $t_{ij}$ if $i<j$ and by $-t_{ji}$ if
$i>j$. Clearly there is a 1-1 correspondence between perfect matchings
of $G$ and terms of $\pf(T)$. 

Recalling \eqref{PfEqn}, 
note that every pair of perfect matchings
 of $G$ gets combined in the product $(\pf(T))^2$.
Let us review a proof that
no cancellations occur when all these terms are added together.

Consider a term %$\sigma$ 
of
$(\pf(T))^2$ corresponding to perfect matchings $M_1,M_2$.
(These matchings may be distinct or identical.)
The terms of 
$(\pf(T))^2$ that might cancel %$\sigma$
this term
must use exactly the same variables $t_{ij}$. In other words
they correspond to
perfect matchings $N_1,N_2 $ 
where 
\begin{equation}
\label{MNEqn}
M_1\uplus M_2=N_1\uplus N_2.
\end{equation}
Here $\uplus$ denotes multiset sum.
$M_1\uplus M_2$ consists of the edges of 
$M_1\cap M_2$ taken twice, plus the set
$M_1\oplus M_2$. The latter  is a collection \C. of even alternating cycles.
The definition of matching shows every cycle of \C. is alternating wrt
$N_1,N_2 $. Assume that in
$(\pf(T))^2$, the first multiplicand $\pf(T)$ gives the matching $N_1$
and the second multiplicand gives
$N_2$. The edges of
each cycle of $\C.$ can be partitioned
in two ways between $N_1$ and $N_2$.
So 
there are precisely $2^{|\C.|}$ matching pairs $N_1,N_2$
of $(\pf(T))^2$ that  satisfy \eqref{MNEqn}.
 
We will show there are no cancellations 
because each matching pair $N_1,N_2$ gives the same term of $(\pf(T))^2$.
%In other words  
More precisely we show the following:

\claim 1 {%
The pairs $N_1,N_2$ satisfying \eqref{MNEqn} collectively contribute
the quantity $(-2)^{|\C.|} \Pi t_{ij}$ to $(\pf(T))^2$, where 
%$s 2^{|\C.|} \Pi t_{ij}$ to $(\pf(T))^2$, where $s\in \{1,-1\}$ and 
the product is over a fixed set of
%the above-diagonal 
entries in $T$ corresponding to
$M_1 \uplus M_2$.}

\noindent
 Note the claim also implies that if we do arithmetic
over a finite field of characteristic $>2$, again there are no cancellations.

\def\bcproof #1 {\noindent{\em Proof of Claim #1.} }

\bigskip

\bcproof 1
Take a pair of matchings $N_1,N_2$ satisfying
\eqref{MNEqn}. Let $\sigma_i$ be the term in
$\pf(T)$ for $N_i$, so $(\pf(T))^2$ contains $\sigma_1\sigma_2$.
We will show the portion of $\sigma_1\sigma_2$
corresponding to $C$ is the same for every $N_1,N_2$.
Then we will conclude this property makes the entire term
$\sigma_1\sigma_2$ independent of choice of $N_1,N_2$.

First suppose $C$ is a single edge $ij$ (belonging to
$M_1\cap M_2$). Wlog both $\sigma_1$ and $\sigma_2$ have the $T$ entry $t_{ij}$
and their permutations both map
some ordered pair of two consecutive integers $(2a-1,2a)$
to the ordered pair $(i,j)$.
Clearly this does not depend on choice of $N_1,N_2$.

Now suppose $C$ contains
vertices $i_1,\ldots, i_\ell$ for some even $\ell\ge 4$, and
$N_1$ contains edges $i_1i_2,\,\ldots,\, i_{\ell-1}i_\ell$
while $N_2$ contains  $i_2 i_3,\,\ldots,\, i_\ell i_1$. 
$\sigma_1$ contains the product
$t_{i_1i_2}\ldots t_{i_{\ell-1},i_\ell}$
and $\sigma_2$ contains $t_{i_2 i_3}\ldots t_{i_{\ell},i_1}$.
So $\sigma_1\sigma_2$ 
contains the product
$(t_{i_1i_2}\ldots t_{i_{\ell-1},i_\ell})(t_{i_2 i_3}\ldots t_{i_{\ell},i_1})$.
Certainly this is independent of choice of $N_1,N_2$.

The permutation for $\sigma_1$ maps some pair of two consecutive integers
$(2a-1,2a)$ to $(i_1,i_2)$, and similarly for
the rest, e.g., $(2b-1,2b)$ goes to $(i_{\ell-1},i_\ell)$.
Wlog $N_2$ maps $(2a-1,2a)$ to $(i_2,i_3)$, and similarly
for the rest, e.g., $(2b-1,2b)$ goes to $(i_{\ell},i_1)$.
The pairs of $\sigma_1$ are
 transformed to the pairs of $\sigma_2$ by applying the product of
transpositions
$(i_1 i_2)(i_1 i_3)\ldots (i_1 i_\ell)$.
This is an odd number of transpositions, i.e., its sign is $-1$. 
So the part of the permutations
for $\sigma_1$ and $\sigma_2$ in $C$ combine to give
the sign $-1$ in $\sigma_1\sigma_2$. 

Applying this analysis of sign to each of the $|\C.|$
components of $\ge 4$ vertices shows the entire permutation of
$\sigma_1\sigma_2$ contributes sign $(-1)^{|\C.|}$.
Since there are $2^{|\C.|}$ pairs
$N_1,N_2$, we get the contribution of the claim.
\ecproof

We shall also use a special case of the above analysis, specifically
when $N_1$ and $N_2$  are identical except for
choosing alternate edges of one cycle of length 4.
$\sigma_1$ and $\sigma_2$
have opposite sign, and so they differ only in the subexpressions
\begin{equation}
\label{SimilarPTermsEqn}
\pm  t_{i_1 i_2} t_{i_3 i_4} \text{ and } \mp t_{i_2 i_3} t_{i_4 i_1}
\end{equation}
for some consistent choice of sign.

\subsection{Analysis of {\boldmath $B(G)$}}
We return to the matrix $B(G)$ for $f$-factors.
An entry of $B(G)$ corresponds to an edge of $G$, so a term of
$\pf(B(G))$ corresponds to a multiset of edges of $G$.

\claim 2 {Each uncancelled term $\sigma$ in $\pf(B(G))$ corresponds to an $f$-factor of $G$.}

\bcproof 2
The above discussion shows
that ignoring sign, 
$\sigma$ is a product of $\phi/2$ quantities
$x^{ij}_ry^{ij}_c$
corresponding to edges that form a matching on 
vertices designated by two indices
$i,r$ ($1\le i\le n$, $0\le r<f(i)$).
So
each vertex $i$ of $G$ is on exactly $f(i)$ of 
the $\phi/2$ edges selected by $\sigma$ (counting loops $ii$ twice).
To show these edges form an $f$-factor of $G$ we must show
that
$\sigma$ uses each edge of $G$ at most
once. We accomplish this 
by showing that in the summation of the Pfaffian,
terms using an edge more than once
cancel in pairs.

Suppose $\sigma$ uses edges
$x^{ij}_ry^{ij}_c$ and $x^{ij}_{r'}y^{ij}_{c'}$ ($r\ne r'$ and $c\ne c'$).
Wlog assume $(i,r;j,c),(i,r';j,c')
\in \mathbb P$. $\sigma$ may have many such duplicated pairs.
Choose the duplicated pair that lexically minimizes
$(i,r,r')$. 
Pair $\sigma$ with the term $\sigma'$
having the same partition except that it replaces
$\{i,r;j,c\}$ and $\{i,r';j,c'\}$
by
$\{i,r;j,c'\}$ and $\{i,r';j,c\}$.
Note that 
permissibility of
$(i,r;j,c)$ and $(i,r';j,c')$
implies permissibility of
$(i,r;j,c')$ and $(i,r';j,c)$, even if $i=j$.
Thus $\sigma'$ is also associated with $(i,r,r')$, so the pairing is well-defined.

Viewed as matchings, $\sigma$ and $\sigma'$ differ only by
choosing alternate edges of the length 4 cycle
$(i,r;j,c;i,r';j,c')$.
So $\sigma$ and $\sigma'$ are identical except for the expressions of
\eqref{SimilarPTermsEqn}, which in the new setting become
\[
\pm  b_{i,r;j,c}\, b_{i,r'; j,c'} \text{ and } \mp b_{j,c;i,r'} \,b_{j,c'; i,r},
\] where $b$ designates matrix $B(G)$. Substituting the
definition of $B(G)$ shows these expressions are
\[
\pm (+x^{ij}_ry^{ij}_{c}) (+x^{ij}_{r'}y^{ij}_{c'})
\text{ and }
\mp (-x^{ij}_{r'}y^{ij}_{c}) (-x^{ij}_{r}y^{ij}_{c'}).
\]
The two minus signs in the second expression follow from
permissibility of
$(i,r;j,c')$ and $(i,r';j,c)$. 
The 4 indeterminates in the 2 above expressions are collectively identical.
So we get $\sigma=-\sigma'$, i.e.,
these two terms cancel each other as desired.

This argument also applies to multigraphs $G$. Here the issue is
that only one copy of each distinct edge $ij,k$ can  used.
A term $\sigma$ with a duplicated edge, like
$x^{ij,k}_ry^{ij,k}_c$ and $x^{ij,k}_{r'}y^{ij,k}_{c'}$,
chooses the duplicated pair to lexically minimize $(i,k,r,r')$.
The rest of the argument is unchanged.
\ecproof

Every $f$-factor $F$ of $G$
has an uncancelled term $\sigma$ in $\pf(B(G))$: 
For each $i\in V$, order the set $\delta(i,F)$ arbitrarily.
This makes $F$ correspond to a term $\sigma$ in $\pf(B(G))$. 
Each variable $x^{ij}_r$, $y^{ij}_c$ in $\sigma$ specifies
its edge, and so 
%$\sigma$ corresponds to a unique $f$-factor $F$.
determines a unique partition pair $\{i,r;j,c\}$.
So no other term of $\pf(B(G))$ has the same variables of $\sigma$,
and $\sigma$ is uncancelled,
(Of course a given $F$ gives rise to many different uncancelled terms.)

Now it is easy to see that \eqref{PfEqn}
gives an analog of the second assertion of Theorem \ref{BipSimpleDetThm}:
$G$ has an $f$-factor if and only if $\det(B(G))\ne 0$.
But we need to prove the stronger first assertion.

We extend the above construction for $F$
to an arbitrary pair of $f$-factors 
$F_1,F_2$. 
As before for each $i\in V$
and for each $F_j$ ($j=1,2$), 
number the edges of $\delta(i,F_j)$ (consecutively, starting at 1).
But now for both sets start with the edges of
$\delta(i,F_1\cap F_2)$, using the same numbering for both.
Using both these numberings gives two terms in $\pf(B(G))$, 
which combine to give a term $\sigma$ of $(\pf(B(G)))^2$. 
Each variable $x^{ij}_r$ %($y^{ij}_c$) 
in $\sigma$ specifies
its edge and corresponds to a unique $y^{ij}_c$. % ($x^{ij}_r$).
Here we are using the fact that an edge in both $f$-factors
gives rise to a subexpression of the form $(x^{ij}_r y^{ij}_c)^2$ in $\sigma$.

Call  a term of   $(\pf(B(G)))^2$ 
{\em consistent} if whenever an  edge $ij$
occurs twice, it uses the same subscripts, thus giving a product
of the form $(x^{ij}_r y^{ij}_c)^2$.
(We are disqualifying terms with an edge appearing as
both $x^{ij}_r y^{ij}_c$ and also $x^{ij}_{r'} y^{ij}_{c'}$
where $r\ne r'$ or $c\ne c'$.) 
Clearly the above construction gives consistent terms.
We will show there is no cancellation 
in $(\pf(B(G)))^2$ 
involving consistent terms.
(After that we comment on inconsistent terms.)

The proof follows the Tutte matrix argument:
Consider a pair of $f$-factors $F_1,F_2$, and some
consistent term $\sigma$ in $(\pf(B(G)))^2$ that corresponds to them. 
View these $f$-factors as matchings $FM_1, FM_2$ on vertices $i,r$. 
Another pair of $f$-factors $H_1,H_2$ involves the same
set of indeterminates exactly when $FM_1\uplus FM_2=HM_1\uplus HM_2$.
(Here we use the fact that by consistency,
each indeterminate $x^{ij}_r$ specifies its edge $ij$
as well as its corresponding indeterminate 
$y^{ij}_c$. Hence the indeterminates in a term determine the matched edges
$(i,r;j,c)$.)
We have now established the analog of \eqref{MNEqn}.
The rest of the argument for Tutte matrices applies unchanged.
Furthermore exactly the same analysis applies when $G$ is a multigraph.

We comment that the 
situation is more involved for an inconsistent term $\sigma$.
Suppose $F_1$ has Pfaffian term $\sigma_1$ containing 
$x^{ij}_r y^{ij}_c$ and $F_2$ has $\sigma_2$ containing
$x^{ij}_{r'}y^{ij}_{c'}$, where $r\ne r'$ and $c\ne c'$.
Furthermore suppose 
the two copies of $ij$ belong to different
cycles $C_1,C_2$ of $FM_1 \uplus FM_2$.
There can be many different possibilities for  two other $f$-factors $H_1,H_2$
 having $HM_1 \uplus HM_2$. For instance suppose edge $ij$
gives the only inconsistency in $\sigma_1\sigma_2$.
Then $C_1$ and $C_2$ only give rise to 2 possible matchings $HM_1,HM_2$
instead of 4,
since an $f$-factor cannot contain both copies of $ij$.
But if there are  inconsistent edges besides $ij$,
other partitions for the edges of $C_1$ and $C_2$ may be possible
(e.g., consider the case of 1 other inconsistent edge, with
copies in  $C_1$ and $C_2$).

\def\FF.{\Phi_2(G)}%formerly \F._2(G) , but this is too much like \F.(G)

As in Section \ref{section:bipartite}
define $\F.$ to be the function that maps each term $\sigma$
of $\det(B(G))$
to its corresponding subgraph denoted
$F_{\sigma}$.
$F_{\sigma}$ is a $2f$-factor (unlike Section \ref{section:bipartite}).
Let $\FF.$ be the set
formed by taking sums of two 
(possibly equal) $f$-factors of $G$, i.e.,
\[
%\F._2(G) = \sum_{F_1, F_2: \textrm{$f$-factor of } G} F_1 \uplus F_2,
\FF.=
\{F_1 \uplus F_2: F_1, F_2\textrm{ $f$-factors of } G\} .
\]
We have proved the following.

\begin{theorem}
\label{GenSimpleDetThm}
\label{GenMultiDetThm}
Let $G$ be a simple graph or a multigraph. The function $\F.$ from terms in 
$\det(B(G))$ is
a surjection onto $\FF.$. 
Consequently, $G$ has an $f$-factor if and only if $\det(B(G))\ne 0$.
\end{theorem}

\noindent
The theorem continues to hold when we do arithmetic
in any finite field of characteristic $>2$.

%% file: algorithms.tex
\section{Finding \boldmath{$f$}-factors in general graphs}
\label{sec:algorithms}
This section gives algorithms to find $f$-factors in general multigraphs. It starts with
simple graphs and then moves to multigraphs. One should keep in mind that it is unknown how to use Gaussian elimination in the non-bipartite case, e.g.,
\cite{ms04} uses a different non-algebraic algorithm for this case. However Harvey~\cite{Harvey06}
later developed a fully algebraic scheme, that we adopt here. A good explanation of this approach
is given in~\cite{Harvey06}. We first use the Sherman-Morrison formula to get an $O(m \phi^2)$ time algorithm.
Then we show it can be sped up to $O(\phi^{\omega})$ time using ideas from Harvey's recursive elimination scheme.

\subsection{Simple Graphs}
We define {\em removable} edge $ij \in E(G)$ to be an edge such that $G-ij$ has an $f$-factor. We can observe that following property.

\begin{corollary}
\label{corollary:removable-edge}
Let $G$ be a simple graph having an $f$-factor. The edge $ij \in E(G)$ is removable if and only if
$\det(B(G-ij))\neq 0$.
\end{corollary}

Let $ij \in E(G)$ then we can observe that $B(G)$ and $B(G-ij)$ differ from one another by two rank-one updates.
\begin{equation}
\label{equation:rank-2-update}
\begin{split}
B(G-ij) = B(G) - x^{ij} (y^{ij})^T + y^{ij} (x^{ij})^T = B(G)+ [-x^{ij},y^{ij}][y^{ij},x^{ij}]^T.
\end{split}
\end{equation}
where $x^{ij}$ and $y^{ij}$ are length $\phi$ vectors. When $i\neq j$ these vectors are defined as
\begin{eqnarray*}
x^{ij}_{k,r} = \begin{cases}
x^{ij}_{r} & \textrm{if } k=j,\\
0 & \textrm{otherwise,}
\end{cases}
&\ \ \ &y^{ij}_{k,r} = \begin{cases}
y^{ij}_{r} & \textrm{if } k=j,\\
0 & \textrm{otherwise.}
\end{cases}
\end{eqnarray*}
On the other hand for $ii \in E(G)$ we have
\begin{eqnarray*}
x^{ii}_{k,r} = \begin{cases}
x^{ii}_r & \textrm{if } k=i \textrm{ and }  r <\f{f(i)/2} ,\\
0 & \textrm{otherwise,}
\end{cases}
&\ \ \ &
y^{ii}_{k,r} = \begin{cases}
y^{ii}_r & \textrm{if } k=i \textrm{ and }  f(j)/2 \le r,\\
0 & \textrm{otherwise.}
\end{cases}
\end{eqnarray*}

Hence, one can use Sherman-Morrison-Woodbury formula to compute $B(G-ij)^{-1}$ from $B(G)^{-1}$ in $O(\phi^2)$ time.
Similarly, we can use this formula to test whether $B(G-ij)$ is nonsingular by checking
\begin{equation}
\label{equation:test-removable}
\det(I_2+\left[y^{ij},x^{ij} \right]^T B(G)^{-1}\left[-x^{ij},y^{ij} \right]) \neq 0.
\end{equation}
Here, the affected matrix size is bounded by $2f(i)+2f(j)$ so we need $O((f(i)+ f(j))^2)$ time for this test.
Using these observations we get Algorithm~\ref{algorithm:f-factor-general-simple} for finding $f$-factors in simple graphs.

\begin{algorithm}[H]
\caption{An $O(m \phi^2)$ time algorithm for finding $f$-factor in the simple graph $G$.}
\begin{algorithmic}[1]
\State{Let $B(G)$ be $\phi \times \phi$ skew-symmetric adjacency matrix of $G$}
\State{Replace the variables in $B(G)$ for random elements from $\mathcal{Z}_p$ for prime $p=\Theta(\phi^2)$ to obtain $B$}
\State{If $B$ is singular return "no $f$-factor".}
\State{(with probability $\ge 1-\frac{1}{\phi}$ matrix $B$ is non-singular when $B(G)$ is non-singular)}\comment{by
Lemma~\ref{corollary-zippel-schwartz}}
\State{Compute $B^{-1}$}
\State{(we remove only removable edges so $B$ remains non-singular during execution of the algorithm)}
\ForAll{$ij \in E$}
\label{algorithm:f-factor-general-simple:line4}
\If{$\det(I_2+\left[y^{ij},x^{ij} \right]^T B^{-1}\left[-x^{ij},y^{ij} \right]) \neq 0$} \comment{Edge $ij$ is removable by~\eqref{equation:test-removable}}
\State{Set $E:=E-e$}
\State{Set $B:=B+ [-x^{ij},y^{ij}]\cdot[y^{ij},x^{ij}]^T$} \comment{This corresponds to $B(G)-ij$ by~\eqref{equation:rank-2-update}}
\State{Recompute $B^{-1}$}\comment{Using Sherman-Morrison-Woodbury formula}
\EndIf
\EndFor
\State{Return $E$}\comment{All removable edges have been removed, so what remains is an $f$-factor}
%\EndProcedure
\end{algorithmic}
\label{algorithm:f-factor-general-simple}
\end{algorithm}

Observe that the above algorithm fits into the framework introduced by Harvey~\cite{Harvey06} for finding $1$-factors in general graphs. The above algorithm corresponds to the algorithm described in Section~3.3 from~\cite{Harvey06}, with the difference that the rank-two updates we use affect submatrices and not single elements. This means that we can use the recursive algorithm that was introduced in Section~3.4 of his paper.
The algorithm uses following elimination procedures with the starting call to {\sf DeleteEdgesWithin($V$)}.
\setdescription{leftmargin=0.2cm,labelindent=0cm}
\begin{description}
\item[{\sf DeleteEdgesWithin($S$)}] -- if $|S|\ge 1$ split $S$ into $S_1$ and $S_2$; call {\sf DeleteEdgesWithin($S_i$)}, for $i=1,2$;
call {\sf DeleteEdgesCrossing($S_1,S_2$)}; update submatrix $B^{-1}[S,S]$;
\item[{\sf DeleteEdgesCrossing($R,S$)}] -- if $|R|=\{r\}$, $|S|=\{s\}$ and $rs$ is removable eliminate edge $rs$;
 otherwise split $R$ into $R_1$,$R_2$ and $S$ into $S_1$, $S_2$; call {\sf DeleteEdgesCrossing($R_i,S_j$)}, for $i,j=1,2$;
 update submatrix $B^{-1}[R\cup S, R\cup S]$.
\end{description}
Observe that we need only to replace tests for removable edges with~\eqref{equation:test-removable}, whereas the updates
to submatrices remain essentially the same and use~\cite[Corollary 2.1]{Harvey06}. Hence, the submatrix of size $|S|\times |S|$ is updated in $O(|f(S)|^{\omega})$ time. On the other hand, the cost we pay to test whether the edge $ij$ is removable is $O(|f(i)+f(j)|^2)$.

Harvey splits the set $S$ (similarly $R$) always in equal halves. We, however, split $S$ in such a way that $f(S_1)$ and $f(S_2)$ are as close as possible. Let us assume $f(S_1)\ge f(S_2)$ then
\begin{itemize2}
\item either $f(S_1)\ge \frac{2}{3} f(S)$ and $|S_1|=1$,
\item or $f(S_1)\le \frac{2}{3} f(S)$ and $f(S_2)\ge \frac{1}{3}f(S)$.
\end{itemize2}
Let $h(S)$ denote the running time of the procedure 
{\sf DeleteEdgesWithin($S$)}. Similarly define $g(R,S)$ for {\sf DeleteEdgesCrossing($R,S$)}. We have
\begin{eqnarray*}
h(S) &\!\!\!\!=\!\!\!\!&\sum_i h(S_i)+ g(S_1,S_2) +  O(|f(S)|^{\omega}) \\
g(R,S) &\!\!\!\!=\!\!\!\!&
\begin{cases}
O(|f(R)\!+\!f(S)|^{2}) \textrm{\ \ \ \ \ \ \ \ \ \ \ \ \ \ \ \ \   if } |R|=|S|= 1,\\
\sum_{i,j} g(R_i,S_j)\!+\!O(|f(R)\!+\!f(S)|^{\omega})  \textrm{\ \ \  otherwise.}
\end{cases}
\end{eqnarray*}
The solution for these equations gives an $O(\phi^{\omega})$ time bound.

\subsection{Multigraphs}
In the case of multigraphs we need to handle multiple copies of the same edge in a different way. Removing separate copies of an edge one by one would lead to a cubic time complexity. The cost charged by edge $ij$ to its submatrix would be $\Omega((f(i)+ f(j))^2\cdot \mu(ij))$. Instead, we use binary search on the number of removable copies of edges. Assume that we want to remove $\mu$ copies of edge $ij \in E(G)$ from the graph. We denote the resulting graph by $B(G-ij^\mu)$. In such case $B(G)$ and $B(G-ij^\mu)$ differ from one another by $2\mu$ updates of  rank one, i.e.,
\begin{equation}
\label{equation:rank-2-update-multi}
\begin{array}{l}
\displaystyle
B(G-ij^\mu) = B(G) +\sum_{k=1}^{\mu} -x^{ij,k} (y^{ij,k})^T + y^{ij,k} (x^{ij,k})^T \\
\displaystyle
= B(G)+ [-x^{ij,1},y^{ij,1},\ldots,-x^{ij,\mu},y^{ij,\mu}][y^{ij,1},x^{ij,1},\ldots,y^{ij,\mu},x^{ij,\mu}]^T.
\end{array}
\end{equation}
where $x^{ij,k}$ and $y^{ij,k}$ are length $\phi$ vectors defined in similar way as in the previous section. Observe that the submatrix affected by these
changes has size $2f(i)+2f(j)$ so, computing $B(G-ij^\mu)^{-1}$ from $B(G)^{-1}$ can be realized using~\cite[Corollary~2.1]{Harvey06}. Moreover, we can test non-singularity of $B(G-ij^\mu)$ by
\begin{equation}
\label{equation:test-removable-multi}
\det\left(I_{\mu}+[-x^{ij,1},\ldots,y^{ij,\mu}]^T B(G)^{-1}[y^{ij,1},\ldots,x^{ij,\mu}]\right) \neq 0.
\end{equation}
This test requires $O((f(i)+f(j))^{\omega})$ time, because the affected matrix size is $2f(i) + 2f(j)$.
In the following algorithm we use
this test together with a version of binary search.
\begin{algorithm}[H]
\caption{An $O(m \phi^2)$ time algorithm for finding $f$-factor in the  multigraph $G$.}
\begin{algorithmic}[1]
\State{Let $B(G)$ be $\phi \times \phi$ skew-symmetric adjacency matrix of $G$}
\State{Replace the variables in $B(G)$ for random elements from $\mathcal{Z}_p$ for prime $p=\Theta(\phi^2)$ to obtain $B$}
\State{(with probability $\ge 1-\frac{1}{\phi}$ matrix $B$ is non-singular when $B(G)$ is non-singular)}\comment{by
Lemma~\ref{corollary-zippel-schwartz}}
\State{Compute $B^{-1}$}
\State{(we remove only removable edges so $B$ remains non-singular during execution of the algorithm)}
\ForAll{$ij \in E$}
\State{Let $\mu$ be the highest power of $2$ not higher then $\min(\mu(ij),f(i),f(j))$.}
\State{Let $k:=0$.}
\While{$\mu\ge 1$}
\If{$\det\left(I_{\mu}+[-x^{ij,k+1},\ldots,y^{ij,k+\mu}]^T B^{-1}[y^{ij,k+1},\ldots,x^{ij,k+\mu}]\right) \neq 0$}
\State \comment{$\mu$ copies of $ij$ are removable by~\eqref{equation:test-removable-multi}}
\State{Set $E:=E-ij^{\mu}$}
\State{Set $B:=B+ [-x^{ij,k+1}\ldots,y^{ij,k+\mu}][y^{ij,k+1},\ldots,x^{ij,k+\mu}]^T$}
\State \comment{This corresponds to $B(G-ij^{\mu})$ by~\eqref{equation:rank-2-update}}
\State{Recompute $B^{-1}$}\comment{Using Sherman-Morrision-Woodbury formula}
\State{Set $k:=k+\mu$} \comment{The number of copies of $ij$ removed so far}
\EndIf
\State{Set $\mu:=\mu/2$}\comment{The number of edges we try to remove is halved}
\EndWhile
\EndFor
\State{Return $E$}\comment{All removable copies of edges have been removed, so what remains is an $f$-factor}
%\EndProcedure
\end{algorithmic}
\label{algorithm:f-factor-general-multi}
\end{algorithm}
This time we modify~\cite[Algorithm 1]{Harvey06} in the same way as given by Algorithm~\ref{algorithm:f-factor-general-multi}. The time for updates remains the same, because a submatrix of size $|S|\times|S|$ is still updated in $O(|S|^{\omega})$ time. On the other hand,
the cost we pay to find maximum number $\mu$ of removable copies of edge $e$ is $O((f(i)+f(j))^{\omega})$. The number of removed
edges in the binary search forms a geometric series, so the cost is dominated by the first element, which in turn is smaller
then the size of the submatrix to power of $\omega$. This time we obtain following bounds
\begin{eqnarray*}
h(S) &\!\!\!\!=\!\!\!\!&\sum_i h(S_i)+ g(S_1,S_2) +  O(|f(S)|^{\omega}) \\
g(R,S) &\!\!\!\!=\!\!\!\!&
\begin{cases}
O(|f(R)+f(S)|^{\omega}) \textrm{\ \ \ \ \ \ \ \ \ \ \ \ \ \ \ \ \  if } |R|=|S|= 1,\\
\sum_{i,j} g(R_i,S_j)+ O(|f(R)+f(S)|^{\omega})  \textrm{\ \ \  otherwise.}
\end{cases}
\end{eqnarray*}
The solution for these equations gives an $O(\phi^{\omega})$ time bound.

%% file: weights.tex
\section{Finding perturbed factor weights}
\label{sec:weights}
This section shows how to compute the quantities $w(F_v)$ and $w(F^v)$, for all $v\in V$. 
(Recall from Section \ref{BackgroundSec} that we are dealing with an $f$-critical graph;
each $v\in V$ has a maximum $f_v$-factor $F_v$ and  a maximum $f^v$-factor $F^v$.)
We start by considering
simple graphs, and then comment on multigraphs. For simplicity assume that the weight function is non-negative, i.e.,
$w:E\to [0..W]$. (If this is not the case, redefine $w(ij):=w(ij)+W$. This increases
the weight of each $f$-factor by exactly $W f(V)/2$.

%\paragraph*{Simple Graphs}
%Let $G=(V,E)$ be a simple graph, and let $w:E \to \mathbb {Z}$ be the edge weight function.
Following~\eqref{SimpleGenMatrixEqn}, define
\begin{equation}
B(G)_{i,r,j,c} =
\begin{cases}
\vspace{3pt}
z^{w(ij)} x^{ij}_r y^{ij}_{c}& %
\minibox{$ij\in E$  and  $(i,r,j,c)\in \mathbb P$,}\cr
\vspace{3pt}
-z^{w(ij)} x^{ij}_c y^{ij}_r&  %
\minibox{$ij \in E$ and  $(j,c, i,r)\in \mathbb P$,}\cr
0&\text{otherwise,}
\end{cases}
\label{SimpleGenMatrixWeightsEqn}
\end{equation}
where $z$ is a new indeterminate.  
For the next result we assume $G$ has an $f$-factor.
Theorem~\ref{GenSimpleDetThm} shows that there is a mapping $\F.$ from
terms of $\det(B(G))$ onto $\FF.$. The degree of $z$ in a term
$\sigma$ equals the total weight of the edges used. This gives
%\if\fullpaper0
% $\deg(\det(B))$ is equal to twice the weight of the maximum $f$-factor $F$ in $G$.
%\else
the following.
\begin{corollary}
\label{corollary-GenSimpleWeights}
For a simple graph $G$ that has an $f$-factor,  $\deg_z(\det(B(G)))$ is twice the weight of a maximum $f$-factor.
\end{corollary}
%\fi

%Let $G^+$ be the graph $G$ with an new vertex $s$ with self loop $ss$ added, and new edges $sv$, for all  $v \in V$.
%All these new edges have weight $0$. We extend $f$ by setting $f(s)=1$.
%%+++++
%One can observe that $G^+$ is critical.
%Let $F$ be an $f$-factor of $G$. For every vertex $v\in V$ and edge $vu \in F$ we construct $f_v$-factor
%by taking $F-vu+us$, whereas $f^v$-factor is constructed by taking $F+vs$. Finally, $F$ and $F+ss$ correspond
%to $f_s$-factor and $f^s$-factor respectively.
%%+++++

Now suppose $G$ is $f$-critical.
For any $v\in V$ let $G_v$ be $G$ with an additional vertex $t$ joined to $v \in V$ by a zero weight edge.
Set $f(t)=1$.
A maximum $f$-factor in $G_v$ weighs the same as a maximum $f_v$-factor in $G$.
For the computation of $f_v$-factors we need the following definition. Let $G_*$ be $G$ with an additional vertex $t$ that is connected to all vertices $v \in V$ with zero weight edges. As previously, we set $f(t)=1$. Let us denote by $F_*$ the maximum $f$-factor in $G_*$.

\iffalse
Let us define $G_v$ to be $G^+$ with an additional vertex $t$ that is connected to vertex $v \in V$ with zero weight edge.\footnote{For definition of
$G^+$ please see Section~\ref{BackgroundSec}.} For $t$ we set $f(t)=1$.
We observe that the weight of the maximum $f$-factor in $G_v$ is equal to the weight of the maximum $f_v$-factor in $G^+$.
For the computation of $f_v$-factors we need the following definition. Let $G_*$ be $G^+$ with an additional vertex $t$ that is connected to all vertices $v \in V-s$ with zero weight edges. As previously, we set $f(t)=1$. Let us denote by $F_*$ the maximum $f$-factor in $G_*$.
\fi

\begin{lemma}
\label{lemma:LowerPerturbation}
$\deg_z(\adj(B(G_*))_{v,0,t,0}) = w(F_*)+w(F_v).$
\end{lemma}
\begin{proof}
Observe that
\[
\adj(B(G_*))_{v,0,t,0} = (-1)^{n(t,0)+n(v,0)} \det(B(G_*)^{t,0,v,0}),
\]
By Theorem~\ref{GenSimpleDetThm} we know that $\det(B(G_*))$ contains terms corresponding to elements of $F_2(G_*)$.
Hence, by the above equality, terms of $\adj(B(G_*))_{v,0,t,0}$ correspond to elements $F_2(G_*)$ that use edge $tv$, but with this
edge removed. In other words terms of $\adj(B(G_*))_{v,0,t,0,}$ are obtained by pairing an $f$-factor in $G_*$ and an
$f$-factor in $G_v$, and removing the edge $tv$. Similarly, as we observed in Corollary~\ref{corollary-GenSimpleWeights}
the degree of $z$ encodes the total weight of elements of $F_2(G_*)$. Moreover, the maximum elements are constructed
by taking maximum $f$-factor of $G_*$ and maximum $f$-factor in $G_v$.
As we already observed maximum $f$-factor in $G_v$ is an maximum $f_v$-factor in $G^+$ so the theorem follows.
\end{proof}
%}
%\if\fullpaper1
%\proofLowerPerturbation
%\else
%The proof is given in the full version.
%%The proof is given in Appendix~\ref{appendix:LowerPerturbation}.
%\fi

The above theorem leads to Algorithm~\ref{algorithm:F_v-factor-weights-simple} that computes weights
of $F_v$, for all $v\in V$, in $\tilde{O}(Wn^{\omega})$ time.

\begin{algorithm}[H]
\caption{An $\tilde{O}(W \phi^{\omega})$ time algorithm for finding weights of $F_v$, for all $v\in V$, in a simple graph $G^+$.}
\begin{algorithmic}[1]
\State{Let $B(G_*)$ be $\phi \times \phi$ matrix representing $G_*$}
\State{Replace the variables in $B(G_*)$ for random elements from $\mathcal{Z}_p$ for prime $p=\Theta(\phi^3)$ to obtain $B$}
\State{Compute $d:=\det(B)$} \comment{requires $\tilde{O}(W\phi^{\omega})$ time using Theorem~\ref{theorem-storjohann}}
\State{($\deg_z(d) = \deg_z(\det(B(G_*)))$ with probability $\ge 1-\frac{1}{\phi^2}$)}\comment{by Lemma~\ref{corollary-zippel-schwartz}}
\State{Set $w(F_*) := \deg_z(d)/2$}\comment{by Corollary~\ref{corollary-GenSimpleWeights}}
\State{Compute $a := \adj(B) e_{t,0} = \det(B) B^{-1} e_{t,0}$}\comment{requires $\tilde{O}(W\phi^{\omega})$ time using Theorem~\ref{theorem-storjohann}}
\For{$v \in V$}
\State{($\deg_z(a_v) = \deg_z(\adj(B(G_*))_{v,0,t,0})$ with probability $\ge 1-\frac{1}{\phi^2}$)}\comment{by Lemma~\ref{corollary-zippel-schwartz}}
\State{Set $w(F_v) := \deg_z(a_v)- w(F_*)$}\comment{by Lemma~\ref{lemma:LowerPerturbation}}
\EndFor\comment{by union bound all $w(F_v)$ are correct with probability $\ge 1-\frac{1}{\phi}$}
%\EndProcedure
\end{algorithmic}
\label{algorithm:F_v-factor-weights-simple}
\end{algorithm}

For the computation of $f^v$-factors we need to proceed in slightly modified way. We construct $G^v$ from $G^+$
by:
\begin{itemize2}
\item adding new vertices $t_u$ and a zero weight edges $ut_u$, for every vertex $u\in V$,
\item adding a new vertex $t$ and zero weight edge $tt_v$.
\end{itemize2}
 Moreover, we define $f'(v)=f(v)+1$, $f'(t_v)=1$ and $f'(t)=1$. Finally, $G^*$ is obtained from $G^+$
in similar way, but with the difference that $t$ is connected to all vertices $t_v$.
Again we observe, that the weight of the maximum $f'$-factor in $G^v$ is equal to the weight of the maximum $f^v$-factor in $G^+$.
This allows us to prove the following.

\begin{lemma}
\label{lemma:UpperPerturbation}
$\deg_z(\adj(B(G^*))_{t_v,0,t,0}) = w(F^*)+w(F^v).$
\end{lemma}
\begin{proof}
By Theorem~\ref{GenSimpleDetThm} we know that $\det(B(G^*))$ contains terms corresponding to elements of $F_2(G^*)$.
Hence, terms of $\adj(B(G^*))_{t_v,0,t,0}$ correspond to elements $F_2(G^*)$ that use edge $tt_v$, but with this
edge removed. In other words terms of $\adj(B(G^*))_{t_v,0,v,0}$ are obtained by pairing an $f'$-factor in $G^*$ and an
$f'$-factor in $G^v$. As previously, the maximum weights elements are constructed by taking the
maximum $f'$-factor of $G^*$ and the maximum $f'$-factor in $G^v$. As observed above $f'$-factors
correspond to $f^v$-factors in $G^+$.
\end{proof}

This leads to $\tilde{O}(Wn^{\omega})$ time algorithm for computing weights of $F^v$, for all $v\in V$.

\begin{algorithm}[H]
\caption{An $\tilde{O}(W \phi^{\omega})$ time algorithm for finding weights of $F^v$, for all $v\in V$, for a simple graph $G^+$.}
\begin{algorithmic}[1]
\State{Let $B(G^*)$ be $\phi \times \phi$ matrix representing $G^*$}
\State{Replace the variables in $B(G^*)$ for random elements from $\mathcal{Z}_p$ for prime $p=\Theta(\phi^3)$ to obtain $B$}
\State{Compute $d:=\det(B)$} \comment{requires $\tilde{O}(W\phi^{\omega})$ time using Theorem~\ref{theorem-storjohann}}
\State{($\deg_z(d) = \deg_z(\det(B(G^*)))$ with probability $\ge 1-\frac{1}{\phi^2}$)}\comment{by Lemma~\ref{corollary-zippel-schwartz}}
\State{Set $w(F^*) := \deg_z(d)/2$}\comment{by Corollary~\ref{corollary-GenSimpleWeights}}
\State{Compute $a := \adj(B) e_{t,0} = \det(B) B^{-1} e_{t,0}$}\comment{requires $\tilde{O}(W\phi^{\omega})$ time using Theorem~\ref{theorem-storjohann}}
\For{$v \in V$}
\State{($\deg_z(a_{t_v}) = \deg_z(\adj(B(G^*))_{t_v,0,t,0})$ with probability $\ge 1-\frac{1}{\phi^2}$)}\comment{by Lemma~\ref{corollary-zippel-schwartz}}
\State{Set $w(F^v) := \deg_z(a_v)- w(F^*)$}\comment{by Lemma~\ref{lemma:UpperPerturbation}}
\EndFor\comment{by union bound all $w(F^v)$ are correct with probability $\ge 1-\frac{1}{\phi}$}
\end{algorithmic}
\label{algorithm:F^v-factor-weights-simple}
\end{algorithm}
%\else
%By simply exchanging lower perturbations with upper perturbations in Algorithm~\ref{algorithm:F_v-factor-weights-simple}
%we obtain an $O(Wn^{\omega})$ time algorithm for computing weights of $F^v$, for all $v\in V$.
%
%\fi

%\if\fullpaper1
\subsection{Multigraphs}
%\fi
%\def\WeightsMultigraphs{
Let $G=(V,E)$ be a multigraph, and let $w:E\times k \to \mathbb {Z}$ be the edge weight function.
This function assigns weight $w(e,k)$ to the $k$'th copy of $e\in E$. Joining ideas from~\eqref{SimpleGenMatrixWeightsEqn} and~\eqref{MultiGenMatrixEqn} we define
\begin{equation*}
%\label{SimpleGenMatrixWeightsEqn}
B(G)_{i,r,j,c} =
\begin{cases}
\vspace{3pt}
\sum_{k=1}^{\mu(ij)} z^{w(ij,k)} x^{ij,k}_{r} y^{ij,k}_{c}& %
\minibox{$ij\in E$ and  $(i,r,j,c) \in \mathbb P$,}\cr
\vspace{3pt}
-\sum_{k=1}^{\mu(ij)} z^{w(ij,k)} x^{ij,k}_{c} y^{ij,k}_{r}&  %
\minibox{$ij \in E$ and  $(j,c, i,r) \in \mathbb P$,}\cr
0&\text{otherwise,}
\end{cases}
\end{equation*}
where $z$ is a new indeterminate.  Theorem~\ref{GenMultiDetThm} shows that there is a mapping $\F.$ from
terms of $\det(B(G))$ onto $\F._2(G)$. Observe that the construction from previous section
requires only the existence of such mapping, so it can be used for multigraphs as well.
%}
%\if\fullpaper1
%\WeightsMultigraphs
%\else
%Finally, in order to handle multigraphs, we need to join ideas from~\eqref{SimpleGenMatrixWeightsEqn} and~\eqref{MultiGenMatrixEqn}. The
%definition of the symbolic matrix is given in the full version, whereas the rest follows exactly  as for simple graphs.
%\fi

%Now let $G$ be a bipartite multigraph. Let $\mu(e)$ denote the
%multiplicity of any edge $e$, and index the copies of $e$ as $e_k$, $k=1,\ldots,
%\mu(e)$. Denote by $w(e,k)$ the weight of the $k$'th copy of the edge.
%We include edge weights in $B$ in the following way
%
%\begin{equation}
%\label{MultiBipartiteMatrixEqn}
%b_{n(i)+r,\, n(j)+c} =
%\sum _{k=1}^{\mu(e)} z^{w(e,k)} x^{ek}_{ir} y^{ek}_{jc} \hskip10pt e=ij,
%\end{equation}
%
%Theorem~\ref{BipartiteMultiDetCor} shows that there is a mapping $\F.$ from
%terms of $\det(B)$ onto $f$-factors in $G$. We can observe that degree of $z$ in term
%$\sigma$ corresponds to the total weight of the $\F.(\sigma)$. As a consequence we have
%the following observation.
%
%\begin{corollary}
%The degree in $z$ of the smallest degree of $\det(B)$ is equal to the weight of the minimum weight $f$-factor in $G$.
%\end{corollary}
%
%We now show how to compute the weights of $f_v$ factors in $G^+$.

%% file: conclusions.tex
\section{Conclusions and open problems}
\label{ConclusionSec}
This paper presents new algebraic algorithms for the fundamental problems of $b$-matching, undirected single-source shortest paths, and $f$-factors. Some intriguing open problem and challenges emerge from this study:

$\bullet$ The matrices we construct for unweighted $f$-factors have a very special block structure. Can this
block structure be exploited to obtain faster algorithms, e.g., time $O(\phi^{\omega-1}n)$?

$\bullet$ Can the running time of our algebraic max-flow algorithm be improved, perhaps
by combining it with scaling techniques? Can scaling be used in the non-bipartite algorithms?

$\bullet$ We gave first algebraic algorithms for simple $2$-factors.
Are there algebraic formulations for triangle-free or square-free $2$-factors?
If so one expects the resulting algorithms to be simpler
than existing combinatoric ones.

$\bullet$ What is the complexity of all-pairs undirected shortest distances
on conservative graphs? Can $\tilde{O}(Wn^\omega)$ time 
be achieved, as in the
case of non-negative weights?

%% file: appendix-allowed.tex
\section{Allowed edges}
\label{appendix-allowed-edges}

\begin{figure}[H]
\centering
  \includegraphics[width=0.5\textwidth]{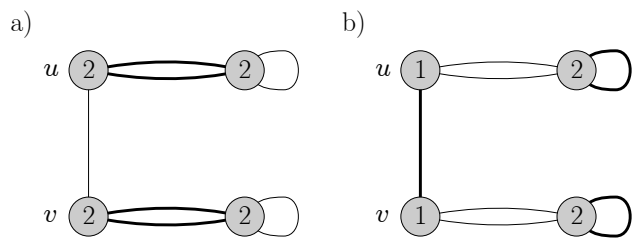}
\caption{The only $2$-factor is shown on panel a), whereas the only $f_{u,v}$-factor is shown on panel b). Edge $uv$ is not allowed
in any $2$-factor, although $f_{u,v}$-factor does exist.}
\label{figure:f_i_j-factor}
\end{figure}

Observe that for any edge $ij$ as long as $f(i)=f(j)=1$ the existence of $f_{i,j}$-factor is equivalent to the fact that $ij$ is allowed. However,
as shown on Figure~\ref{figure:f_i_j-factor}, the edge $uv$ is not allowed, although the $f_{u,v}$-factor does exit. The reason for this is that
we are trying to use edge $uv$ twice. Hence, similar criteria as in Corollary~\ref{corollary:allowed-edge} does not hold in non-bipartite case.

%The right criteria that follows directly from Lemma~\ref{lemma:f_i_j-factor} is given by the next corollary.

%\begin{corollary}
%Let $G$ be a bipartite graph having an $f$-factor. The edge $ij \in E(G)$ is allowed if and only if $(B(G-ij)^{-1})_{n(j),n(i)}\neq 0$.
%\end{corollary}
%
%The above corollary shows that in order to find the $f$-factor, we would need to combine Gaussian elimination with matrix updates. 

%% file: det.bbl
\begin{thebibliography}{10}

\bibitem{Ahuja93}
R.~K. Ahuja, T.~L. Magnanti, and J.~B. Orlin.
\newblock {\em Network Flows: Theory, Algorithms, and Applications}.
\newblock Prentice Hall, Englewood Cliffs, NJ, 1993.

\bibitem{anstee}
R.~Anstee.
\newblock A polynomial algorithm for b-matching: An alternative approach.
\newblock {\em IPL}, 24:153--157, 1987.

\bibitem{hb}
J.~Bunch and J.~Hopcroft.
\newblock Triangular factorization and inversion by fast matrix multiplication.
\newblock {\em Mathematics of Computation}, 28(125):231--236, 1974.

\bibitem{Cheriyan90}
J.~Cheriyan, T.~Hagerup, and K.~Mehlhorn.
\newblock Can a maximum flow be computed in o(nm) time?
\newblock In {\em IN PROC. ICALP}, pages 235--248. Springer-Verlag, 1990.

\bibitem{cheung}
H.~Y. Cheung, L.~C. Lau, and K.~M. Leung.
\newblock Graph connectivities, network coding, and expander graphs.
\newblock In {\em Proc. of FOCS'11}, pages 190--199, 2011.

\bibitem{CLRS}
T.~H. Cormen, C.~E. Leiserson, R.~L. Rivest, and C.~Stein.
\newblock {\em Introduction to Algorithms}.
\newblock McGraw-Hill, New York, 2nd edition, 2001.

\bibitem{CGSa}
M.~Cygan, H.~N. Gabow, and P.~Sankowski.
\newblock Algorithmic applications of {B}aur-{S}trassen's theorem: shortest
  cycles, diameter and matchings.
\newblock In {\em Proc. of FOCS'12}, pages 531--540, 2012.

\bibitem{e}
J.~Edmonds.
\newblock Maximum matching and a polyhedron with 0,1-vertices.
\newblock {\em Journal of Research National Bureau of Standards-B},
  69B:125--130, 1965.

\bibitem{Edmonds67}
J.~Edmonds.
\newblock {An introduction to matching. Mimeographed notes, Engineering Summer
  Conference, U. Michigan, Ann Arbor, MI}, 1967.

\bibitem{edmonds-karp-72}
J.~Edmonds and R.~M. Karp.
\newblock Theoretical improvements in algorithmic efficiency for network flow
  problems.
\newblock {\em Journal of the ACM}, 19(2):248--264, 1972.

\bibitem{gabow-76}
H.~N. Gabow.
\newblock An efficient implementation of {E}dmonds' algorithm for maximum
  matching on graphs.
\newblock {\em J. ACM}, 23(2):221--234, 1976.

\bibitem{hal83}
H.~N. Gabow.
\newblock An efficient reduction technique for degree-constrained subgraph and
  bidirected network flow problems.
\newblock In {\em Proc. of STOC'83}, pages 448--456, 1983.

\bibitem{gabow-90}
H.~N. Gabow.
\newblock Data structures for weighted matching and nearest common ancestors
  with linking.
\newblock In {\em Proc. of SODA'90}, pages 434--443, 1990.

\bibitem{G}
H.~N. Gabow.
\newblock A combinatoric interpretation of dual variables for weighted matching
  and $f$-factors.
\newblock {\em Theoretical Computer Science}, 454:136--163, 2012.

\bibitem{gabow-tarjan-89}
H.~N. Gabow and R.~E. Tarjan.
\newblock Faster scaling algorithms for network problems.
\newblock {\em SIAM Journal on Computing}, 18(5):1013--1036, 1989.

\bibitem{Goldberg98}
A.~V. Goldberg and S.~Rao.
\newblock Beyond the flow decomposition barrier.
\newblock {\em J. ACM}, 45(5):783--797, Sept. 1998.

\bibitem{Harvey06}
N.~J.~A. Harvey.
\newblock Algebraic algorithms for matching and matroid problems.
\newblock {\em SIAM SIAM J. Comput.}, 2(39):679--702, 2009.

\bibitem{l}
L.~Lov\'asz.
\newblock On determinants, matchings and random algorithms.
\newblock In L.~Budach, editor, {\em Fundamentals of Computation Theory}, pages
  565--574. Akademie-Verlag, 1979.

\bibitem{lp}
L.~Lov\'asz and M.~D. Plummer.
\newblock {\em Matching Theory}.
\newblock Akad\'emiai Kiad\'o, 1986.

\bibitem{Marsh}
A.~B. Marsh.
\newblock {\em Matching algorithms}.
\newblock PhD thesis, The John Hopkins Univeristy, Baltimore, 1979.

\bibitem{ms04}
M.~Mucha and P.~Sankowski.
\newblock Maximum matchings via {G}aussian elimination.
\newblock In {\em Proc. of FOCS'04}, pages 248--255, 2004.

\bibitem{Orlin88}
J.~B. Orlin.
\newblock A faster strongly polynominal minimum cost flow algorithm.
\newblock In {\em Prof. of STOC'88}, pages 377--387, 1988.

\bibitem{Pulleyblank}
W.~Pulleyblank.
\newblock {\em Faces of matching polyhedra}.
\newblock PhD thesis, University of Waterloo, Ontario, Canada, 1973.

\bibitem{rv89}
M.~O. Rabin and V.~V. Vazirani.
\newblock Maximum matchings in general graphs through randomization.
\newblock {\em Journal of Algorithms}, 10:557--567, 1989.

\bibitem{Sankowski05}
P.~Sankowski.
\newblock Shortest paths in matrix multiplication time.
\newblock In {\em Proc. of ESA'05}, pages 770--778, 2005.

\bibitem{Sankowski09}
P.~Sankowski.
\newblock Maximum weight bipartite matching in matrix multiplication time.
\newblock {\em Theoretical Computer Science}, 410(44):4480--4488, 2009.

\bibitem{Schrijver}
A.~Schrijver.
\newblock {\em Combinatorial Optimization - Polyhedra and Efficiency}.
\newblock Springer-Verlag, 2003.

\bibitem{s80}
J.~T. Schwartz.
\newblock Fast probabilistic algorithms for verification of polynomial
  identities.
\newblock {\em J. ACM}, 27:701--717, 1980.

\bibitem{Sebo}
A.~Seb\"o.
\newblock Undirected distances and the postman-structure of graphs.
\newblock {\em J. Combin. Theory Ser. B}, 49(1):10 -- 39, 1990.

\bibitem{sebo2}
A.~Seb\"o.
\newblock Potentials in undirected graphs and planar multiflows.
\newblock {\em SIAM J. Comput.}, 26(2):582--603, 1997.

\bibitem{storjohann03}
A.~Storjohann.
\newblock High-order lifting and integrality certi- fication.
\newblock {\em J. Symbolic Comput.}, 36(3-4):613--648, 2003.

\bibitem{urquhart}
R.~Urquhart.
\newblock {\em Degree-constrained subgraphs of linear graphs}.
\newblock PhD thesis, University of Michigan, 1967.

\bibitem{Williams}
V.~V. Williams.
\newblock {Multiplying matrices faster than Coppersmith-Winograd}.
\newblock In {\em Proc. STOC'12}, pages 887--898, 2012.

\bibitem{YusterZ05}
R.~Yuster and U.~Zwick.
\newblock Answering distance queries in directed graphs using fast matrix
  multiplication.
\newblock In {\em Proc. of FOCS'05}, pages 389--396, 2005.

\bibitem{z79}
R.~Zippel.
\newblock Probabilistic algorithms for sparse polyno- mials.
\newblock In {\em Proc. of EUROSAM'79}, pages 216--226, 1979.

\end{thebibliography}
